\documentclass[12pt]{article}
\usepackage{lmodern}
\usepackage{amssymb,amsmath,amsthm}
\usepackage{bbm}
\usepackage{ifxetex,ifluatex}
\usepackage{adjustbox} 
\usepackage{mathtools}
\usepackage{threeparttable}
\usepackage{fixltx2e} 
\ifnum 0\ifxetex 1\fi\ifluatex 1\fi=0 
  \usepackage[T1]{fontenc}
  \usepackage[utf8]{inputenc}
\else 
  \ifxetex
    \usepackage{mathspec}
  \else
    \usepackage{fontspec}
  \fi
  \defaultfontfeatures{Ligatures=TeX,Scale=MatchLowercase}
\fi
\IfFileExists{upquote.sty}{\usepackage{upquote}}{}
\IfFileExists{microtype.sty}{%
\usepackage[]{microtype}
\UseMicrotypeSet[protrusion]{basicmath} 
}{}
\usepackage[colorlinks = true,
            linkcolor = blue,
            urlcolor  = blue,
            citecolor = blue,
            anchorcolor = blue]{hyperref}
\PassOptionsToPackage{hyphens}{url} 
\hypersetup{
            pdftitle={Group Interaction Experiments},
            pdfauthor={Jiawei Fu},
            pdfborder={0 0 0},
            breaklinks=true}
\usepackage[margin=1in]{geometry}
\usepackage{graphicx,grffile} 
\makeatletter
\def\maxwidth{\ifdim\Gin@nat@width>\linewidth\linewidth\else\Gin@nat@width\fi}
\def\maxheight{\ifdim\Gin@nat@height>\textheight\textheight\else\Gin@nat@height\fi}
\makeatother
\setkeys{Gin}{width=\maxwidth,height=\maxheight,keepaspectratio}
\ifx\paragraph\undefined\else
\let\oldparagraph\paragraph
\renewcommand{\paragraph}[1]{\oldparagraph{#1}\mbox{}}
\fi
\ifx\subparagraph\undefined\else
\let\oldsubparagraph\subparagraph
\renewcommand{\subparagraph}[1]{\oldsubparagraph{#1}\mbox{}}
\fi

\newcommand\independent{\protect\mathpalette{\protect\independenT}{\perp}}
\def\independenT#1#2{\mathrel{\rlap{$#1#2$}\mkern2mu{#1#2}}}
\def\Var{\text{Var}\,}
\def\E{\text{E}\,}

\def\Cov{\text{Cov}\,}

\def\N{\text{N}\,}

\renewcommand{\a}{\mathbf{a}}

\makeatletter
\def\fps@figure{htbp}
\makeatother

\usepackage{setspace}
\usepackage{ragged2e,subcaption,caption}
\usepackage[linesnumbered,ruled,vlined]{algorithm2e}
\usepackage{longtable,booktabs}

\usepackage{tikz}
\usetikzlibrary{positioning,shapes.geometric,chains,calc,shapes}

\newtheorem{theorem}{Theorem}
\newtheorem{lemma}{Lemma}
\newtheorem{prop}{Proposition}
\newtheorem{corollary}{Corollary}
\newtheorem{assn}{Assumption}

\usepackage{natbib}

\usepackage{comment}

\title{Inference for Group Interaction Experiments\thanks{Jiawei Fu is Assistant Professor of Political Science, Duke University (Email: \url{jiawei.fu@duke.edu});
 Cyrus Samii is Professor of Politics, New York University (Email: \url{cds2083@nyu.edu}); 
 Ye Wang is Assistant Professor of Political Science, University of North Carolina at Chapel Hill (Email: \url{yewang@unc.edu}). 
 PM Aronow offered crucial suggestions for the analysis presented in this paper. 
 For helpful comments, we thank Peng Ding, Donald Green, Christopher Harshaw, Lihua Lei, Xinran Li, Fredrik S\"avje, Matthew Tyler, Yiqing Xu, Anqi Zhao, and seminar participants at EPSA, ACIC, APSA, Columbia University, Stanford University, University of California-Berkeley, and the Isaac Newton Institute. An AI usage disclosure appears at the end of the main text.}}
\author{Jiawei Fu \and Cyrus Samii \and Ye Wang}
\providecommand{\institute}[1]{}
\date{\today}
\begin{document}
\maketitle


\medskip
\begin{abstract}

A common experimental research design is one in which individuals are randomly allocated into groups that then interact under different group-level treatment conditions. We develop design-based inference for such ``group interaction'' experiments, covering scenarios in which groups are either fixed or randomly formed and in which potential outcomes are either fixed relative to others' group assignments or subject to interference. For each scenario, we characterize the causal estimand that the design targets and the inferential strategy appropriate to it. Working in a sparse-sampling asymptotic regime, we show that cluster-robust inference remains consistent and accounts for dependencies from various sources when interference is present, delivering valid inference on marginalized exposure effects. When interference is absent and groups are formed randomly, the design reduces to an individually randomized experiment, and individual-level heteroskedasticity-robust inference suffices for the average treatment effect. Our results on the asymptotic distribution of commonly used estimators rely on a novel coupling strategy that may be useful for design-based inference in other complex experiments.
\end{abstract}
\textbf{Keywords:} \emph{Causal Inference, Interference, Marginalized Exposure Effects, Group Interaction Experiments, Design-Based Inference, Cluster-Robust Inference}

\clearpage
\doublespacing

\clearpage



\section{Introduction}
A common experimental design in social science has individuals assigned into groups where they interact with their group mates under different group-level treatment assignments. We call these ``group interaction'' experiments.\footnote{\citet{lohr2014partially} use ``partially nested'' design and \citet{pals2008individually} use ``individually randomized group treatment trials'' for the same design, although these phrasings does not evoke group interaction as explicitly.} 
Prominent examples include \citet{mendelberg2014does}, who assigned individuals to deliberation groups varying by gender composition and decision rule, and \citet{iacovone2022improving}, who assigned business owners to individual or group-based consulting.

Despite similar designs, analytical approaches often differ. \citet{mendelberg2014does} used cluster-robust standard errors, appealing to the intuition that group interaction creates dependence, as others have on similar grounds \citep{pals2008individually, candel2025efficient}, sometimes via multilevel models \citep{lohr2014partially}. \citet{iacovone2022improving}, by contrast, relied on individual-level randomization to justify inference without clustering \citep{abadie2023should}. This discrepancy reflects a broader confusion---spanning education, psychology, and economics---about whether individual assignment to groups warrants clustering when treatments involve interaction.\footnote{A recent Bluesky exchange with prominent applied microeconomists exemplifies issues that this paper tries to clarify: \url{https://bsky.app/profile/seema.bsky.social/post/3lbi3yzytfk2g}}

To our knowledge, this paper provides the first design-based analysis of inference for group interaction experiments. Group interaction implies the potential for {\it interference} such that an individual's potential outcomes may depend on {\it how others are assigned} to groups and treatments \citep{aronow2021spillover}. This represents a departure from the standard ``stable unit treatment value assumption'' (SUTVA) upon which standard inferential results are based \citep{imbens2015causal}.
To build up insight on the factors that affect inference in such experiments, we consider scenarios in which interference is present or not, and when groups are randomly assigned or fixed. 
We define valid estimands in each of these scenarios, using the concept of average marginalized exposure effects when interference is present. Marginalized exposure effects are a now-common concept in the literature on causal inference with interference \citep{aronow2021spillover, hudgens2008toward, li2019randomization, hu2022average, savje2021average}. They refer to effects defined in terms of ``averages of averages''---that is, the population average of individual-level averages over different potential exposures. 

Next, we derive convergence results for commonly used treatment effect estimators---namely, difference-in-means, difference in inverse probability weighted (IPW) means, and covariate-adjusted treatment effects regression---in relation to the estimands in each scenario. These results rest on a sparse-sampling asymptotic regime, in which the group sizes remain fixed and the individual-level sampling proportion goes to zero at a fast rate. Generalization to covariate-controlled regression estimators and generalized linear model predictions follows naturally. This regime operates similarly to the superpopulation regime of \citet{bai2022optimality} and \citet{bai2024primer} for analyzing experimental designs under SUTVA.

Building up to analyzing the implications of interference and random group formation, we first consider {\it fixed} groups (rather than randomly constructed ones) in the {\it absence} of interference. In this case, the results are essentially identical to what \citet{abadie2023should} and \citet{su2021model} have derived: the treatment effect estimators given above and cluster-robust standard errors are consistent for inference on the average treatment effect (ATE) in cluster-randomized trials. When we instead introduce {\it random group formation} while maintaining no interference, the distribution of the treatment effect estimators match what one obtains from an individual-level randomized trial targeting the ATE. This case corresponds to the intuition, invoked by some authors \citep{li2019randomization, iacovone2022improving}, that the data can be analyzed on the basis of ``individual random assignment.'' The two cases differ in their variance: it is larger under fixed groups than under random group formation and, in the former case, increases with the group-level intraclass correlation induced by outcome homophily.

However, interference changes things. With {\it fixed} groups and interference, the treatment effect estimators are consistent for an {\it average marginalized exposure effect}, analogous to the ``average total effect'' of \citet{hudgens2008toward} under full group-level treatment saturation. Interference adds a second source of intraclass correlation, on top of any already induced by outcome homophily in the pre-existing groups. Because the cluster-robust standard error captures within-group dependence whatever its source, it remains consistent for inference on this effect under our sparse-sampling regime. Finally, with random allocation into groups, the homophily-induced intraclass correlation is absent by design, but the one generated by interference still presents. Treatment effect estimators now target an average marginalized exposure effect, averaging over potential outcomes under different groupings of individuals. The cluster-robust standard error is consistent for inference on this estimand, while the heteroskedasticity-robust one is not. A key take-away from our analysis is that, when interference is present, the advice to ``cluster at the level of randomization'' can be misleading. Rather, what is required is to cluster at the level of exposure variation, and this depends on the full list of potential outcome inputs. 

We further derive the asymptotic distribution of the treatment effect estimators under random group formation, both with and without interference, through a new coupling argument \citep{hajek1960limiting, polyanskiy2025information}. We view random group formation as sampling from the collection of all possible groups that could be formed from the population, subject to the restriction that overlapping groups cannot be sampled together. We then show that this restricted sampling process can be approximated by one that draws groups independently, so that a standard central limit theorem applies at the group level \citep{ohlsson1989asymptotic}. This approach may be of independent interest for design-based inference in other complex designs. Taken together, our analysis unifies the four design--outcome cases within a single framework and provides valid inference for each.

The sections that follow explicate these results formally. Section~\ref{sec:setting} lays out the inferential setting, defining the design and outcome cases. Section~\ref{sec:est} introduces the estimands, along with the treatment effect estimators and their variance estimators. Section~\ref{sec:inference} presents our main result, characterizing the limiting distribution of these estimators and the validity of each variance estimator across the design and outcome cases, as well as a test for interference based on the within-group intraclass correlation. Section~\ref{sec:simulation} presents evidence from a simulation study to validate the results. Finally, Section~\ref{sec:app} applies these methods to the \citet{mendelberg2014does} and \citet{iacovone2022improving} experiments, illustrating their practical importance, and Section~\ref{sec:con} concludes.

\section{Inferential settings}\label{sec:setting}
When treatments are administered in group settings, units may interact within their groups. Such interaction means that a unit's potential outcomes may depend on who else is assigned to its group---a form of interference \citep{cox1958planning, sobel2006randomized}. Group formation may lie outside the experimenter's control and be homophily-based, in which case potential outcomes typically exhibit within-group (intraclass) correlation even before any interaction; interaction can then add further correlation after treatment. Alternatively, the experiment may assign group membership at random. This latter case is our main focus, but we also treat (i) fixed groups and (ii) settings without interference, to clarify the distinct implications of random group formation, interference, and their combination. Throughout the paper, we use upper-case letters for random variables, lower-case letters for their realizations, and calligraphic letters for sets. 

\subsection{Experimental designs}
Let $\mathcal{U}$ be a reference population of $n = |\mathcal{U}|$ units. We consider two research designs, distinguished by whether groups are fixed or randomly formed.
\begin{description}
\item[Design Case 0: Fixed groups.] Units in the population belong to pre-existing groups. Using simple random sampling without replacement, we draw $G_N$ groups in sequence from $\mathcal{U}$; call this sample $\mathcal{N}$, with $|\mathcal{N}| = N$. We index the sampled groups $g=1,...,G_N$, assigning $g=1,...,G_{N1}$ to ``treatment'' and $g=G_{N1}+1,...,G_N$ to ``control.'' Treatment status is captured by a group treatment indicator $Z_g \in \{0,1\}$, equal to $1$ for treated groups and $0$ for control groups. Each group $g$ has size $M_g$, so the sample contains $N \equiv \sum_{g=1}^{G_N} M_g$ units, of which $N_1$ are under treatment and $N_0$ under control.

\item[Design Case 1: Random group formation.] The design is intended to construct $G_N$ randomly formed groups, indexed by $g=1,...,G_N$. Each group $g$ has a pre-specified target size $M_g \in \mathcal{M} = \{m_1,...,m_K\}$, where $\mathcal{M}$ is the set of $K$ distinct admissible group sizes and $m_k$ is the $k$-th such size. The simplest case is $K=1$ with $M_g = M$ for all $g$ (homogeneous group sizes).\footnote{With a slight abuse of notation, we also use $M$ to denote the single admissible group size in the homogeneous case.} Using simple random sampling without replacement, we draw $N \equiv \sum_{g=1}^{G_N} M_g$ {\it individual} units from $\mathcal{U}$, call this sample $\mathcal{N}$, and randomly partition them into the $G_N$ groups. The partition is created through a random permutation of an index vector of length $N$ of the form $\underbrace{1,...,1}_{M_1},...,\underbrace{g,...,g}_{M_g},...,\underbrace{G_N,...,G_N}_{M_{G_N}}$, in which each index value $g$ appears $M_g$ times, so that each unit is randomly assigned to a group. Groups $g=1,...,G_{N1}$ are then assigned to ``treatment'' and $g=G_{N1}+1,...,G_N$ to ``control,'' with group treatment indicator $Z_g \in \{0,1\}$ equal to $1$ for treated groups and $0$ for control groups. We write $G_N(m)$ for the number of groups with $M_g = m \in \mathcal{M}$, $G_{N1}(m)$ for the number of these assigned to treatment, and $G_{N0}(m) = G_N(m) - G_{N1}(m)$ for the number assigned to control. Because group sizes can differ across arms, we let $\mathcal{M}_1$ and $\mathcal{M}_0$ denote the distinct group sizes appearing in the treatment and control arms, respectively; under homogeneous group sizes, both reduce to the same singleton. $N_1$ and $N_0$ denote the numbers of units under treatment and control, respectively, while $N_z(m)$ denotes the number of units in size-$m$ groups assigned to treatment status $z$.
\end{description}

Given random sampling, Design Case 0 is equivalent in distribution (by exchangeability of the sampling sequence) to a design in which we first randomly sample $G_N$ groups and then randomly assign them to treatment or control. The same would hold for Design Case 1 (by exchangeability of the sampling and permutation orders) if group sizes were homogeneous ($M_g = M$ for all $g$); but our description allows group sizes to vary, as in both of the motivating examples from \citet{mendelberg2014does} and \citet{iacovone2022improving}. The treatment and control arms may contain different group sizes, as in \citet{iacovone2022improving}, where control units are left ungrouped, each forming a singleton ($\mathcal{M}_0 = \{1\}$).

In Design Case 1, the number of groups of each size, $G_N(m)$, is itself a design choice that researchers can set for different purposes---for instance, to equalize the number of groups across sizes ($G_N(m)$ constant), or to equalize the number of units across sizes ($G_N(m) \propto 1/m$). We have described complete randomization of groups to treatment and control for illustration; the framework also accommodates more sophisticated assignment mechanisms, such as randomizing treatment separately within each group size (i.e., blocking on group size), which the varying-size results below handle through size-specific treatment probabilities.

For both designs, let $\mathcal{A}_g$ denote the set of units in group $g$, with $|\mathcal{A}_g| = M_g$. For any unit $i$ in the sample, we let $A_i$ denote the index of the group to which unit $i$ belongs, so that $A_i \in \{1,...,G_N\}$ and $A_i = g$ if and only if $i \in \mathcal{A}_g$, or equivalently,
$$
A_i = \sum_{g=1}^{G_N} g \cdot \mathbf{1}\{i \in \mathcal{A}_g\}.
$$
The variable $A_i$ thus captures the random group assignment induced by the sampling and, in Design Case 1, the partitioning. When groups are fixed (Design Case 0), $\mathcal{A}_g$ is the pre-existing group that receives index $g$ based on the order in which groups were sampled; when groups are formed randomly (Design Case 1), $\mathcal{A}_g$ is the set of units assigned to position $g$ by the partitioning process. Unit $i$'s treatment status is that of its group, $Z_{A_i} \in \{0,1\}$.

\subsection{Potential outcomes}
For each $i \in \mathcal{U}$ define the potential outcome $Y_i\left(z, \mathcal{A}\right) \in \mathbb{R}$, where $z \in \{0,1\}$ is the treatment status of $i$'s group and $\mathcal{A} \subseteq \mathcal{U}$ with $i \in \mathcal{A}$ is the group to which $i$ belongs. This is a general definition of potential outcomes that allows for arbitrary heterogeneity on the basis of treatment and group membership. Below, we will impose restrictions on the potential outcomes to define the cases with and without interference. First, we assume the following high level assumptions:
\begin{assn}\label{assn:po}
For all $i$ in $\mathcal{U}$ and all assignment conditions, we have 
\begin{enumerate}
\item (bounded potential outcomes) $|Y_i(z, \mathcal{A})| < C$ for all $z \in \{0,1\}$ and $\mathcal{A} \ni i$, where $C\in \mathbb{R_+}$, and
\item (lower bound on variance) The distribution of $Y_i(z, \mathcal{A})$ in $\mathcal{U}$ is non-degenerate such that sample variances of potential outcomes are bounded away from zero.
\end{enumerate}
\end{assn}
Bounded potential outcomes ensures fast convergence rates of sample statistics. The non-degeneracy condition ensures that true sampling and randomization variances are bounded away from zero.  

We consider two types of assumptions on interference, which have implications for the potential outcomes and the observed outcomes:
\begin{description}
    \item[Outcome Case 0: No interference.] For every $i \in \mathcal{U}$, potential outcomes depend only on the realized group-level treatment assignment for unit $i$; that is, $Y_i(z, \mathcal{A}) = Y_i(z)$ for all $\mathcal{A} \ni i$, and for each individual in the sample, we observe the outcome $Y_i = Y_i(Z_{A_i})$.
    \item[Outcome Case 1: Group interference.] For every $i \in \mathcal{U}$, potential outcomes depend on both the group-level treatment assignment and which other individuals are assigned to the same group as unit $i$; that is, $Y_i(z, \mathcal{A})$ may vary across $\mathcal{A} \ni i$, and for each individual in the sample, we observe the outcome $Y_i = Y_i(Z_{A_i}, \mathcal{A}_{A_i})$. 
\end{description}

The dependence of $Y_i(z, \mathcal{A})$ on $\mathcal{A}$ lets unit $i$'s potential outcomes vary with which other units share its group---a form of interference \citep{cox1958planning, sobel2006randomized} restricted to within groups. This restriction is similar to partial interference of \citet{hudgens2008toward}, but the source of interference differs: Hudgens and Halloran fix the group structure and vary the within-group treatment allocation, whereas we fix treatment within groups and let group composition vary randomly, allowing for group interaction, outcome-based contagion, or other composition-driven heterogeneity. Table~\ref{tab:case-table} summarizes the four cases that arise under the different design and outcome assumptions, which our analysis treats in turn.


\begin{table}[]
    \centering
\begin{center}
\begin{tabular}{cc|l|l|}
& & \multicolumn{2}{c|}{Random Group Formation}\\
& & No & Yes\\
\hline
Interference & No & Design Case 0 & Design Case 1 \\
& & Outcome Case 0 & Outcome Case 0\\
\hline
& Yes & Design Case 0 & Design Case 1 \\
&  & Outcome Case 1 & Outcome Case 1 \\
\hline
\end{tabular}
\end{center}
    \caption{\it Cases that emerge under the different design and outcome assumptions.}
    \label{tab:case-table}
\end{table}

\subsection{Illustrations}
Figure~\ref{fig:toy-design} provides a toy illustration of Design Case 1 with $N=4$ sampled units, $G_N = 2$ groups, and $M_1 = M_0 = 2$ units per group. Panel~A shows how the group index vector $(1,1,2,2)$ is randomly permuted across the four unit positions, yielding $\binom{4}{2}=6$ equally likely partitions. Panel~B shows, for unit $i=1$, the potential outcome inputs under each permutation: without interference, the potential outcome depends only on the group-level treatment status $z_{a_1}$; with group interference, it also depends on the identity of unit 1's groupmate. The combination of group-level treatment status and groupmate assignments constitutes the ``exposures'' that the design generates \citep{aronow2017estimating}.


\begin{figure}[!t]
\centering

\begin{tabular}{cl}
\multicolumn{2}{c}{\bf Panel A: Random permutation of group index vector}\\
\hline
& Index vector to permute: $(1,\; 1,\; 2,\; 2)$ \\
& Unit labels in fixed order: $i=1,2,3,4$ \\
\hline
& \\[-6pt]
& All distinct permutations $\a = (a_1, a_2, a_3, a_4)$: \\[4pt]
$\a^{(1)}$: & $(1,\; 1,\; 2,\; 2)$ \quad $\Rightarrow$ \quad Group 1: $\{1,2\}$, \; Group 2: $\{3,4\}$ \\[2pt]
$\a^{(2)}$: & $(1,\; 2,\; 1,\; 2)$ \quad $\Rightarrow$ \quad Group 1: $\{1,3\}$, \; Group 2: $\{2,4\}$ \\[2pt]
$\a^{(3)}$: & $(1,\; 2,\; 2,\; 1)$ \quad $\Rightarrow$ \quad Group 1: $\{1,4\}$, \; Group 2: $\{2,3\}$ \\[2pt]
$\a^{(4)}$: & $(2,\; 1,\; 1,\; 2)$ \quad $\Rightarrow$ \quad Group 1: $\{2,3\}$, \; Group 2: $\{1,4\}$ \\[2pt]
$\a^{(5)}$: & $(2,\; 1,\; 2,\; 1)$ \quad $\Rightarrow$ \quad Group 1: $\{2,4\}$, \; Group 2: $\{1,3\}$ \\[2pt]
$\a^{(6)}$: & $(2,\; 2,\; 1,\; 1)$ \quad $\Rightarrow$ \quad Group 1: $\{3,4\}$, \; Group 2: $\{1,2\}$ \\[2pt]
\hline
\end{tabular}

\vspace{12pt}

\begin{tabular}{ccccc}
\multicolumn{5}{c}{\bf Panel B: Potential outcome inputs for unit $i=1$}\\
\hline
 & & & \multicolumn{2}{c}{Potential outcome} \\
\cline{4-5}
Permutation & Treatment $(z_{a_1})$ & $\mathcal{A}_{A_1}$ & No interference & Group interference \\
\hline
$\a^{(1)}$ & $1$ & $\{1,2\}$ & $Y_1(1)$ & $Y_1(1,\{1,2\})$ \\[2pt]
$\a^{(2)}$ & $1$ & $\{1,3\}$ & $Y_1(1)$ & $Y_1(1,\{1,3\})$ \\[2pt]
$\a^{(3)}$ & $1$ & $\{1,4\}$ & $Y_1(1)$ & $Y_1(1,\{1,4\})$ \\[2pt]
$\a^{(4)}$ & $0$ & $\{1,4\}$ & $Y_1(0)$ & $Y_1(0,\{1,4\})$ \\[2pt]
$\a^{(5)}$ & $0$ & $\{1,3\}$ & $Y_1(0)$ & $Y_1(0,\{1,3\})$ \\[2pt]
$\a^{(6)}$ & $0$ & $\{1,2\}$ & $Y_1(0)$ & $Y_1(0,\{1,2\})$ \\[2pt]
\hline
\end{tabular}

\caption{\it Illustration of Design Case 1 with $G_N = 2$ groups, $M_1 = M_0 = 2$ units per group, and $N=4$ sampled units. Group 1 is treated ($Z_1=1$) and Group 2 is control ($Z_2 = 0$). Panel~A shows how the group index vector $(1,1,2,2)$ is randomly permuted across the four unit positions, yielding $\binom{4}{2}=6$ equally likely partitions. $\a^{(k)}$ represents the $k$-th permutation. Panel~B shows the resulting potential outcome inputs for unit $i=1$ under each permutation: without interference (Outcome Case 0), the potential outcome depends only on treatment status $z_{a_1}$; with group interference (Outcome Case 1), it also depends on the group set $\mathcal{A}_{A_1}$.}
\label{fig:toy-design}
\end{figure}

Table~\ref{table:karp-design1} shows the experimental design table from \citet{mendelberg2014does}, which illustrates several features that our setup covers. Because the outcomes of interest concern women's deliberative behavior, we take the population of analysis to be all women and view the design as one that partitions these women into groups of varying size and then assigns each group a treatment---an instance of random group formation (Design Case 1). The group set $\mathcal{A}$ for a woman is the set of women in her deliberation group, and its size is the number of women that group contains. Since a five-person group may hold anywhere from zero to five women, there are six possible sizes; the all-male groups contain no member of the population of analysis and drop out, leaving five admissible group sizes, $\mathcal{M} = \{1, 2, 3, 4, 5\}$. Each row of Table~\ref{table:karp-design1} collects the groups of one size, while the decision rule---the group-level treatment status $z$---is given by the first two columns. A cell of the table therefore corresponds to a size-by-treatment combination, each restricting the set of $(z, \mathcal{A})$ conditions to which a woman could be assigned.

Whether a woman's outcome is sensitive to the group set $\mathcal{A}$ is precisely the distinction between our two outcome cases. Under no interference (Outcome Case 0), her potential outcome depends only on the decision rule assigned to her group, so which other women share her group---and how many---is immaterial. Under group interference (Outcome Case 1), the potential outcome $Y_i(z, \mathcal{A})$ depends on the group set as well, so two women assigned the same decision rule but with different group partners (whether in terms of the identities of the group members or the number of group members) occupy different exposure conditions.

\begin{table}[!t]
\centering
\begin{tabular}{lcccc}
\toprule
 & \# Unanimous & \# Majority & Total \# & \# of \\
 & Groups & Groups & Groups & Individuals \\
\midrule
0 Females & 8  & 7 & 15 & 75 \\
1 Female  & 10 & 9 & 19 & 95 \\
2 Females & 6  & 7 & 13 & 65 \\
3 Females & 9  & 7 & 16 & 80 \\
4 Females & 8  & 8 & 16 & 80 \\
5 Females & 7  & 8 & 15 & 75 \\
\midrule
Total \# of Groups & 48 & 46 & 94 & \\
\# of Individuals  & 240 & 230 & & 470 \\
\bottomrule
\end{tabular}
\caption{\it Experimental design table from \citet{mendelberg2014does}, showing experimental conditions and sample sizes in each condition}
\label{table:karp-design1}
\end{table}


\section{Estimands and estimators}\label{sec:est}
We consider a variety of estimands motivated by the design and outcome cases, all defined with respect to the reference population $\mathcal{U}$. For any set of units $\mathcal{S}$, let $\mathcal{G}_{\mathcal{S}}(m) = \{\mathcal{A} \subseteq \mathcal{S} : |\mathcal{A}| = m\}$ denote the collection of all $m$-member groups that can be formed from units in $\mathcal{S}$, and let $\mathcal{G}_{\mathcal{S},i}(m) = \{\mathcal{A} \subseteq \mathcal{S} : i \in \mathcal{A},\ |\mathcal{A}| = m\}$ denote those that can be formed from $\mathcal{S}$ and contain unit $i$. We index a generic group from the population by $\omega$, whether fixed or randomly formed---the population-level counterpart of the index $g$ used for groups in the sample---and write $\mathcal{A}_\omega$ for its set of members.

\subsection{Estimands under fixed groups}
When SUTVA \citep{imbens2015causal} holds (no interference), as in Outcome Case 0, we target the population average treatment effect,
$$
\mathrm{ATE} := \frac{1}{n} \sum_{i \in \mathcal{U}}\left[Y_i(1)-Y_i(0)\right].
$$
Under group interference (Outcome Case 1), this expression is no longer well-defined: it is written in terms of $Y_i(1)$ and $Y_i(0)$, potential outcomes indexed by treatment alone, whereas interference makes a unit's outcome depend on its group composition as well. Recall that the general potential outcome is $Y_i(z, \mathcal{A})$, indexed by both the group's treatment status $z$ and its membership $\mathcal{A}$. Holding $i$'s group fixed at $\mathcal{A}_{\omega_i}$ while switching the group's treatment status defines the within-group contrast
$$
Y_i(1, \mathcal{A}_{\omega_i}) - Y_i(0, \mathcal{A}_{\omega_i}).
$$
This contrast is analogous to what \citet{hudgens2008toward} call the individual-level ``total effect'' under full ($100\%$) versus zero ($0\%$) group-level treatment saturation, given the group $\mathcal{A}_{\omega_i}$. The difference accounts for arbitrary interference between $i$ and the other units in $\mathcal{A}_{\omega_i}$. Averaging over the population gives the population average total effect,
$$
\mathrm{TOT} := \frac{1}{n}\sum_{i \in \mathcal{U}} \left[Y_i(1, \mathcal{A}_{\omega_i}) - Y_i(0, \mathcal{A}_{\omega_i})\right].
$$

\subsection{Estimands under random group formation}\label{estimands-rg}
Allowing for random group formation requires that we consider potential outcomes for any group to which unit $i$ might belong. We first define the marginalized potential outcome under homogeneous group sizes and then describe how it extends to varying sizes.

\paragraph{Homogeneous group sizes.} When all groups have a common size $M$, each group in $\mathcal{G}_{\mathcal{U},i}(M)$ is determined by its $M-1$ members besides $i$, chosen from the $n-1$ other units in $\mathcal{U}$, so $|\mathcal{G}_{\mathcal{U},i}(M)| = \binom{n-1}{M-1}$. The individualistic marginalized potential outcome for unit $i$ averages its potential outcome uniformly over these groups:
$$
\mu_i(z) = \frac{1}{\binom{n-1}{M-1}} \sum_{\mathcal{A} \in \mathcal{G}_{\mathcal{U},i}(M)} Y_i(z, \mathcal{A}).
$$
The individual marginalized effect is $\tau_i = \mu_i(1) - \mu_i(0)$, and the population average marginalized effect (PAME) is
$$
\tau_{\mathrm{PAME}} = \frac{1}{n} \sum_{i \in \mathcal{U}} \tau_i.
$$
The PAME aggregates these unit-level effects, each marginalized over the variation in group composition that random group formation induces. Because that design makes $\mathcal{A}_{A_i}$ uniform over $\mathcal{G}_{\mathcal{U},i}(M)$ and independent of $Z_i$, the marginalized potential outcome coincides with the design-conditional expectation, $\mu_i(z) = \E[Y_i(z, \mathcal{A}_{A_i}) \mid Z_i = z]$.

\paragraph{Varying group sizes.} When group sizes vary within an arm, the two characterizations of $\mu_i(z)$---as a uniform average over candidate groupings and as the design-conditional expectation $\E[Y_i \mid Z_i = z]$---do not always coincide. Generalizing the PAME therefore requires a choice about how to weight potential outcomes across group sizes, and the right choice depends on substantive considerations.

We first define the within-size-$m_k$ marginalized potential outcome under treatment status $z$:
$$
\mu_i(z; m_k) = \frac{1}{\binom{n-1}{m_k-1}}\sum_{\mathcal{A} \in \mathcal{G}_{\mathcal{U},i}(m_k)} Y_i(z, \mathcal{A}),
$$
which reduces to $\mu_i(z)$ when group sizes are homogeneous. Recall that $\mathcal{M}_z$ is the set of admissible group sizes in arm $z$; for instance, if groups of size 2, 3, or 4 are admissible under treatment while only singletons are admissible under control, then $\mathcal{M}_1 = \{2,3,4\}$ and $\mathcal{M}_0 = \{1\}$. We combine the within-size outcomes through a weighted average,
$$
\bar\mu^{(\phi)}_i(z) = \sum_{m_k \in \mathcal{M}_z}\phi(z;m_k)\, \mu_i(z;m_k),
$$
where the weights $\phi(z;m_k) \ge 0$ sum to one over $m_k \in \mathcal{M}_z$ and encode how much each admissible size counts toward the estimand. The unit-level marginalized exposure effect under these weights is $\tau^{(\phi)}_i = \bar\mu^{(\phi)}_i(1) - \bar\mu^{(\phi)}_i(0)$, and the population average marginalized effect under $\phi$ is 
$$
\begin{aligned}
\tau^{(\phi)}_{PAME} = & \frac{1}{n}\sum_{i \in \mathcal{U}}\tau^{(\phi)}_i = \frac{1}{n}\sum_{i \in \mathcal{U}}\left(\bar\mu^{(\phi)}_i(1) - \bar\mu^{(\phi)}_i(0)\right) \\
= & \sum_{m_k \in \mathcal{M}_1}\phi(1;m_k) \mu(1;m_k) - \sum_{m_k \in \mathcal{M}_0}\phi(0;m_k) \mu(0;m_k),
\end{aligned}
$$
where $\mu(z;m_k) = \frac{1}{n}\sum_{i \in \mathcal{U}} \mu_i(z;m_k)$. Different choices of $\phi$ give different forms of this single estimand; we highlight two of them.

The first weights each size by the design's own assignment probability, $\phi(z;m_k) = \Pr(|\mathcal{A}_{A_i}| = m_k \mid Z_i = z) = N_z(m_k) / N_z$, the share of arm-$z$ units the design places in size-$m_k$ groups. Under this choice $\bar\mu^{(\phi)}_i(z) = \E[Y_i(z, \mathcal{A}_{A_i}) \mid Z_i = z]$ is exactly the design-conditional expectation from above, and $\tau^{(\phi)}_{PAME}$ is the probability limit of the unweighted difference in means. Its drawback is that this form is not a stable population quantity: it depends on the experimental assignment probabilities, so adjusting the design to change how often individuals land in groups of different sizes changes the target.

The second weights all admissible sizes evenly, $\phi(z;m_k) = 1/|\mathcal{M}_z|$. This form is stable: since the weights ignore the design's allocation across sizes, $\tau^{(\phi)}_{PAME}$ is unchanged by that allocation. It is natural when the researcher regards each group size as a substantively meaningful condition deserving equal weight. The two forms coincide under homogeneous group sizes, and more generally whenever the design assigns the same number of units to each admissible size within each arm, so that the design shares are themselves uniform. Other weightings are possible and may be preferable in particular applications; Appendix~A.2 develops the results for a general $\phi$.

\subsection{Estimators}
We begin with the difference-in-means estimator, which we use across most of the design and outcome cases, and then generalize it to target the weighted PAME. Let $\bar Y_z = N_z^{-1}\sum_{i \in \mathcal{N}}\mathbf{1}\{Z_i = z\}Y_i$ be the mean outcome of arm-$z$ units. The difference in means is $\bar Y_1 - \bar Y_0$. When the treatment probability varies across groups, we replace it with the more general IPW or H\'ajek estimator. Details are deferred to the Appendix.

To target $\tau^{(\phi)}_{PAME}$ under weights $\phi$, we first estimate, for each arm $z$ and admissible size $m_k \in \mathcal{M}_z$, the within-size mean outcome
$$
\hat\mu(z;m_k) = \frac{\sum_{i \in \mathcal{N}}\mathbf{1}\{|\mathcal{A}_{A_i}| = m_k\}\,\mathbf{1}\{Z_i = z\}\,Y_i}{G_{Nz}(m_k)\,m_k}.
$$
The plug-in estimator for $\tau^{(\phi)}_{PAME}$ is then
$$
\hat\tau_N = \sum_{m_k \in \mathcal{M}_1}\phi(1;m_k)\,\hat\mu(1;m_k) \;-\; \sum_{m_k \in \mathcal{M}_0}\phi(0;m_k)\,\hat\mu(0;m_k).
$$
Setting $\phi$ to the realized design shares $\hat\phi(z;m_k) = N_z(m_k)/N_z$ makes the within-size means aggregate back to the arm means, so $\hat\tau_N = \bar Y_1 - \bar Y_0$ recovers the difference in means and targets the design-weighted form. Setting $\phi(z;m_k) = 1/|\mathcal{M}_z|$ instead targets the evenly-weighted form from the previous subsection. 


The estimator $\hat\tau_N$ also has a regression representation: it can be computed as a (weighted) least squares regression of $Y_i$ on treatment, and covariate adjustment enters as in \citet{lin2013agnostic}, with the regression estimator interpretable as a (weighted) difference in adjusted means.

\paragraph{Variance estimators.} We also need variance estimators for $\hat\tau_N$. We give them first for the difference in means $\bar Y_1 - \bar Y_0$, which is $\hat\tau_N$ under homogeneous group sizes, and then generalize. The heteroskedasticity-robust (``HC2'' or Neyman) estimator \citep{splawa1990application, samii2012equivalencies, abadie2023should} and the cluster-robust (``CR2'') estimator \citep{pustejovsky2018small} are, respectively,
$$
\begin{aligned}
\hat V^{HR}_N &= \frac{\sum_{i \in \mathcal{N}}Z_i(Y_i - \bar Y_1)^2}{N_1(N_1-1)} + \frac{\sum_{i \in \mathcal{N}}(1-Z_i)(Y_i - \bar Y_0)^2}{N_0(N_0-1)}, \\
\hat V^{CR}_N &= \sum_{g: Z_g=1}\frac{M_g^2(\bar Y_g - \bar Y_1)^2}{N_1(N_1 - M_g)} + \sum_{g: Z_g=0}\frac{M_g^2(\bar Y_g - \bar Y_0)^2}{N_0(N_0 - M_g)},
\end{aligned}
$$
with $\bar Y_g = M_g^{-1}\sum_{i \in \mathcal{A}_g}Y_i$ the sample mean of group $g$. When clusters are of equal size within each arm, $\hat V^{CR}_N$ reduces further to the cluster-level HC2 analogue \citep{su2021model}, $\hat V^{CR}_N = G_{N1}^{-1}(G_{N1}-1)^{-1}\sum_{g:Z_g=1}(\bar Y_g - \bar Y_1)^2 + G_{N0}^{-1}(G_{N0}-1)^{-1}\sum_{g:Z_g=0}(\bar Y_g - \bar Y_0)^2$. The two are also the sandwich variance estimators from regressing $Y_i$ on $Z_i$, at the unit and group level respectively, and differ only in whether the within-group cross-products of outcome residuals are retained: keeping them is exactly how the cluster-robust estimator accounts for group-level \emph{intra-cluster correlation}, a point to which we return below.

With varying group sizes the same construction applies cell by cell. Random sampling and random assignment render the within-arm, within-size means $\hat\mu(z;m_k)$ independent across cells as the population goes to infinity. Thus, the variance of $\hat\tau_N$ obeys the following as the population size goes to infinity:
$$
\frac{\sum_{m_k \in \mathcal{M}_1}\phi(1;m_k)^2\,\Var[\hat\mu(1;m_k)] + \sum_{m_k \in \mathcal{M}_0}\phi(0;m_k)^2\,\Var[\hat\mu(0;m_k)]}{\Var[\hat\tau_N]} \to 1 \text{ as } n \to \infty,
$$
with cells defined separately within each arm, so the two arms may differ in their size structure. We estimate each cell variance $\Var[\hat\mu(z;m_k)]$ by the within-cell analogue of the heteroskedasticity-robust or cluster-robust estimator above, aggregating with the weights $\phi(z;m_k)^2$ that define $\hat\tau_N$. Which of the two is appropriate---in particular, when the heteroskedasticity-robust form remains valid---is the subject of the next section. We give the exact per-cell forms, together with the weighted-least-squares sandwich representation that yields both, in Appendix~A.2.

\section{Theory}\label{sec:inference}
Our inferential conclusions are based on a sparse-sampling asymptotic regime in which the admissible group sizes $\mathcal{M}$ are fixed, $n, N, G_N \to \infty$, and $\frac{N^2}{n} \to 0$. This regime yields results similar to those of the design-based superpopulation sampling analysis of \citet{bai2022optimality} and \citet{bai2024primer}. In either case, the goal is to characterize the relationship between what transpired in a particular experiment and a more general population.

We begin by stating our two main results. Theorem~\ref{thm:normality} establishes the asymptotic normality of $\hat\tau_N$ in all four cases, and Theorem~\ref{thm:variance} characterizes the two variance estimators. We then discuss how to construct confidence intervals based on these theorems and conduct inference in practice.

\begin{theorem}[Asymptotic normality]\label{thm:normality}
Given Assumption~\ref{assn:po}, the Design Cases indexed by $d \in \{0,1\}$ and Outcome Cases indexed by $o \in \{0,1\}$, and fixed admissible group sizes $\mathcal{M}$, with $n, N, G_N \to \infty$ and $N^2/n \to 0$, the following holds in all four Design--Outcome cases:
\[
\frac{\hat\tau_N - \tau_{do}}{\sqrt{\mathrm{Var}[\hat\tau_N]}}
\rightsquigarrow \mathcal{N}(0,1),
\]
where $\tau_{00} = \tau_{10} = \mathrm{ATE}$, $\tau_{01} = \mathrm{TOT}$, and $\tau_{11} = \tau_{\mathrm{PAME}}$ (homogeneous group sizes) or $\tau^{(\phi)}_{\mathrm{PAME}}$ (varying sizes).
\end{theorem}

\begin{theorem}[Variance estimation]\label{thm:variance}
Under the conditions of Theorem~\ref{thm:normality}, the cluster-robust estimator is ratio-consistent for $\mathrm{Var}[\hat\tau_N]$ in all four cases,
\[
\frac{\hat V^{CR}_N}{\mathrm{Var}[\hat\tau_N]} \xrightarrow{P} 1,
\]
while the heteroskedasticity-robust estimator $\hat V^{HR}_N$ is ratio-consistent in case 1-0 only.
\end{theorem}

In every Design--Outcome case except the random-groups, no-interference one (1-0), $\mathrm{Var}[\hat\tau_N]$ carries a within-group (intra-cluster) correlation component. The cluster-robust estimator retains the within-group cross-products and so captures this component whatever its source---homophily in fixed groups (case 0-0) or interference (cases 0-1 and 1-1)---and is therefore consistent throughout. The heteroskedasticity-robust estimator discards those cross-products and is inconsistent for $\mathrm{Var}[\hat\tau_N]$ whenever the correlation is present. Only in case 1-0 does it vanish: there the design reduces to unit-level randomization and the variance collapses to the corresponding unit-level Neyman variance, so the two estimators share a probability limit and both are consistent, with the heteroskedasticity-robust estimator the more efficient of the two. Explicit forms for all cases appear in Appendix~A.1 and~A.2.

Table~\ref{tab:summary-results} collects the implications of the two theorems. The estimand depends on the cell: any commonly used estimator targets the $\mathrm{ATE}$ in both no-interference cells, the $\mathrm{TOT}$ under fixed groups with interference, and $\tau_{\mathrm{PAME}}$ under random groups with interference (its weighted form $\tau^{(\phi)}_{\mathrm{PAME}}$ when group sizes vary).

\begin{table}[]
    \centering
\begin{center}
\begin{tabular}{cc|c|c|}
& & \multicolumn{2}{c|}{Group formation}\\
& & No & Yes\\
\hline
Interference & No & $\hat \tau_N$ Targets $ATE$ &$\hat \tau_N$ Targets $ATE$ \\ 
& & CR Consistent & HR and CR Consistent \\ 
\hline
& Yes & $\hat \tau_N$ Targets $TOT$ & $\hat \tau_N$ Targets $PAME$\\
& & CR Consistent & CR Consistent \\
\hline
\end{tabular}
\end{center}
    \caption{Summary of population inference results.}
    \label{tab:summary-results}
\end{table}

Combining the two theorems through Slutsky's theorem yields the studentized statements used in practice, which we present as a corollary: 
\begin{corollary}[Asymptotic normality of the t-statistic]\label{thm:cor}
Under the conditions of Theorem~\ref{thm:normality}, the following holds in all four Design--Outcome cases:
\[
\frac{\hat\tau_N - \tau_{do}}{\sqrt{\hat V^{CR}_N}}
\rightsquigarrow \mathcal{N}(0,1),
\]
where $\tau_{00} = \tau_{10} = \mathrm{ATE}$, $\tau_{01} = \mathrm{TOT}$, and $\tau_{11} = \tau_{\mathrm{PAME}}$ (homogeneous group sizes) or $\tau^{(\phi)}_{\mathrm{PAME}}$ (varying sizes). In case 1-0 only, the same limit holds with $\hat V^{HR}_N$ in place of $\hat V^{CR}_N$. 
\end{corollary}

The recommended $(1-\alpha)$ confidence interval is
\[
CI^{BM} = \left(\hat\tau_N + t^{K_{BM}}_{\alpha/2}\sqrt{\hat V_N},\;
\hat\tau_N + t^{K_{BM}}_{1-\alpha/2}\sqrt{\hat V_N}\right),
\]
where $K_{BM}$ is the Bell--McCaffrey degrees-of-freedom adjustment of \citet{imbens2016robust}, and $\hat V_N = \hat V^{CR}_N$ in all four cases, with $\hat V_N = \hat V^{HR}_N$ available in case 1-0. 

\subsection{Sketch of the argument with homogeneous group sizes}
We provide a sketch of the proofs for Theorems~\ref{thm:normality} and~\ref{thm:variance} in the simplest case, in which all groups share a common size $M = N/G_N$. We focus on the novel combination of Design Case 1 and Outcome Case 1, for which the difference in means is the appropriate estimator, and indicate how the remaining design--outcome combinations follow as special cases. The full arguments appear in Appendix~A.1, and the generalization to varying group sizes is taken up in the next subsection. The analysis proceeds in five steps: (i) a group-level representation of the design, (ii) unbiasedness of the difference in means, (iii) the group-level variance and its estimator, (iv) the collapse to unit-level randomization under no interference, and (v) asymptotic normality via a coupling argument. We write $a_N \simeq b_N$ to mean that $a_N$ and $b_N$ share the same limit.

\textit{Step 1: a group-level representation of the design.} The key analytical move is to recast a random group formation experiment as a sampling design over a population of \textit{potential groups} that could be formed by the individuals in the population. Recall that $\mathcal{G}_{\mathcal{U}}(M)$ denotes the collection of all size-$M$ subsets of $\mathcal{U}$, with $|\mathcal{G}_{\mathcal{U}}(M)| = \binom{n}{M}$. Sequential random sampling of units followed by random group formation (Design Case 1) is equivalent in distribution to a restricted group-level design in which $G_N$ \textit{non-overlapping} members of $\mathcal{G}_{\mathcal{U}}(M)$ are drawn without replacement and then assigned to treatment and control as in a classical cluster-randomized experiment. Writing $W_\omega = 1$ if potential group $\omega$ is realized, each group enters with equal marginal probability $\E[W_\omega] = \pi_N = G_N/|\mathcal{G}_{\mathcal{U}}(M)|$, and the only departure from independent (with-replacement) sampling is that two overlapping groups can never be drawn together. That restriction is immaterial as the population size goes to infinity: the share of groups overlapping any fixed $\mathcal{A}_\omega$ is $1 - \binom{n-M}{M}/\binom{n}{M} \simeq M^2/n $ and for the complete $G_N$-tuple is $O(M^2 G^2_N/n)=O(N^@2n)$, which vanishes under the sparse-sampling condition, so realized groups are sampled \textit{as if} uniformly and independently from $\mathcal{G}_{\mathcal{U}}(M)$. This large-population independence is the feature that drives every result below. Because all members of a realized group share a treatment status and composition, the estimand reduces to a simple average of potential-group effects,
\[
\tau_{PAME} = \frac{1}{|\mathcal{G}_{\mathcal{U}}(M)|}\sum_{\omega=1}^{|\mathcal{G}_{\mathcal{U}}(M)|}\bigl[\bar{Y}_\omega(1) - \bar{Y}_\omega(0)\bigr], \qquad \bar{Y}_\omega(z) = \frac{1}{M}\sum_{i\in\mathcal{A}_\omega}Y_i(z,\mathcal{A}_\omega).
\]

\textit{Step 2: unbiasedness.} In this representation, unbiasedness of $\hat \tau_N$ is almost immediate. Conditional on the realized groups, $\hat{\tau}_N$ is the classical difference-in-means estimator applied to group-level outcomes, hence unbiased for the sample average treatment effect at the group level. Averaging over the sampling of potential groups then gives $\E[\hat{\tau}_N] = \tau_{PAME}$. The same conclusion applies at the unit level, which explains why the difference in means targets a marginalized estimand: under Design Case 1 the group realized for unit $i$ is uniform over its $\binom{n-1}{M-1}$ candidate size-$M$ groups, so $\E[Y_i \mid Z_i = z] = \mu_i(z)$, the individualistic marginalized potential outcome, and the difference in means therefore targets $n^{-1}\sum_i[\mu_i(1)-\mu_i(0)] = \tau_{PAME}$. The other cases follow by the similar sampling logic: under Design Case 0 with interference $\hat{\tau}_N$ is unbiased for the $TOT$ given random sampling of the fixed groups, and under either design with no interference it is unbiased for the $ATE$---for fixed groups because the group sizes are common\footnote{With varying sizes, the statement weakens to consistency as $n, N \to \infty$\citep{su2021model}.}, and for random groups because random formation makes the group-level assignment independent of potential outcomes, mimicking an individual-level trial.

\textit{Step 3: the group-level variance and its estimator.} Working at the group level yields a clean variance and a natural estimator for it. Decomposing $\Var[\hat{\tau}]$ by the law of total variance into a treatment-assignment component and a group-sampling component, the cross term and the sampling-overlap contribution are negligible under the sparse-sampling regime, leaving
\[
G_N \Var[\hat{\tau}] \simeq \frac{1}{p}\,\frac{1}{|\mathcal{G}_{\mathcal{U}}(M)|}\sum_{\omega}\bigl(\bar{Y}_\omega(1)-\bar{Y}(1)\bigr)^2 + \frac{1}{1-p}\,\frac{1}{|\mathcal{G}_{\mathcal{U}}(M)|}\sum_{\omega}\bigl(\bar{Y}_\omega(0)-\bar{Y}(0)\bigr)^2.
\]
This is exactly the limit of the Neyman variance of a completely randomized experiment run on group-level outcomes \citep{imbens2015causal}. The group-level HC2 estimator $\hat{V}^{CR}_N$ retains the within-group cross-products and is therefore consistent for the true variance. At the unit level, the cluster-robust estimator is equal in value to the group-level HC2 estimator, and is therefore consistent, whereas the heteroskedasticity-robust $\hat{V}^{HR}_N$ discards those cross-products and is not consistent.

\textit{Step 4: collapse to unit-level randomization under no interference.} Under Outcome Case 0, groupmates do not affect potential outcomes, and Design Case 1 becomes a completely randomized experiment at the unit level, so the group-level variance collapses to the unit-level Neyman variance. To see this, write the group-level sampling variance through an intra-cluster decomposition, $G_{Nz}\Var[\bar{Y}_z] \simeq M^{-1} S^2_{\mathcal{G}}\,(1 + \rho_{\mathcal{G}})$, where $S^2_{\mathcal{G}}$ is the potential-group analogue of the population variance of potential outcomes and $\rho_{\mathcal{G}}$ is the potential-group intra-cluster correlation. Under no interference, sampling two distinct group members without
replacement gives
\[
\rho_{\mathcal{G}} = -\frac{1}{n-1} \to 0 \text{ as } n \to \infty.
\]
There is thus no substantive within-group clustering of outcomes under no interference, only a negligible $O(1/n)$ negative correlation induced by the finite population. As a result, the heteroskedasticity-robust estimator becomes consistent and is the more efficient of the two; together with the cluster-robust consistency established in Step~3, this completes Theorem~\ref{thm:variance}. The vanishing of $\rho_{\mathcal{G}}$ is also what the ICC-based interference test developed below exploits.

\textit{Step 5: asymptotic normality via coupling.} We now turn to the limiting distribution of $\hat \tau_N$. The difficulty is that the restricted design's inclusion indicators $W_\omega$ are dependent, so a central limit theorem for independent draws does not apply directly. We instead \textit{couple} the restricted design to an auxiliary one that draws $G_N$ groups independently and with replacement from $\mathcal{G}_{\mathcal{U}}(M)$. Writing $Q$ and $Q^*$ for the laws of the ordered $G_N$-tuple of realized groups under the restricted and independent designs, we have $\mathrm{supp}(Q) \subseteq \mathrm{supp}(Q^*)$, and the total variation distance between them vanishes under $N^2/n \to 0$:
\[
d_{TV}(Q, Q^*) = 1 - \prod_{r=0}^{G_N-1}\frac{\binom{n-rM}{M}}{\binom{n}{M}} \simeq \frac{M^2 G_N^2}{2n} \to 0,
\]
again because the overlapping share is negligible, so the two designs are asymptotically indistinguishable. By the data-processing inequality the same convergence holds for the law of any function of the tuple, in particular the standardized group-level average; since that average is asymptotically normal under the independent design (Lindeberg--Feller, the summands being i.i.d.\ and bounded), the restricted design inherits the same normal limit. Combining this with the \citet{ohlsson1989asymptotic} two-stage martingale central limit theorem for the assignment stage delivers the conclusion of Theorem~\ref{thm:normality}.

\subsection{Generalization to variable group sizes}
The five-step argument extends to varying group sizes with two modifications, both developed at length in Appendix~A.2.

First, the group-level sampling representation extends to a per-size-class version: the population of potential groups becomes $\bigcup_{k=1}^{K} \mathcal{G}_{\mathcal{U}}(m_k)$, and the design realizes $G_N(m_k)$ groups of size $m_k$ with within-class sampling probability $\pi_N(m_k) = G_N(m_k)/|\mathcal{G}_{\mathcal{U}}(m_k)|$. The variance, central limit theorem, and consistency arguments carry through within each size class, and the simple-average aggregation across classes inherits the within-class asymptotic distributions since sampling across classes is asymptotically independent. The cluster-robust estimator $\hat{V}^{CR}_N$ (in its varying-sizes form above) remains consistent for $\Var[\hat\tau_N]$, while the heteroskedasticity-robust estimator $\hat V^{HR}_N$ is consistent only under no interference and random group formation.

Second, which $\tau^{(\phi)}_{PAME}$ an estimator is unbiased for depends on the weighting it places across size classes. The unweighted difference in means $\bar Y_1 - \bar Y_0$ is unbiased for $\tau^{(\phi)}_{PAME}$ at the design-share weights $\phi(z; m_k) = N_z(m_k)/N_z$, while the IPW estimator $\hat\tau_N$ reweights each size class evenly, targeting $\tau^{(\phi)}_{PAME}$ at $\phi(z; m_k) = 1/|\mathcal{M}_z|$ (the second weighting highlighted in Section~\ref{estimands-rg}). In general, within a size class $m_k$, the permutation-based group allocation in Design Case 1 assigns $\mathcal{A}_{A_i}$ uniformly over the $\binom{n-1}{m_k-1}$ size-$m_k$ groups in the population containing $i$, so $\E[Y_i \mid Z_i = z, |\mathcal{A}_{A_i}| = m_k] = \mu_i(z; m_k)$. Different estimators weight these conditional expectations differently and so converge to different members of the $\tau^{(\phi)}_{PAME}$ family.

\subsection{A diagnostic for interference}\label{sec:test}
When groups are randomly formed but interference is absent, both cluster-robust and heteroskedasticity-robust inference are asymptotically valid, but the heteroskedasticity-robust one is more efficient. It is therefore useful to determine which outcome regime we are in (Outcome Case 0 or 1), since this affects both the validity and the efficiency of inference.

One could simply test for a difference in the variance estimates, but a more direct test works with the intra-cluster correlations (ICCs) of the treated and control outcomes. Under Design Case 1 and Outcome Case 0, randomly formed groups do not generate persistent within-group dependence: within any treatment-by-size cell, the intraclass correlation is asymptotically zero, as in the homogeneous-size calculation yielding $\rho_{\mathcal{G}}=-(M-1)/(n-1)\to 0$. Under Outcome Case 1, by contrast, interaction among members of the realized group $\mathcal{A}_g$ may induce within-group dependence in observed outcomes. We therefore compute ICC statistics within treatment-by-size cells and combine the resulting evidence across cells.\footnote{We cannot pool the treatment and control groups, because the ICC is sensitive to the mean difference induced by the treatment effect. To see this, consider a balanced design with a constant effect and homogenous potential outcomes, such that $Y_i(1) = 10$ and $Y_i(0) = 5$ for all $i$. Within each treatment arm, for any group partition, the between group variation is zero, and so the arm-specific ICCs are zero. However, when we pool the two arms, the between group variation is no longer zero, and so the ICC is greater than zero.} This is not a direct test of interference, but it provides suggestive evidence.

For $z\in\{0,1\}$ and $m_k\in\mathcal{M}_z$, define the treatment-by-size cell
\[
\mathcal{C}_{z,m_k}
=
\left\{
g\in\{1,\ldots,G_N\}: Z_g=z,\; M_g=m_k
\right\},
\]
so that $G_{Nz}(m_k)=|\mathcal{C}_{z,m_k}|$ and $N_z(m_k)=m_kG_{Nz}(m_k)$.  We restrict attention to informative cells
\[
\mathcal{C}_{\mathrm{ICC}}
=
\left\{
(z,m_k): z\in\{0,1\},\; m_k\in\mathcal{M}_z,\; m_k\ge 2,\; G_{Nz}(m_k)\ge 2
\right\},
\]
since singleton groups do not contain within-group variation and cells with only one group do not identify between-group variation. For any $(z,m_k)\in\mathcal{C}_{\mathrm{ICC}}$, let
$
\bar Y_g = \frac{1}{m_k}\sum_{i\in\mathcal{A}_g}Y_i
$
denote the mean outcome in group $g\in\mathcal{C}_{z,m_k}$, and recall that the corresponding cell mean is
$
\hat\mu(z;m_k)
=
\frac{1}{N_z(m_k)}
\sum_{g\in\mathcal{C}_{z,m_k}}
\sum_{i\in\mathcal{A}_g}Y_i
=
\frac{1}{G_{Nz}(m_k)}
\sum_{g\in\mathcal{C}_{z,m_k}}\bar Y_g .
$
Define the between-group and within-group mean squares in cell $(z,m_k)$ as
\[
MSB_{z,m_k}
=
\frac{
	m_k\sum_{g\in\mathcal{C}_{z,m_k}}
	\left(\bar Y_g-\hat\mu(z;m_k)\right)^2
}{
	G_{Nz}(m_k)-1
},
\]
and
\[
MSW_{z,m_k}
=
\frac{
	\sum_{g\in\mathcal{C}_{z,m_k}}
	\sum_{i\in\mathcal{A}_g}
	\left(Y_i-\bar Y_g\right)^2
}{
	N_z(m_k)-G_{Nz}(m_k)
}.
\]
When $MSW_{z,m_k}>0$, the cell-level ANOVA statistic and ICC estimate are
$
F_{z,m_k}
=
\frac{MSB_{z,m_k}}{MSW_{z,m_k}},
$ and $
\hat\rho_{z,m_k}
=
\frac{MSB_{z,m_k}-MSW_{z,m_k}}
{MSB_{z,m_k}+(m_k-1)MSW_{z,m_k}}
=
\frac{F_{z,m_k}-1}{F_{z,m_k}+m_k-1} .
$

Under the traditional balanced one-way ANOVA model within cell $(z,m_k)$ and the null of zero within-group dependence,
$
F_{z,m_k}
\sim
F_{G_{Nz}(m_k)-1,\;N_z(m_k)-G_{Nz}(m_k)}.
$
Without the normality assumptions of the ANOVA model, and with fixed $m_k$ and $G_{Nz}(m_k)\to\infty$, we use the asymptotic approximation
$$
T_{z,m_k}
=
\frac{
	\sqrt{G_{Nz}(m_k)}\left(F_{z,m_k}-1\right)
}{
	\sqrt{2m_k/(m_k-1)}
}
\rightsquigarrow
\mathcal{N}(0,1),
$$
or equivalently,
$$
\sqrt{\frac{G_{Nz}(m_k)m_k(m_k-1)}{2}}\;
\hat\rho_{z,m_k}
\rightsquigarrow
\mathcal{N}(0,1).
$$
Let $p_{z,m_k}$ denote the corresponding upper-tail $p$-value for positive within-group dependence, using either the finite-sample $F$ reference distribution or the normal approximation, $p_{z,m_k}=1-\Phi(T_{z,m_k})$.

To combine evidence across informative cells, let
$J=|\mathcal{C}_{\mathrm{ICC}}|$ and define Fisher's combined statistic
$$
S_{\mathrm{ICC}}
=
-2
\sum_{(z,m_k)\in\mathcal{C}_{\mathrm{ICC}}}
\log p_{z,m_k}.
$$
Asymptotically, 
$
S_{\mathrm{ICC}}
\rightsquigarrow
\chi^2_{2J}
$
under the null of zero within-cell ICCs. We reject at level $\alpha$ when
$
S_{\mathrm{ICC}}
\ge
\chi^2_{2J;\,1-\alpha}.
$
In the homogeneous group-size case, $\mathcal{M}_0=\mathcal{M}_1=\{M\}$ and, provided $M\ge2$ with at least two groups in each arm, $J=2$, so the procedure reduces to combining the two treatment-stratum tests with a $\chi^2_4$ reference distribution.

\section{Simulation study}\label{sec:simulation}

We use a Monte Carlo simulation study to demonstrate that our convergence results provide reasonable approximations for inference even with a finite sample. To align the simulations with the asymptotic condition $N^2/n \to 0$, we let the population size $n$ increase faster than $N^2$.  The simulation setting is as follows (see the Appendix for a full description of the Monte Carlo specifications):

\begin{itemize}
\item $N \in \{80, 200, 300, 400\}$ and $M=4$. For these four values of $N$, respectively, we set $n \in \{404{,}772, 4{,}000{,}000, 11{,}022{,}704, 22{,}627{,}420\}$ according to $n \propto N^{5/2}$, calibrated so that $N^2/n=0.01$ at $N=200$. The corresponding values of $N^2/n$ are $0.0158$, $0.0100$, $0.0082$, and $0.0071$. Each unit $i$ carries a uniformly drawn scalar covariate $X_i$.
\item Individuals in the {\it population} are already partitioned into groups of size $M$ prior to sampling.  The potential outcomes contain a group-level random effect, representing homophily in the population.
\item Under the fixed-group design, we sample $G_N=N/M$ pre-existing population groups without replacement and retain all of their members. The sampled population groups therefore become the experimental groups, and the population-level within-group dependence carries into the experiment.
\item Under the random-group design, we instead draw $N$ individual units by simple random sampling without replacement and randomly partition them into $G_N=N/M$ experimental groups of size $M$. The pre-existing population groups are not used to select or form the experimental groups so the population grouping becomes irrelevant. 
\item When interference is present, {\it individual level treatment effects} are a function of other group members' covariates.
\item We compute averages over 2,000 Monte Carlo simulation runs.
\end{itemize}
Table~\ref{tab:sim-outcomes} shows the true estimand, the Monte Carlo mean and variance of the difference in means estimator ($\hat \tau$), and then intraclass correlation of potential outcomes for units in the control groups and treated groups.  The $N^2/n$ column reports the realized sparsity ratio for each design. We see that the estimand ($\tau$) differs across scenarios: $\tau$ corresponds to the PATE in the cases with no interference, the TOT with interference but fixed groups, and the PAME with interference and random groups. 
The TOT averages over potential groups drawn from the population, whereas PAME averages over potential groups that can be formed from individuals in the population.
Our simulation is such that these two averages are nearly the same, although this does not have to be the case with regard to actual populations.
The difference in means is unbiased for these estimands. 

\begin{table}[htbp]
\centering
\begingroup\small
\scalebox{0.9}{
\begin{tabular}{lrrrrrr}
  \hline
 & $N^2/n$ & $\tau$ & $\mathbb{E}[\hat{\tau}]$ & True Var. & ICC $Y(0)$ & ICC $Y(1)$ \\ 
  \hline
N=80 &  &  &  &  &  &  \\ 
  No Interference, Fixed Groups & 0.0158 & 0.5002 & 0.5073 & 0.0795 & 0.9107 & 0.1479 \\ 
  Interference, Fixed Groups & 0.0158 & 1.7122 & 1.7189 & 0.1675 & 0.9107 & 0.4285 \\ 
  No Interference, Randomly Formed Groups & 0.0158 & 0.5008 & 0.4988 & 0.0418 & -0.0188 & -0.0167 \\ 
  Interference, Randomly Formed Groups & 0.0158 & 1.7050 & 1.6986 & 0.1306 & -0.0188 & 0.3245 \\ 
   \hline
N=200 &  &  &  &  &  &  \\ 
  No Interference, Fixed Groups & 0.0100 & 0.4986 & 0.4994 & 0.0322 & 0.9192 & 0.1648 \\ 
  Interference, Fixed Groups & 0.0100 & 1.7063 & 1.7055 & 0.0676 & 0.9192 & 0.4487 \\ 
  No Interference, Randomly Formed Groups & 0.0100 & 0.4988 & 0.4963 & 0.0178 & -0.0044 & -0.0060 \\ 
  Interference, Randomly Formed Groups & 0.0100 & 1.7027 & 1.6970 & 0.0549 & -0.0044 & 0.3438 \\ 
   \hline
N=300 &  &  &  &  &  &  \\ 
  No Interference, Fixed Groups & 0.0082 & 0.4998 & 0.5003 & 0.0207 & 0.9199 & 0.1686 \\ 
  Interference, Fixed Groups & 0.0082 & 1.7097 & 1.7141 & 0.0450 & 0.9199 & 0.4546 \\ 
  No Interference, Randomly Formed Groups & 0.0082 & 0.4992 & 0.4998 & 0.0111 & -0.0040 & -0.0031 \\ 
  Interference, Randomly Formed Groups & 0.0082 & 1.7042 & 1.7043 & 0.0350 & -0.0040 & 0.3484 \\ 
   \hline
N=400 &  &  &  &  &  &  \\ 
  No Interference, Fixed Groups & 0.0071 & 0.4979 & 0.4953 & 0.0158 & 0.9211 & 0.1693 \\ 
  Interference, Fixed Groups & 0.0071 & 1.7038 & 1.7008 & 0.0315 & 0.9211 & 0.4537 \\ 
  No Interference, Randomly Formed Groups & 0.0071 & 0.5000 & 0.4976 & 0.0083 & -0.0022 & -0.0028 \\ 
  Interference, Randomly Formed Groups & 0.0071 & 1.7088 & 1.7063 & 0.0258 & -0.0022 & 0.3529 \\ 
   \hline
\end{tabular}
}
\endgroup
\caption{True treatment effect estimands ($\tau$), mean and variance of the difference-in-means estimator ($\hat{\tau}$), and intra-class correlations for outcomes in control $Y(0)$ and treatment $Y(1)$ groups.} 
\label{tab:sim-outcomes}
\end{table}

For untreated control groups, Table~\ref{tab:sim-outcomes} shows substantial intraclass outcome correlation when groups are fixed.
This is due to the homophily in the groups that exist in the population.
The control group intraclass correlation is zero when groups are randomly formed.
Recall that interference in the simulation comes only through treatment effects, which do not pertain to control group observations. 
This shows how, in the absence of interference, random group formation yields what is essentially an individually-randomized experiment.
For the treated groups, when groups are fixed and there is no interference, unit-level heterogeneity in treatment effects makes the intra-class correlation smaller than for the control groups.
The introduction of interference with fixed groups brings the intraclass correlation back up. 
The difference between the intraclass correlation in the first two scenarios is due to the within-group interference that translates into treatment effects.
When we have random group formation and no interference, again, there is no group-level intraclass correlation due to homophily, showing that the scenario is equivalent to unit-level randomization without interference.
But when interference is reintroduced (the last scenario), this creates a novel source of group-level intraclass correlation.

\begin{figure}[!t]
	\begin{center}
		\includegraphics[width=.75\textwidth]{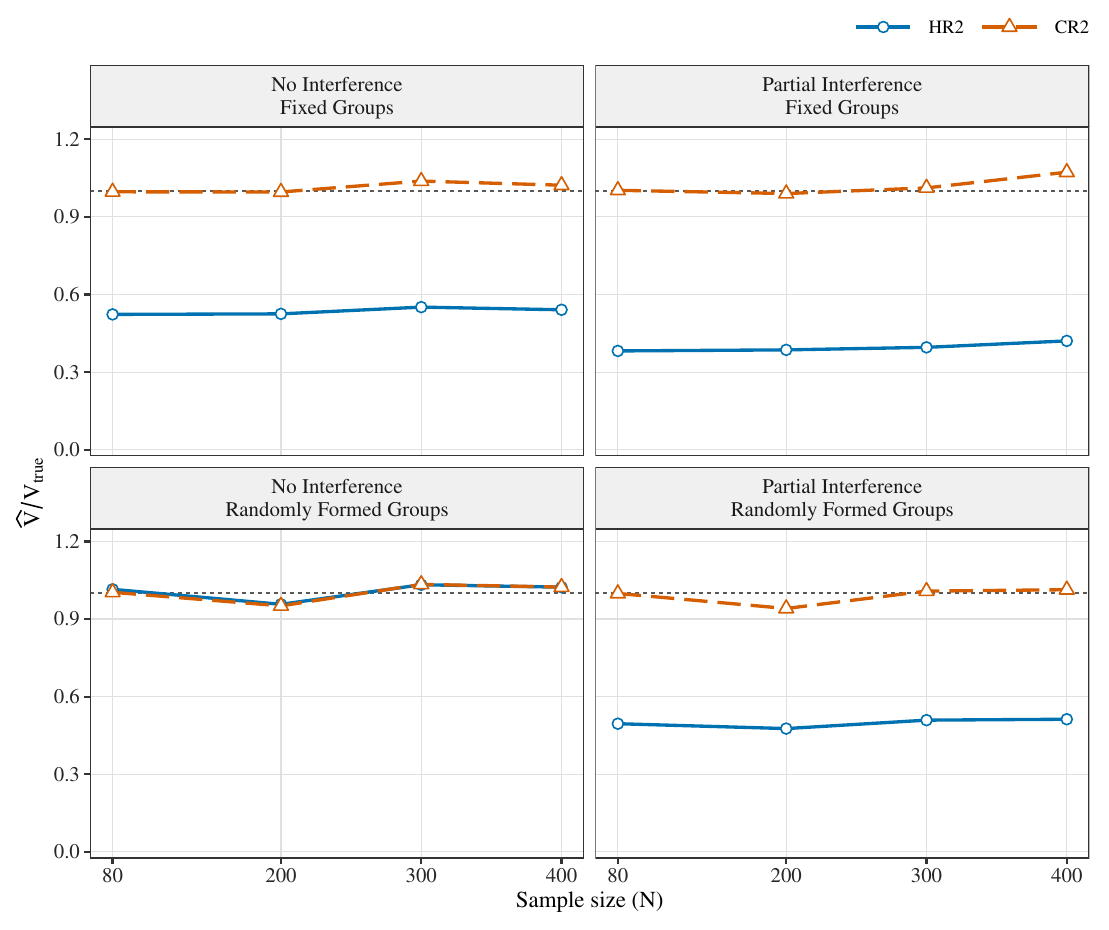}
		\caption{Simulation results for ratio of variance estimators relative to the true variance of the difference in means estimator. ``HR2'' denotes the individual-level HR2 heteroskedasticity robust variance estimator, and ``CR2'' denotes the individual-level cluster-robust variance estimator.}
		\label{fig:sim-vratio}
	\end{center}
\end{figure}

Figure~\ref{fig:sim-vratio} and Table~\ref{tab:sim-sdv} present simulation results for the variance estimators. Figure~\ref{fig:sim-vratio} compares the performance of the heteroskedasticity-robust (HR2) and cluster-robust (CR2) variance estimators. We observe convergence of the cluster-robust variance estimator to the true variance across all scenarios. By contrast, the heteroskedasticity-robust variance estimator is accurate only in the case with no interference and randomly formed groups. Table~\ref{tab:sim-sdv} reports the empirical coverage rates of confidence intervals based on HR2 and CR2. We find that the coverage rates using CR2 are very close to the nominal 95\% level across all four cases. However, when examining the standard deviations of the variance estimators, we observe that under no interference and randomly formed groups, HR2 is more precise than CR2, as expected.

\begin{table}[htbp]
\centering
\begingroup\small
\scalebox{0.9}{
\begin{tabular}{lrrrrr}
  \hline
 & $N^2/n$ & Coverage.HR2 & Coverage.CR2 & sd(HR2) & sd(CR2) \\ 
  \hline
N=80 &  &  &  &  &  \\ 
  No Interference, Fixed Groups & 0.0158 & 0.8330 & 0.9330 & 0.0088 & 0.0281 \\ 
  Partial Interference, Fixed Groups & 0.0158 & 0.7500 & 0.9230 & 0.0175 & 0.0718 \\ 
  No Interference, Randomly Formed Groups & 0.0158 & 0.9445 & 0.9280 & 0.0081 & 0.0168 \\ 
  Partial Interference, Randomly Formed Groups & 0.0158 & 0.8295 & 0.9215 & 0.0158 & 0.0612 \\ 
   \hline
N=200 &  &  &  &  &  \\ 
  No Interference, Fixed Groups & 0.0100 & 0.8390 & 0.9415 & 0.0022 & 0.0070 \\ 
  Partial Interference, Fixed Groups & 0.0100 & 0.7705 & 0.9340 & 0.0044 & 0.0172 \\ 
  No Interference, Randomly Formed Groups & 0.0100 & 0.9395 & 0.9345 & 0.0021 & 0.0042 \\ 
  Partial Interference, Randomly Formed Groups & 0.0100 & 0.8220 & 0.9335 & 0.0040 & 0.0149 \\ 
   \hline
N=300 &  &  &  &  &  \\ 
  No Interference, Fixed Groups & 0.0082 & 0.8530 & 0.9500 & 0.0012 & 0.0038 \\ 
  Partial Interference, Fixed Groups & 0.0082 & 0.7820 & 0.9405 & 0.0025 & 0.0098 \\ 
  No Interference, Randomly Formed Groups & 0.0082 & 0.9600 & 0.9550 & 0.0011 & 0.0023 \\ 
  Partial Interference, Randomly Formed Groups & 0.0082 & 0.8275 & 0.9485 & 0.0023 & 0.0082 \\ 
   \hline
N=400 &  &  &  &  &  \\ 
  No Interference, Fixed Groups & 0.0071 & 0.8510 & 0.9510 & 0.0008 & 0.0024 \\ 
  Partial Interference, Fixed Groups & 0.0071 & 0.7920 & 0.9530 & 0.0016 & 0.0062 \\ 
  No Interference, Randomly Formed Groups & 0.0071 & 0.9495 & 0.9455 & 0.0007 & 0.0015 \\ 
  Partial Interference, Randomly Formed Groups & 0.0071 & 0.8445 & 0.9450 & 0.0015 & 0.0052 \\ 
   \hline
\end{tabular}
}
\endgroup
\caption{Monte Carlo coverage and standard deviations for the variance estimators. Coverage.HR2 and Coverage.CR2 denote empirical coverage of nominal 95\% confidence intervals based on the HR2 and CR2 variance estimators.} 
\label{tab:sim-sdv}
\end{table}

In the group interaction experiment, the treatment effect is determined by treatment assignment, individual covariates, group members’ covariates, and group-level covariates. As in other experimental analyses, we can collect and control for these covariates in the regression. Figure~\ref{fig:sim-cov} presents the results from regression estimators with control variables. We consider several specifications. We begin with no covariates, then add only individual-level covariates, only group-level covariates, and finally both. We also estimate the treatment effect using Lin's regression adjustment \citep{lin2013agnostic}. The vertical black line marks the estimand, $\tau_{PAME}$, while the colored density curves show the distributions of the estimates. All estimates are centered around the true PAME, and the biases are negligible across specifications. However, the distributions become visibly tighter as more relevant covariates are included, especially under Lin's regression adjustment, where the RMSE is reduced from 0.65 to 0.45. 


\begin{figure}[!t]
	\begin{center}
		\includegraphics[width=.9\textwidth]{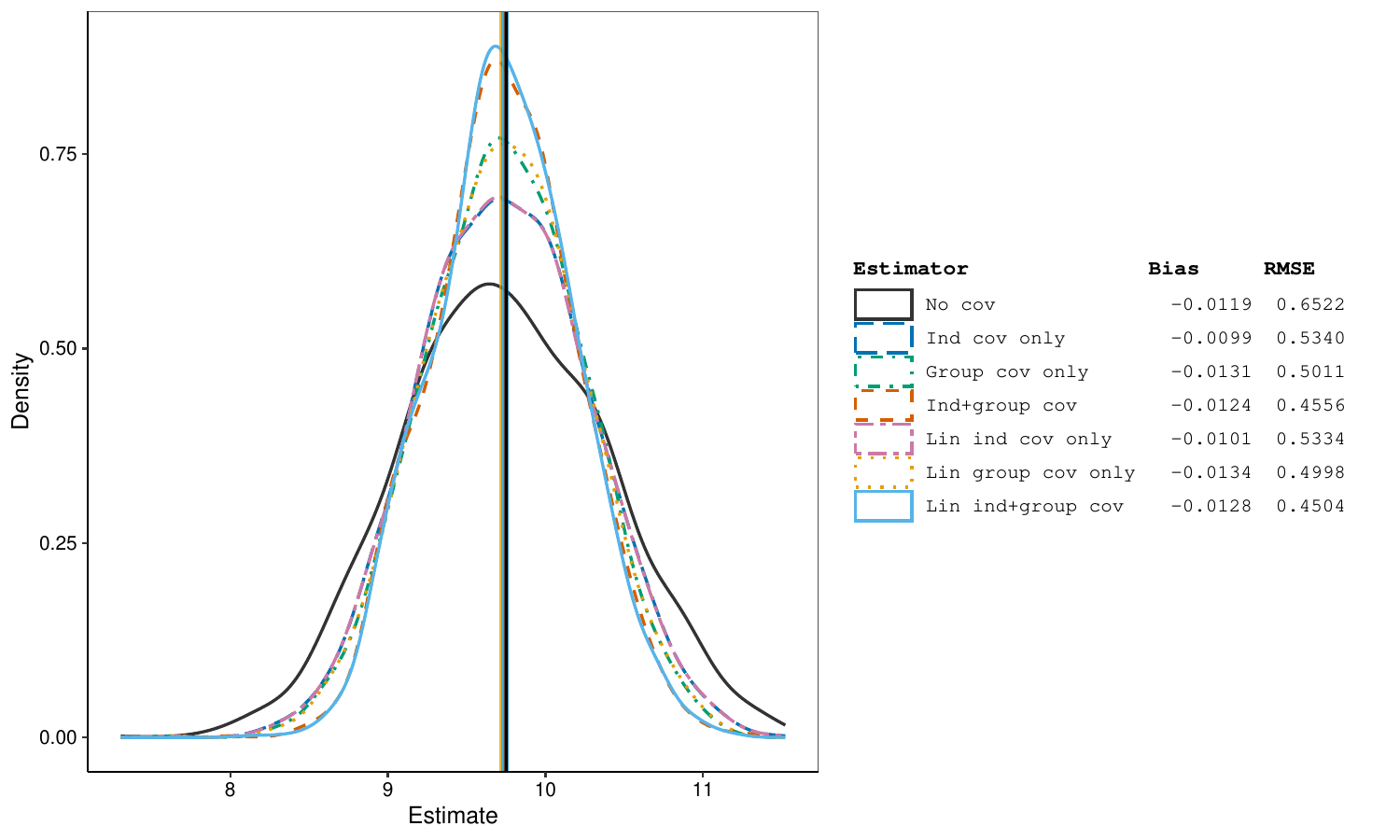}
		\caption{Distribution, Bias and RMSE of regression estimates with and without control variables in the group interaction experiment (N=400).}
		\label{fig:sim-cov}
	\end{center}
\end{figure}

Table \ref{tab:int_test} shows the simulation results for our proposed interference test. The column $\mathbb{P}(p \le 0.05)$ reports the proportion of p-values from the test that are smaller than 0.05. When there is no interference, we expect this value to be close to the nominal level of 0.05, indicating correct Type I error control. When interference is present, we expect this value to be larger and close to 1, indicating high power to detect interference and reject the null hypothesis. Overall, even when the sample size is 80, the test exhibits correct Type I error (0.051) and acceptable power (0.952).

\begin{table}[htbp]
\centering
\begingroup\small
\scalebox{0.9}{
\begin{tabular}{lr}
  \hline
 & $\mathbb{P}(p \leq 0.05)$ \\ 
  \hline
N=80 &  \\ 
  No Interference, Randomly Formed Groups & 0.06 \\ 
  Interference, Randomly Formed Groups & 0.59 \\ 
   \hline
N=200 &  \\ 
  No Interference, Randomly Formed Groups & 0.05 \\ 
  Interference, Randomly Formed Groups & 0.92 \\ 
   \hline
N=300 &  \\ 
  No Interference, Randomly Formed Groups & 0.05 \\ 
  Interference, Randomly Formed Groups & 0.98 \\ 
   \hline
N=400 &  \\ 
  No Interference, Randomly Formed Groups & 0.05 \\ 
  Interference, Randomly Formed Groups & 1.00 \\ 
   \hline
\end{tabular}
}
\endgroup
\caption{Monte Carlo simulation results for the interference test. The second column reports the proportion of p-values no larger than 0.05.} 
\label{tab:int_test}
\end{table}

\section{Application}\label{sec:app}


\subsection{Application I: gender composition and decision rule on deliberation}

We revisit the experiment presented in \citet{mendelberg2014does}.  
The experiment randomly assigned individual men and women to 5-person deliberation groups who were instructed to ``decide collectively on the `most just' principle of redistribution and set a poverty line in dollars'' (p. 4). 
The groups varied in the share of women versus men (from 0 to 5 women).  The groups also varied in whether their decision would be made via majority vote or consensus. 
The experimental design was shown in Table~\ref{table:karp-design1} above.
The main outcome in the paper was the extent to which women raised ``care issues'' (e.g., concerns about raising children) in the context of the deliberations.

\begin{table}[!t]
    \centering
    \begin{threeparttable}
        \caption{Reanalyses of \citet{mendelberg2014does}}\label{tab:mendrep}
            \footnotesize
            \begin{tabular}{llllll}
  \hline
 Outcome: & \multicolumn{2}{c}{Care issues} & \multicolumn{3}{c}{ Taking Time (Minutes) } \\ 
   & (1) & (2) & (3) & (4) & (5) \\
  \hline
  \hline
  Coefficients: \\
  \\
Maj.Rule & -1.035 & -0.786 & -0.875 & -1.026 & -1.580 \\ 
   Clust.Rob.SE & (0.446)** & (0.333)** & (0.432)** & (0.429)** & (1.295) \\ 
   Het.Rob.SE & (0.395)*** & (0.332)** & (0.376)** & (0.400)** & (1.289) \\ 
   &  &  &  &  &  \\ 
  Num.Women & -0.001 &  &  &  &  \\ 
   Clust.Rob.SE & (0.086) &  &  &  &  \\ 
   Het.Rob.SE & (0.084) &  &  &  &  \\ 
   &  &  &  &  &  \\ 
  Maj.Rule * Num.Women & 0.367 &  &  &  &  \\ 
   Clust.Rob.SE & (0.158)** &  &  &  &  \\ 
   Het.Rob.SE & (0.128)*** &  &  &  &  \\ 
   &  &  &  &  &  \\ 
  Maj.Rule X 2 women &  & 0.419 &  &  & 0.520 \\ 
   Clust.Rob.SE &  & (0.522) &  &  & (1.545) \\ 
   Het.Rob.SE &  & (0.486) &  &  & (1.554) \\ 
   &  &  &  &  &  \\ 
  Maj.Rule X 3 women &  & 1.092 &  &  & 0.427 \\ 
   Clust.Rob.SE &  & (0.556)* &  &  & (1.522) \\ 
   Het.Rob.SE &  & (0.487)** &  &  & (1.440) \\ 
   &  &  &  &  &  \\ 
  Maj.Rule X 4 women &  & 1.119 &  &  & 1.162 \\ 
   Clust.Rob.SE &  & (0.476)** &  &  & (1.488) \\ 
   Het.Rob.SE &  & (0.409)*** &  &  & (1.469) \\ 
   &  &  &  &  &  \\ 
  ICC Y(0) &  0.271*** &  &  0.314*** &  &  \\ 
  ICC Y(1) &  0.183*** &  &  0.388*** &  &  \\ 
  $N$ & 157 & 157 & 157 & 157 & 157 \\ 
IPW weights: & N & N & N & Y & N\\
   \hline
\end{tabular}
        \begin{tablenotes}
            \item\textit{Notes:} Entries are OLS coefficient estimates for women in mixed-gender groups. 
  The dependent variables are the frequency of care issues raised and talking time in minutes. 
  Columns 1 and 2 report models for care issues: column 1 uses a linear count of the number of women and its interaction with majority rule, while column 2 uses indicators for groups with two, three, or four women and their interactions with majority rule. 
  Columns 3 and 4 report models for speaking time without and with inverse-probability weights, respectively, and column 5 reports the interaction effects for speaking time. 
  All regressions include controls for liberalism, the number of liberals, and the Princeton indicator. 
  Clustered robust standard errors are clustered by group; heteroskedasticity-robust standard errors are reported separately. 
  Significance stars are attached to the corresponding standard-error row and are calculated using that row's standard error: * p $<$ 0.10, ** p $<$ 0.05, *** p $<$ 0.01. 
  ICC rows report intraclass correlations for the outcome under non-majority-rule and majority-rule conditions, respectively, and Significance stars on the ICC rows denote results from the interference test in Section \ref{sec:test}. 
  $N$ is the number of individual observations..
        \end{tablenotes}
    \end{threeparttable}
\end{table}

The main analysis in the paper studied the interaction effect between the decision rule (unanimity versus majority vote) and the number of women in the group. It was implemented, for women in mixed groups (those with at least one woman and one man), by regressing the frequency with which a woman raised care issues on an indicator for being in a majority-rule group, the number of women in the group, and their interaction. Recognizing that the group structure might be relevant, \citet{mendelberg2014does} used cluster-robust standard errors.

The first column in Table~\ref{tab:mendrep} presents the results from this analysis, with standard errors from both the CR2 and HR2 estimators. The difference is pronounced---for the interaction term, the cluster-robust standard error is 23\% larger---and tracks the scale of the intraclass correlations shown at the bottom of the table. Because individuals were assigned to groups randomly, this intraclass correlation is attributable to within-group interference. We also conduct the interference test of Section~\ref{sec:test}, and the null of no interference is rejected at the 0.001 level. In the second column, we further interact majority rule with indicators for each number of women in the group. The gap between the two is again much larger for the interaction effects: for the group with three women, we may even reach different conclusions about the treatment effect's significance depending on which estimator is used.

The third and fourth results columns illustrate the role of weighting in targeting different marginal PAMEs under varying group sizes. Here we look at the effect of the majority-vote decision rule on another outcome, and one likely to be sensitive to the decision rule: the average time each woman spent in the deliberation. (As is apparent from the coefficients in columns 1 and 2, the marginal PAME is close to zero for the ``care issues'' outcome.)  We marginalize over the different gender compositions. From Table~\ref{table:karp-design1}, the number of groups is roughly uniform across the five gender compositions (sizes $m_k = k$ for $k = 1, \dots, 4$ women). The number of \emph{women} is not, because a composition-$k$ group contributes $k$ women: the 19 single-woman groups supply 19 women, while the 16 four-woman groups supply 64. The assignment is also not perfectly balanced across the unanimity and majority-vote arms. We consider the two estimands discussed in Subsection~\ref{estimands-rg}: the design-weighted PAME ($\phi(z; m_k) = \Pr(|\mathcal{A}_{A_i}| = m_k \mid Z_i = z) = N_z(m_k)/N_z$), which one obtains with an {\it unweighted} regression and the PAME that gives uniform weight to each size class ($\phi(z; m_k) = 1/|\mathcal{M}_z|$), which requires that one weight the regression by the inverse of the probability of appearing in any given size class.  Columns 3 and 4 of Table~\ref{tab:mendrep} show the consequences. The estimate for the PAME that gives equal weight to each size class is 17\% larger in absolute value than that for the design-weighted PAME. This reflects that treatment effects are of greater magnitude for women in groups with fewer other women (apparent from the interaction effect estimates in column 5), and that groups with fewer women proportionally less weight in the design-weighted estimand.

\subsection{Application II: improving management with individual and group-based consulting}
\citet{iacovone2022improving} conduct an experiment to test two approaches to improving poor management in the Colombian auto parts manufacturing sector. The first uses intensive, individualized consulting; the second is a group-based approach that delivers consulting at roughly a third of the cost, inspired by agricultural extension, and aims to leverage group-learning dynamics.

In the experiment, a sample of 159 firms was sorted into 53 matched triplets to improve balance on observables, and within each triplet one firm was randomly assigned to each of three arms: a control group, an individual-consulting treatment group, and a group-consulting treatment group, yielding 53 firms per arm. Firms assigned to the individual-consulting treatment received individual support over a period of six months. Treatment began with training in five areas (logistics, human resources, finance, marketing and sales, and production) delivered by five different consultants. This involved a theoretical part aimed at familiarizing the firm's management with modern management concepts and methods, complemented with practical exercises to apply these concepts to their firm, and was then followed by individual consulting to help the firm implement the improvement plan developed during the diagnostic phase.

In the group-consulting treatment, groups were formed with 2 to 8 firms. Leaders from the firms in a group signed an agreement to work together and help each other improve. As in the individual treatment, the group treatment began with training classes covering theoretical aspects of management, but delivered to the group in a classroom setting rather than one-on-one in the firm, and was then followed by group consulting sessions.

\begin{table}[!t]
    \centering
    \begin{threeparttable}
        \caption{Reanalyses of \citet{iacovone2022improving}}\label{tab:iacorep}
        {\footnotesize
\begin{tabular}{llllllll}
  \hline
   & Overall & Finance  & HR  & Logistics  & Marketing  & Production  \\ 
    &  Score &  Practices &  Practices &  Practices &  Practices &  Practices \\ 
  \hline
  Ind Treatment*During & 9.703 & 9.644 & 10.793 & 8.708 & 10.637 & 5.696 \\ 
    CR.SE Firm & (1.37)*** & (1.852)*** & (1.822)*** & (1.603)*** & (2.28)*** & (1.806)*** \\ 
  CR.SE Group & (1.38)*** & (1.848)*** & (1.814)*** & (1.66)*** & (2.292)*** & (1.814)*** \\ 
  Ind Treatment*Post & 9.62 & 9.712 & 8.974 & 8.585 & 9.451 & 8.488 \\ 
  CR.SE Firm & (1.83)*** & (2.413)*** & (2.508)*** & (2.457)*** & (2.466)*** & (1.993)*** \\ 
  CR.SE Group & (1.836)*** & (2.426)*** & (2.493)*** & (2.476)*** & (2.452)*** & (2.006)*** \\ 
  Group Treatment*During & 11.971 & 13.841 & 12.249 & 9.327 & 11.899 & 11.798 \\ 
  CR.SE Firm & (1.66)*** & (2.057)*** & (2.078)*** & (2.047)*** & (2.599)*** & (1.993)*** \\ 
  CR.SE Group & (1.941)*** & (2.231)*** & (2.109)*** & (2.284)*** & (2.89)*** & (2.403)*** \\ 
  Group Treatment*Post & 8.544 & 9.82 & 7.156 & 5.86 & 9.046 & 10.694 \\ 
  CR.SE Firm & (1.894)*** & (2.306)*** & (2.655)*** & (2.539)** & (2.637)*** & (2.048)*** \\ 
  CR.SE Group & (2.256)*** & (2.503)*** & (2.666)*** & (2.86)** & (2.939)*** & (2.577)*** \\ 
     \hline
    ICC(group) & 0.435*** & 0.405** & 0.358*** & 0.210** & 0.185  & 0.457*** \\ 
   \hline
\end{tabular}
}
        \begin{tablenotes}
            \footnotesize
            \item\textit{Notes:} Standard errors in parentheses. The first is clustered at the firm level as in the original paper, while the second is clustered at the groups induced by treatment. $^{***}$p $<$ 0.01; $^{**}$p $<$ 0.05; $^{*}$p $<$ 0.1. Significance stars on the ICC row denote results from the test for interference in Section~\ref{sec:test}.
        \end{tablenotes}
    \end{threeparttable}
\end{table}

The main outcome is what the original paper refers to as the ``Anexo K management score,'' the average across 141 management practices (each rated on a five-point scale), measured at three points: baseline, during the intervention (at the end of implementation), and post-intervention (at the one-year follow-up). Following the specification in the original paper, we replicate their main results using a saturated regression that interacts each of the two treatment indicators---individual and group consulting---with the during- and post-intervention periods, with control as the reference arm; the regression controls for the baseline management score, randomization-triplet fixed effects, and a post-period indicator. The estimates are reported in Table~\ref{tab:iacorep}. The four interaction coefficients---each treatment arm in each period---are the effects relative to control, and they show immediate and persistent positive effects on firm management.

For standard errors, \citet{iacovone2022improving} cluster at the firm level. Based on our theory, we instead compute CR2 standard errors at the group level: for firms in the group treatment we treat each consulting group as a single cluster, while for firms in the individual-treatment and control arms each firm forms its own cluster. These CR2 standard errors are shown in the second standard-error row of each cell. As expected, they are larger than the original ones, with a substantial increase in some cells.

We also calculate the ICC and test for interference in the group treatment, with results reported in the last row of Table~\ref{tab:iacorep}. The null of no interference is rejected at the 0.001 level, and the overall intraclass correlation exceeds 0.5. For nearly all submeasures (except Marketing Practices), the group treatment induces substantial intraclass correlation. This is intuitive, since firms within a group receive the same treatment and sign an agreement to work together and help each other.



\section{Conclusion}\label{sec:con}
We have developed design-based inference for group interaction experiments, in which individuals are allocated to groups that interact under group-level treatment conditions. Organizing the problem around two design cases (fixed versus randomly formed groups) and two outcome cases (no interference versus group interference), we characterized for each of the four cells the estimand that commonly used estimators target and the variance estimator that is valid for it. These estimators target the average treatment effect (ATE) when interference is absent, the total effect (TOT) under fixed groups with interference, and an average marginalized exposure effect (the PAME) under randomly formed groups with interference. The cluster-robust variance estimator (CR2) is consistent in all four cells; the heteroskedasticity-robust estimator (HR2) is consistent only when groups are randomly formed and interference is absent, where it is also the more efficient choice.

These results revise a familiar rule of thumb. Rather than clustering at the level of treatment assignment, as is commonly inferred from \citet{abadie2023should}, one should cluster at the level of the potential outcome inputs\footnote{We thank Jonne Kamphorst for suggesting this phrasing.}---the arguments that determine which potential outcome an assignment pattern reveals. Under SUTVA, a unit's own treatment is the only such input. Once interference is present, the inputs include other units' treatment or group assignments, and valid inference requires clustering at the level of the resulting interference neighborhood even when treatment is assigned individually. In group interaction experiments that neighborhood is the group, so clustering at the group level is required whenever interaction generates interference; because such experiments typically assign individuals to groups with the expectation of interaction between individuals, cluster-robust inference should be the default. The intraclass-correlation test we develop helps researchers determine which regime they face, and hence which variance estimator to use. Our reanalyses of \citet{mendelberg2014does} and \citet{iacovone2022improving} bear these points out: in both, the null of no interference is rejected and group-level cluster-robust standard errors are materially larger than the standard errors that do not account for group interaction.

Methodologically, our asymptotic results rest on a sparse-sampling regime and a coupling argument that recasts random group formation as sampling from a population of candidate groups, shows the no-overlap restriction to be asymptotically immaterial when $N^2/n \to 0$, and thereby transports a standard central limit theorem to this design. The same device may prove useful for design-based inference in other experiments whose realized units are sampled subsets of a larger population. 

Several extensions remain open for future research, and two seem particularly important. First, the analysis here suggests potential trade-offs that come from choosing a design with larger versus smaller groups. With a fixed sample size, one could have many small groups or a few larger groups. Having many smaller groups yields more independent clusters, however the group interaction may be more intense, and so the implications for estimator variance are ambiguous. A design with many small groups also targets a different estimand (marginalizing over small numbers of group-mates) than having fewer larger groups. Both statistical and substantive considerations are needed to determine which design is best.  Second, the analysis here takes groupmate composition to be a source of unmeasured heterogeneity, and the estimands we consider here marginalize over that heterogeneity. Another target of inference, and one that is the focus of papers like \citet{li2019randomization}, is the effect of measured variation in peer composition. The analysis here suggests ways that one might revisit the question of ``peer effects''.  

\section*{Statement of AI Usage} In writing this paper, the authors declare use of AI/LLM tools in the following: helping to develop asymptotic analyses and in particular helping with the generalization from the homogenous group size case to the varying group size case (Chat GPT 5 and Sol models), constructing tables and illustrative figures (Claude Opus models), identifying notational inconsistencies (Claude Opus and ChatGPT 5 and Sol models), drafting the Monte Carlo description (Claude Opus models), and in reviewing the full draft and noting areas for revision (Refine.Ink, Claude Opus, and ChatGPT Sol).  All outputs have been reviewed by the authors, although we are at present conducting continued verification. Updated drafts will contain any necessary corrections or revisions that we spot.

\clearpage

\begin{singlespace}      
\begin{appendix}

\section{Appendix}
\section*{Contents}
  \begin{description}
    \item[A.1\quad Homogeneous group sizes]\hfill
    \begin{description}
      \item[A.1.1] Set up
      \item[A.1.2] A group-level representation of the experimental design
    \end{description}

    \item[A.2\quad Varying group sizes]\hfill
    \begin{description}
      \item[A.2.1] Generalizing the estimands with varying group sizes
      \item[A.2.2] Translation to the group-level representation
      \item[A.2.3] Translating analytical results to the varying group size case
    \end{description}

    \item[A.3\quad Monte Carlo simulation design]\hfill
    \begin{description}
      \item[A.3.1] Simulation study 1: variance estimation and coverage
      \item[A.3.2] Simulation study 2: covariate adjustment
    \end{description}
\end{description}

\begin{table}[!h]
\centering
\caption{Notation used in the appendix}
\label{tab:notation}
{\small
\renewcommand{\arraystretch}{0.8}
\begin{tabular}{@{}ll@{}}
\hline
\textbf{Symbol} & \textbf{Definition} \\
\hline
\multicolumn{2}{@{}l@{}}{\it Population, sample, and groups} \\
$\mathcal{U}$, $\mathcal{N}$ & Population ($|\mathcal{U}|=n$) and sample ($|\mathcal{N}|=N$) \\
$R_i$, $r$ & Sampling indicator and probability: $r = N/n$ \\
$G_N$, $G_{N1}$, $G_{N0}$ & Total, treated, and control groups \\
$M$, $M_g$ & Group size (equal-size); group size for the $g$th group \\
$\mathcal{M}$, $\mathcal{M}_z$, $m_k$ & Collection of admissible group sizes; those in arm $z$; the $k$th such size \\
$\mathcal{A}_g \subseteq \mathcal{U}$ & Set of units in group $g$, $|\mathcal{A}_g| = M_g$ \\
$A_i$ & Group index: $A_i = g \iff i \in \mathcal{A}_g$ \\
$\mathcal{A}_{A_i}$ & Group set for unit $i$ (incl.\ $i$), $|\mathcal{A}_{A_i}| = M_{A_i}$ \\
$\mathcal{G}_{\mathcal{U}}(M)$ & All size-$M$ subsets of $\mathcal{U}$; $|\mathcal{G}_{\mathcal{U}}(M)| = \binom{n}{M}$ \\
$\mathcal{G}_{\mathcal{U},i}(M)$ & $\{\mathcal{S} \subseteq \mathcal{U} : i \in \mathcal{S},\, |\mathcal{S}| = M\}$; size $\binom{n-1}{M-1}$ \\
\hline
\multicolumn{2}{@{}l@{}}{\it Treatment} \\
$Z_g$, $Z_i = Z_{A_i}$ & Group and individual treatment status \\
$p$ & $P(Z_g = 1) = G_{N1}/G_N$ \\
\hline
\multicolumn{2}{@{}l@{}}{\it Potential outcomes and estimands} \\
$Y_i(z, \mathcal{A})$ & Potential outcome for $i$ under treatment $z$, group $\mathcal{A} \ni i$ \\
$\mu_i(z)$ & Marginalized potential outcome: $\binom{n-1}{M-1}^{-1} \!\sum_{\mathcal{A} \in \mathcal{G}_{\mathcal{U},i}} Y_i(z, \mathcal{A})$ \\
$\tau_i$, $\tau_{PAME}$ & $\mu_i(1) - \mu_i(0)$;\; $n^{-1}\sum_{i \in \mathcal{U}} \tau_i$ \\
$\bar{Y}_g(z)$, $\tau_g$ & Group mean potential outcome; group effect $\bar{Y}_g(1) - \bar{Y}_g(0)$ \\
$\tau_{SAME}$ & Sample average marginalized effect: $\frac{1}{G_N} \sum_{g=1}^{G_N} \tau_{g}$ \\
\hline
\multicolumn{2}{@{}l@{}}{\it Group-level design representation} \\
$\omega$, $\mathcal{A}_\omega$ & Index and membership of potential group $\omega$ \\
$W_\omega$ & $W_\omega = 1$ if group $\omega$ is realized \\
$\pi_N$ & Group sampling prob.: $G_N / |\mathcal{G}_{\mathcal{U}}(M)|$ \\
$\bar{Y}_\omega(z)$, $\tau_\omega$ & Mean potential outcome and effect for potential group $\omega$ \\
$\bar{Y}_N(z)$, $\bar{Y}_I(z)$ & $\tfrac{1}{N}\sum_{i=1}^N Y_i\left(z, \mathcal{A}_g\right)$; $\frac{1}{n} \sum_{i = 1}^n Y_i(z)$ (no interference)  \\
$\bar{Y}(z)$, $\widehat{\bar{Y}}(z)$ & Mean potential outcome under $z$ across potential groups; its estimate  \\
$\check{\tau}$ & $(|\mathcal{G}_{\mathcal{U}}(M)|\pi_N)^{-1} \sum_\omega W_\omega \tau_\omega$ \\
$\sigma_\tau^2$ & $|\mathcal{G}_{\mathcal{U}}(M)|^{-1} \sum_\omega (\tau_\omega - \tau_{PAME})^2$ \\
\hline
\multicolumn{2}{@{}l@{}}{\it Estimator and asymptotics} \\
$\hat{\tau}_N$ & A commonly used estimator \\
$a_N \simeq b_N$ & $\lim_{n,N \to \infty} \frac{a_N}{b_N} = 1$ under the given asymptotic regime \\
$a_N \asymp b_N$ & $a_N = O(b_N)$ and $b_N = O(a_N)$ under the given asymptotic regime \\
$\rightarrow$, $\xrightarrow{P}$, $\rightsquigarrow$ & Convergence of a sequence, in probability, and in distribution. \\
\hline
\end{tabular}}
\end{table}

\clearpage

This appendix is self-contained. We begin by restating the inferential setting and notation, then establish the results formally. We develop the analysis for the Design Case 1 (random group formation) and Outcome Case 1 (group interference) combination introduced in the main text. Results for the other design and outcome cases appear as special cases or with reference to existing results (e.g., for cluster-randomized designs without interference). 

We have a population, $\mathcal{U}$, with $n = |\mathcal{U}|$ units. From the population $\mathcal{U}$, we randomly draw a sample $\mathcal{N}$ with $N$ units using simple random sampling (SRS). The sampling indicator is denoted as $R_i$, where $R_i = 1$ if unit $i \in \mathcal{U}$ is in the sample and $R_i = 0$ otherwise. We use $r$ to denote the probability for unit $i$ to be included in the sample: $r = P(R_i = 1) = \frac{N}{n}$. Next, we randomly partition the $N$ units into $G_N$ groups based on the group index permutation approach described in the main text. 

\subsection{Homogeneous group sizes}
We start from the simple scenario where all groups have the same size, $M = \frac{N}{G_N}$. This case allows for a simpler and more intuitive presentation of key results. We generalize to varying group sizes in the subsection that follows. We construct $G_N$ randomly formed groups, indexed by $g=1,...,G_N$, by random partition. The partition is created through a random permutation of an index vector of length $N$ of the form $\underbrace{1,...,1}_{M},...,\underbrace{g,...,g}_{M},...,\underbrace{G_N,...,G_N}_{M}$, in which each index value $g$ appears $M$ times, so that each unit is randomly assigned to a group. Groups $g=1,...,G_{N1}$ are then assigned to ``treatment'' and $g=G_{N1}+1,...,G_N$ to ``control,'' with group treatment indicator $Z_g \in \{0,1\}$ equal to $1$ for treated groups and $0$ for control groups. This is Design Case 1 in the main text with homogeneous group sizes.

\subsubsection{Set up}
Let $\mathcal{A}_g \subseteq \mathcal{U}$ denote the set of units in group $g$, with $|\mathcal{A}_g| = M$ for all $g \in \{1, \dots, G_N\}$, and write $A_i = g$ if and only if $i \in \mathcal{A}_g$. The collection of all possible groups with $M$ members that can be formed from $\mathcal{U}$ is denoted as $\mathcal{G}_{\mathcal{U}}(M)$, with $|\mathcal{G}_{\mathcal{U}}(M)| = \binom{n}{M}$. As decribed in the main text, $G_{N1}$ out of $G_N$, are selected into treatment, with the remaining $G_{N0} = G_N - G_{N1}$ groups selected into control. Let $Z_g$ denote the treatment status of group $g$. Conditional on the units included in the sample, $P(Z_g = 1) = p = \frac{G_{N1}}{G_N}$. For any unit $i \in \mathcal{A}_g$, the treatment status is given by $Z_i = Z_g = Z_{A_i}$.

At the end of the experiment, we observe the outcome $Y_i$ for each unit $i \in \mathcal{N}$. We allow $Y_i$ to depend not only on the treatment status of the group to which unit $i$ belongs but also on the composition of that group, $\mathcal{A}_{A_i}$. Accordingly, we write
$$
Y_{i} = Y_{i}\left(Z_{A_i}, \mathcal{A}_{A_i}\right).
$$
Given the random sampling of units and the random group formation, $\mathcal{A}_{A_i}$ can be any size-$M$ subset of $\mathcal{U}$ containing $i$, with probability $P(\mathcal{A}_{A_i} = \mathcal{A}) = 1/{n - 1 \choose M - 1}$ for any $\mathcal{A} \in \mathcal{G}_{\mathcal{U},i}(M)$, where $\mathcal{G}_{\mathcal{U},i}(M) = \{\mathcal{S} \subseteq \mathcal{U} : i \in \mathcal{S},\; |\mathcal{S}| = M\}$.

For any unit $i$, we can marginalize $Y_i$ over the distribution of $\mathcal{A}_{A_i}$ to obtain its individualistic marginalized potential outcome:
$$
\mu_i\left(z\right) = \E\left[Y_{i}\left(z, \mathcal{A}_{A_i}\right)\right] = \frac{1}{{n - 1 \choose M - 1}}\sum_{\mathcal{A} \in \mathcal{G}_{\mathcal{U},i}(M)} Y_{i}\left(z, \mathcal{A}\right).
$$
The individualistic marginalized effect is then defined as 
$$
\tau_i = \mu_i(1) - \mu_i(0).
$$
Our estimand of interest is the population average marginalized effect (PAME) in the population:
$$
\tau_{PAME} = \frac{1}{n}\sum_{i \in \mathcal{U}}\tau_i.
$$

We estimate $\tau_{PAME}$ using the difference-in-means (DIM) estimator:
$$
\begin{aligned}
\hat \tau_N = & \frac{\sum_{i \in \mathcal{N}} Z_iY_i}{G_{N1}M} - \frac{\sum_{i \in \mathcal{N}} (1 - Z_i)Y_i}{G_{N0}M}.
\end{aligned}
$$
This is equivalent to $Z_i$'s coefficient estimate from the OLS estimator, where we regress $Y_i$ on an intercept and $Z_i$ \citep{samii2012equivalencies}. When the probability of treatment differs, we can replace the DIM estimator with the more general H\'ajek estimator:
$$
\begin{aligned}
\hat \tau_N = & \frac{\sum_{i \in \mathcal{N}} Z_iY_i / p_i}{\sum_{i \in \mathcal{N}} Z_i / p_i} - \frac{\sum_{i \in \mathcal{N}} (1 - Z_i)Y_i / (1 - p_i)}{\sum_{i \in \mathcal{N}} (1 - Z_i) / (1 - p_i)}.
\end{aligned}
$$
It can also be implemented through a regression of $Y_i$ on an intercept and $Z_i$, where each observation is reweighted by $\frac{Z_i}{p_i} + \frac{1 - Z_i}{1 - p_i}$. We can further improve the efficiency of the estimator by incorporating demeaned covariates and their interactions with $Z_i$ in the regression model \citep{lin2013agnostic}.

To establish the large-sample properties of $\hat \tau_N$, we consider an asymptotic regime where the group size $M$ remains fixed, $n, N, G_N \to \infty$, and  $\frac{N^2}{n} \to 0$.

\subsubsection{A group-level representation of the experimental design}

We now demonstrate that the random group formation experimental design of interest can be represented equivalently as a design that randomly samples, in a restricted way, from the population of ``potential groups'' that the design allows to be formed.  The group level representation greatly facilitates the analysis.

\begin{description}
    \item[Design Case 1G:] First, a sample of $G_N$ groups are drawn via restricted random sampling from the population of ``potential groups,'' $\mathcal{G}_{\mathcal{U}}(M)$, where $|\mathcal{G}_{\mathcal{U}}(M)| = \binom{n}{M}$.  The restricted sampling design proceeds in sequence without replacement, such that the first group, $\mathcal{A}_1$ is drawn at random from $\mathcal{G}_{\mathcal{U}}(M)$, the second group, $\mathcal{A}_2$ is drawn from the set of groups that do not overlap with $\mathcal{A}_1$, $\mathcal{G}_{\mathcal{U}}(M)\setminus \{\tilde{A} :  \mathcal{A}_1 \cap \tilde{A} \neq \emptyset\}$, the third group, $\mathcal{A}_3$ is drawn from the set of groups that overlap with neither $\mathcal{A}_1$ nor $\mathcal{A}_2$, $\mathcal{G}_{\mathcal{U}}(M)\setminus \{\tilde{A} :  \mathcal{A}_1 \cap \tilde{A} \neq \emptyset \text{ or } \mathcal{A}_2 \cap \tilde{A} \neq \emptyset \} $, and so on.  From the $G_N$ sampled groups, we randomly assign $G_{N1}$ to treatment and $G_{N0} = G_N - G_{N1}$ to control.
\end{description}
This design is equivalent to a classical randomized trial implemented at the group level, but with a restrictive sampling scheme that does not allow potential groups with overlapping membership to be selected together. 

For any realized partition $a=(\mathcal{A}_1,\mathcal{A}_2,...,\mathcal{A}_{G_N})$, we use $\mathbb{P}_I(a)$ and $\mathbb{P}_{G}(a)$ to denote the probability of the realized partition $a$. The following proposition shows that the restricted group sampling and group-treatment representation is equivalent to Design Case 1 in the sense that they induce the same probability distribution on realized groups $(\mathcal{A}_1,\mathcal{A}_2,...,\mathcal{A}_{G_N})$.

\begin{lemma}
	For any realized partition $a=(\mathcal{A}_1,\mathcal{A}_2,...,\mathcal{A}_{G_N})$, both designs have the same probability:
	$$
	\mathbb{P}_I(a)=\mathbb{P}_G(a)=\frac{(n-N)!(M!)^{G_N}}{n!}
	$$
\end{lemma} The proof is in the section \ref{sec:scheme}, where we generalize to the case of varying group sizes.

For any group $\omega \in \{1,2,\dots,|\mathcal{G}_{\mathcal{U}}(M)|\}$, let $W_{\omega} = 1$ if the $\omega$th group is sampled and $W_{\omega} = 0$ otherwise. Then, as a representation of the random sampling and random group formation process defined above, we would have the following:
\begin{lemma}\label{lemma:group-sampling}
Given Design Case 1G, we have
$$
\pi_N := \E[W_{\omega}] = P(W_{\omega} = 1) = \frac{G_N}{|\mathcal{G}_{\mathcal{U}}(M)|},$$ 
$$
Var[W_{\omega}] = \pi_N(1-\pi_N), \text{ and }
$$
$$
E[W_{\omega} W_{\omega'}] = \begin{cases}\pi_N * \frac{G_N-1}{|\mathcal{G}_{\mathcal{U}}(M)\setminus \mathcal{A}_{\omega}|} & \text{if } \mathcal{A}_{\omega} \cap \mathcal{A}_{\omega'} = \emptyset \\ 0 & \text{if } \mathcal{A}_{\omega} \cap \mathcal{A}_{\omega'} \neq \emptyset  \end{cases} \text{ for } \omega \ne \omega',
$$
where $\mathcal{G}_{\mathcal{U}}(M)\setminus \mathcal{A}_{\omega} = \{\mathcal{A}_{\omega'}: \mathcal{A}_{\omega'} \in \mathcal{G}_{\mathcal{U}}(M), \mathcal{A}_{\omega} \cap \mathcal{A}_{\omega'} = \emptyset\}$ denotes the collection of groups in $\mathcal{G}_{\mathcal{U}}(M)$ that do not overlap with $\mathcal{A}_{\omega}$, hence $|\mathcal{G}_{\mathcal{U}}(M)\setminus \mathcal{A}_{\omega}| = \binom{n-M}{M}$. As a result,
$$
Cov[W_{\omega}, W_{\omega'}] = \begin{cases}\pi_N\left(\frac{G_N-1}{|\mathcal{G}_{\mathcal{U}}(M)\setminus \mathcal{A}_{\omega}|} - \pi_N\right) & \text{if } \mathcal{A}_{\omega} \cap \mathcal{A}_{\omega'} = \emptyset \\ -\pi_N^2 & \text{if } \mathcal{A}_{\omega} \cap \mathcal{A}_{\omega'} \neq \emptyset  \end{cases}.
$$
\end{lemma}
The restriction on the joint inclusion probabilities of overlapping potential groups is the key restriction relative to uniform random sampling from the set of possible groups.  

As all members of a group share the same treatment status and group composition, $Y_{i} = Y_{i}\left(Z_g, \mathcal{A}_{g}\right)$, and we can define the group-level average outcome: 
$$
\bar Y_g = \frac{1}{M}\sum_{i \in \mathcal{A}_g} Y_{i} = \frac{1}{M}\sum_{i \in \mathcal{A}_g} Y_{i}\left(Z_g, \mathcal{A}_{g}\right) = \bar Y_g\left(Z_g, \mathcal{A}_{g}\right) = \bar Y_g\left(Z_g\right).
$$
The group-level treatment effect can thus be defined as
$$
\tau_g = \bar Y_g(1) - \bar Y_g(0) = \frac{1}{M}\sum_{i \in \mathcal{A}_g}\left[Y_{i}\left(1, \mathcal{A}_{g}\right) - Y_{i}\left(0, \mathcal{A}_{g}\right)\right].
$$
This restricted group sampling and group-treatment design imply that the following hold:
\begin{lemma}\label{lemma:group-po}
For any $z \in \{0,1\}$ and $\omega, \omega' \in \{1,...,|\mathcal{G}_{\mathcal{U}}(M)|\}$, we have 
\begin{enumerate}
\item (random group sampling) $\bar Y_{\omega}(z) \independent W_{\omega'}$,
\item (random treatment assignment) $\bar Y_{\omega}(z) \independent Z_{\omega'} \mid W_{\omega'} = 1$.
\end{enumerate}
\end{lemma}

Under Assumption 1 from the main text and given Design Case 1G, we can show that the following lemmas hold:
\begin{lemma}[$\tau_{PAME}$ as difference in group-level potential outcomes]\label{lemma:group-ame}
$$
\tau_{PAME} = \frac{1}{|\mathcal{G}_{\mathcal{U}}(M)|} \sum_{\omega = 1}^{|\mathcal{G}_{\mathcal{U}}(M)|} \tau_{\omega} = \frac{1}{|\mathcal{G}_{\mathcal{U}}(M)|} \sum_{\omega = 1}^{|\mathcal{G}_{\mathcal{U}}(M)|}\left[\bar Y_{\omega}(1) - \bar Y_{\omega}(0)\right].
$$
\end{lemma}
\begin{proof}
From its definition, we know that
$$
\begin{aligned}
\tau_{PAME} = & \frac{1}{n}\sum_{i \in \mathcal{U}}\frac{1}{{n - 1 \choose M - 1}}\sum_{\mathcal{A} \in \mathcal{G}_{\mathcal{U},i}(M)} \left[Y_{i}\left(1, \mathcal{A}\right) - Y_{i}\left(0, \mathcal{A}\right)\right] \\
= & \frac{1}{n{n - 1 \choose M - 1}}\sum_{i \in \mathcal{U}}\sum_{\mathcal{A} \in \mathcal{G}_{\mathcal{U},i}(M)} \left[Y_{i}\left(1, \mathcal{A}\right) - Y_{i}\left(0, \mathcal{A}\right)\right] \\
= & \frac{M}{n{n - 1 \choose M - 1}}\sum_{\mathcal{A} \in \mathcal{G}_{\mathcal{U}}(M)}\frac{1}{M}\sum_{i \in \mathcal{A}}\left[Y_{i}\left(1, \mathcal{A}\right) - Y_{i}\left(0, \mathcal{A}\right)\right] \\
= & \frac{M}{n{n - 1 \choose M - 1}}\sum_{\omega = 1}^{| \mathcal{G}_{\mathcal{U}}(M)|}\tau_{\omega} = \frac{1}{|\mathcal{G}_{\mathcal{U}}(M)|}\sum_{\omega = 1}^{| \mathcal{G}_{\mathcal{U}}(M)|}\tau_{\omega}.
\end{aligned}
$$
\end{proof}

\begin{lemma}[$\hat \tau_N$ as difference in group means]\label{lemma:group-tau}
$$
\begin{aligned}
\hat \tau_N = & \frac{\sum_{g=1}^{G_N}Z_g\bar Y_g}{G_{N1}} - \frac{\sum_{g=1}^{G_N}(1 - Z_g)\bar Y_g}{G_{N0}} \\
= & \frac{1}{|\mathcal{G}_{\mathcal{U}}(M)|}\sum_{\omega = 1}^{| \mathcal{G}_{\mathcal{U}}(M)|}\frac{W_{\omega}Z_{\omega}\bar Y_{\omega}}{\pi_Np} - \frac{1}{|\mathcal{G}_{\mathcal{U}}(M)|}\sum_{\omega = 1}^{| \mathcal{G}_{\mathcal{U}}(M)|}\frac{W_{\omega}(1 - Z_{\omega})\bar Y_{\omega}}{\pi_N(1 - p)}.
\end{aligned}
$$
\end{lemma}
\begin{proof}
Follows from straightforward algebra.
\end{proof}

For the analysis that follows, it is convenient to introduce a sample-level analog of $\tau_{PAME}$ that we call the \emph{sample average marginalized effect} (SAME). The SAME plays the same role here as the sample average treatment effect (SATE) plays relative to the population average treatment effect (PATE) in analyses under SUTVA: $\tau_{PAME}$ is a population-level quantity defined over all possible groups in $\mathcal{G}_\mathcal{U}(M)$ and all units in $\mathcal{U}$, while the SAME is its analog computed over the realized sample of groups, treating each sampled group's full-data potential outcomes $\bar Y_\omega(1)$ and $\bar Y_\omega(0)$ as known: 
$$
\begin{aligned}
\tau_{SAME} = & \frac{1}{G_N} \sum_{g=1}^{G_N(M)|} \tau_{g} = \frac{1}{|\mathcal{G}_{\mathcal{U}}(M)|\pi_N} \sum_{\omega=1}^{|\mathcal{G}_{\mathcal{U}}(M)|}W_{\omega} \tau_{\omega} \\
= & \frac{1}{N}\sum_{i=1}^N \left(Y_i\left(1, \mathcal{A}_{A_i}\right) - Y_i\left(0, \mathcal{A}_{A_i}\right)\right)  = \bar Y_N(1) - \bar Y_N(0),
\end{aligned}
$$
where $\bar Y_N(z) = \tfrac{1}{G_N}\sum_{g=1}^{G_N} \bar Y_g\left(z\right) = \tfrac{1}{N}\sum_{i=1}^N Y_i\left(z, \mathcal{A}_{A_i}\right)$. The SAME is not observable in practice---we cannot observe both potential outcomes for any group---but it serves as an analytical bridge: it isolates the variability arising from group sampling from the variability arising from treatment assignment within sampled groups, just as the SATE--PATE decomposition isolates sampling variability in standard finite-population causal inference. 

\begin{lemma}[$\hat \tau_N$ is an unbiased estimator for $\tau_{PAME}$]\label{lemma:group-tau}
Under Assumption 3,
$$
\begin{aligned}
\E\left[\hat \tau_N\right] = & \tau_{PAME}.
\end{aligned}
$$
\end{lemma}
\begin{proof}
Conditional on group sampling, randomness in $\hat \tau_N$ comes solely from treatment assignment. Thus, it resembles the classical DIM estimator at the group level. Lemma~\ref{lemma:group-po} suggests that its conditional expectation equals
$$
\begin{aligned}
\E\left[\hat \tau_N \mid \{W_{\omega}\}_{\omega = 1}^{|\mathcal{G}_{\mathcal{U}}(M)|}\right] = & \frac{1}{|\mathcal{G}_{\mathcal{U}}(M)|}\sum_{\omega = 1}^{| \mathcal{G}_{\mathcal{U}}(M)|}\frac{W_{\omega}}{\pi_N}\E\left[\frac{Z_{\omega}\bar Y_{\omega}}{p}\mid \{W_{\omega}\}_{\omega = 1}^{|\mathcal{G}_{\mathcal{U}}(M)|}\right] \\
& - \frac{1}{|\mathcal{G}_{\mathcal{U}}(M)|}\sum_{\omega = 1}^{| \mathcal{G}_{\mathcal{U}}(M)|}\frac{W_{\omega}}{\pi_N}\E\left[\frac{(1 - Z_{\omega})\bar Y_{\omega}}{1 - p}\mid \{W_{\omega}\}_{\omega = 1}^{|\mathcal{G}_{\mathcal{U}}(M)|}\right] \\
= & \frac{1}{|\mathcal{G}_{\mathcal{U}}(M)|}\sum_{\omega = 1}^{| \mathcal{G}_{\mathcal{U}}(M)|}\frac{W_{\omega}\tau_{\omega}}{\pi_N} = \frac{1}{G_N}\sum_{g = 1}^{G_N}\tau_g = \tau_{SAME}
\end{aligned}
$$
Taking expectations over group sampling yields
$$
\E\left[\E\left[\hat \tau_N \mid \{W_{\omega}\}_{\omega = 1}^{|\mathcal{G}_{\mathcal{U}}(M)|}\right]\right] = \frac{1}{|\mathcal{G}_{\mathcal{U}}(M)|}\sum_{\omega = 1}^{| \mathcal{G}_{\mathcal{U}}(M)|}\frac{\E[W_{\omega}]\tau_{\omega}}{\pi_N} = \frac{1}{|\mathcal{G}_{\mathcal{U}}(M)|}\sum_{\omega = 1}^{| \mathcal{G}_{\mathcal{U}}(M)|}\tau_{\omega} = \tau_{PAME}.
$$
\end{proof}

We use $a_N \simeq b_N$ to denote that $a_N$ and $b_N$ converge to the same limit and $a_N \asymp b_N$ to denote that they converge to limits of the same order as $n, N \to \infty$. Next, we establish the limits of several terms that will appear later.
\begin{lemma}[Several useful limits]\label{lemma:limits}
For any $\tau_{\omega}$ such that $|\tau_{\omega}| < C$,
$$
\begin{aligned}
& \frac{\left(1 - G_N\right)}{|\mathcal{G}_{\mathcal{U}}(M)|^2} \sum_{\omega=1}^{|\mathcal{G}_{\mathcal{U}}(M)|}\sum_{\mathcal{A}_{\omega} \cap \mathcal{A}_{\omega'} \neq \emptyset}\tau_{\omega} \tau_{\omega'} \to 0, \\
& \frac{\left(1 - G_N\right)}{|\mathcal{G}_{\mathcal{U}}(M)|}\left(\frac{1}{|\mathcal{G}_{\mathcal{U}}(M)|} -  \frac{1}{|\mathcal{G}_{\mathcal{U}}(M)\setminus \mathcal{A}_{\omega}|}\right) \sum_{\omega=1}^{|\mathcal{G}_{\mathcal{U}}(M)|}\sum_{\mathcal{A}_{\omega} \cap \mathcal{A}_{\omega'} = \emptyset}\tau_{\omega} \tau_{\omega'} \to 0
\end{aligned}
$$
as $n, N$ grow to infinity under our asymptotic regime.
\end{lemma}

\begin{proof}
First note that 
$$
\begin{aligned}
\frac{|\mathcal{G}_{\mathcal{U}}(M)\setminus \mathcal{A}_{\omega}|}{|\mathcal{G}_{\mathcal{U}}(M)|} = \frac{{n-M\choose M}}{{n\choose M}} = \prod_{k=0}^{M-1} \left(1 - \frac{M}{n - k}\right) \simeq \left(1 - \frac{M}{n}\right)^M \simeq 1 - \frac{M^2}{n} \to 1.
\end{aligned}
$$
Thus, 
$$
\frac{|\mathcal{G}_{\mathcal{U}}(M)| - |\mathcal{G}_{\mathcal{U}}(M)\setminus \mathcal{A}_{\omega}|}{|\mathcal{G}_{\mathcal{U}}(M)|} \simeq \frac{M^2}{n} \to 0,
$$
and
$$
\begin{aligned}
& \frac{\left(1 - G_N\right)}{|\mathcal{G}_{\mathcal{U}}(M)|^2} \sum_{\omega=1}^{|\mathcal{G}_{\mathcal{U}}(M)|}\sum_{\mathcal{A}_{\omega} \cap \mathcal{A}_{\omega'} \neq \emptyset}\tau_{\omega} \tau_{\omega'} \asymp \frac{\left(1 - G_N\right)}{|\mathcal{G}_{\mathcal{U}}(M)|^2} \sum_{\omega=1}^{|\mathcal{G}_{\mathcal{U}}(M)|}\sum_{\mathcal{A}_{\omega} \cap \mathcal{A}_{\omega'} \neq \emptyset} C^2 \\
= & \frac{|\mathcal{G}_{\mathcal{U}}(M)| - |\mathcal{G}_{\mathcal{U}}(M)\setminus \mathcal{A}_{\omega}|}{|\mathcal{G}_{\mathcal{U}}(M)|}\left(1 - G_N\right)C^2 \simeq \frac{M^2\left(1 - G_N\right)C^2}{n} \to 0, \\
& \frac{\left(1 - G_N\right)}{|\mathcal{G}_{\mathcal{U}}(M)|}\left(\frac{1}{|\mathcal{G}_{\mathcal{U}}(M)|} -  \frac{1}{|\mathcal{G}_{\mathcal{U}}(M)\setminus \mathcal{A}_{\omega}|}\right) \sum_{\omega=1}^{|\mathcal{G}_{\mathcal{U}}(M)|}\sum_{\mathcal{A}_{\omega} \cap \mathcal{A}_{\omega'} = \emptyset}\tau_{\omega} \tau_{\omega'} \\
\simeq & - \frac{\left(1 - G_N\right)}{|\mathcal{G}_{\mathcal{U}}(M)||\mathcal{G}_{\mathcal{U}}(M)\setminus \mathcal{A}_{\omega}|}\frac{M^2}{n}\sum_{\omega=1}^{|\mathcal{G}_{\mathcal{U}}(M)|}\sum_{\mathcal{A}_{\omega} \cap \mathcal{A}_{\omega'} = \emptyset}\tau_{\omega} \tau_{\omega'} \\
\asymp & - \frac{\left(1 - G_N\right)}{|\mathcal{G}_{\mathcal{U}}(M)||\mathcal{G}_{\mathcal{U}}(M)\setminus \mathcal{A}_{\omega}|}\frac{M^2}{n}\sum_{\omega=1}^{|\mathcal{G}_{\mathcal{U}}(M)|}\sum_{\mathcal{A}_{\omega} \cap \mathcal{A}_{\omega'} = \emptyset}C^2 \\
\simeq & - \frac{M^2\left(1 - G_N\right)C^2}{n} \to 0.
\end{aligned}
$$
\end{proof}

Now, we are ready to derive the asymptotic variance of $\hat \tau_N$. We start from the variance from sampling groups from $\mathcal{G}_{\mathcal{U}}(M)$. We derive the result for a general random variable $\tilde \tau_{\omega}$ defined for group $\omega$.
\begin{lemma}[Variance from group sampling]\label{lemma:var-g}
For any $\tilde \tau_{\omega}$ defined for group $\omega$ such that $|\tilde \tau_{\omega}| < C$, $\check \tau = \frac{1}{|\mathcal{G}_{\mathcal{U}}(M)|\pi_N} \sum_{\omega=1}^{|\mathcal{G}_{\mathcal{U}}(M)|}W_{\omega} \tilde \tau_{\omega}$, and $\tau = \frac{1}{|\mathcal{G}_{\mathcal{U}}(M)|} \sum_{\omega=1}^{|\mathcal{G}_{\mathcal{U}}(M)|}\tilde \tau_{\omega}$,
$$
\begin{aligned}
G_N\Var\left[\check \tau\right] \simeq \frac{1}{|\mathcal{G}_{\mathcal{U}}(M)|} \sum_{\omega=1}^{|\mathcal{G}_{\mathcal{U}}(M)|}\left(\tilde\tau_{\omega} - \tau\right)^{2} = \sigma_{\tau}^2
\end{aligned}
$$
and $\check \tau \xrightarrow{P} \tau$, as $n, N$ grow to infinity under our asymptotic regime.
\end{lemma}
\begin{proof}
First note that
$$
\begin{aligned}
& \tau^2 = \frac{1}{|\mathcal{G}_{\mathcal{U}}(M)|^2} \sum_{\omega=1}^{|\mathcal{G}_{\mathcal{U}}(M)|}\tilde\tau_{\omega}^{2} + \frac{1}{|\mathcal{G}_{\mathcal{U}}(M)|^2} \sum_{\omega=1}^{|\mathcal{G}_{\mathcal{U}}(M)|}\sum_{\mathcal{A}_{\omega} \cap \mathcal{A}_{\omega'} \neq \emptyset}\tilde\tau_{\omega} \tilde\tau_{\omega'} + \frac{1}{|\mathcal{G}_{\mathcal{U}}(M)|^2}\sum_{\omega=1}^{|\mathcal{G}_{\mathcal{U}}(M)|}\sum_{\mathcal{A}_{\omega} \cap \mathcal{A}_{\omega'} = \emptyset}\tilde\tau_{\omega} \tilde\tau_{\omega'}.
\end{aligned}
$$

Using Lemma~\ref{lemma:group-po} and Lemma~\ref{lemma:limits}, we can see that
$$
\begin{aligned}
& G_N\Var\left[\check \tau\right] =  \frac{1}{G_N} \sum_{\omega=1}^{|\mathcal{G}_{\mathcal{U}}(M)|}\tilde\tau_{\omega}^2\Var\left[ W_{\omega}\right] + \frac{1}{G_N} \sum_{\omega=1}^{|\mathcal{G}_{\mathcal{U}}(M)|}\sum_{\omega' \neq \omega}\tilde\tau_{\omega} \tilde\tau_{\omega'}\Cov\left[ W_{\omega}, W_{\omega'}\right] \\
= & \frac{\pi_N(1-\pi_N)}{G_N} \sum_{\omega=1}^{|\mathcal{G}_{\mathcal{U}}(M)|}\tilde\tau_{\omega}^2 + \frac{\pi_N\left(\frac{G_N-1}{|\mathcal{G}_{\mathcal{U}}(M)\setminus \mathcal{A}_{\omega}|} - \pi_N\right)}{G_N} \sum_{\omega=1}^{|\mathcal{G}_{\mathcal{U}}(M)|}\sum_{\mathcal{A}_{\omega} \cap \mathcal{A}_{\omega'} = \emptyset}\tilde\tau_{\omega} \tilde\tau_{\omega'} \\
& - \frac{\pi_N^2}{G_N} \sum_{\omega=1}^{|\mathcal{G}_{\mathcal{U}}(M)|}\sum_{\mathcal{A}_{\omega} \cap \mathcal{A}_{\omega'} \neq \emptyset}\tilde\tau_{\omega} \tilde\tau_{\omega'} \\
= & \left(\frac{1}{|\mathcal{G}_{\mathcal{U}}(M)|} - \frac{G_N}{|\mathcal{G}_{\mathcal{U}}(M)|^2}\right) \sum_{\omega=1}^{|\mathcal{G}_{\mathcal{U}}(M)|}\tilde\tau_{\omega}^{2} - \frac{G_N}{|\mathcal{G}_{\mathcal{U}}(M)|^2} \sum_{\omega=1}^{|\mathcal{G}_{\mathcal{U}}(M)|}\sum_{\mathcal{A}_{\omega} \cap \mathcal{A}_{\omega'} \neq \emptyset}\tilde\tau_{\omega} \tilde\tau_{\omega'} \\
& + \left(\frac{G_N-1}{|\mathcal{G}_{\mathcal{U}}(M)||\mathcal{G}_{\mathcal{U}}(M)\setminus \mathcal{A}_{\omega}|} - \frac{G_N}{|\mathcal{G}_{\mathcal{U}}(M)|^2}\right) \sum_{\omega=1}^{|\mathcal{G}_{\mathcal{U}}(M)|}\sum_{\mathcal{A}_{\omega} \cap \mathcal{A}_{\omega'} = \emptyset}\tilde\tau_{\omega} \tilde\tau_{\omega'}\\
= & \frac{1}{|\mathcal{G}_{\mathcal{U}}(M)|} \sum_{\omega=1}^{|\mathcal{G}_{\mathcal{U}}(M)|}\tilde\tau_{\omega}^{2} + \frac{G_N-1}{|\mathcal{G}_{\mathcal{U}}(M)||\mathcal{G}_{\mathcal{U}}(M)\setminus \mathcal{A}_{\omega}|} \sum_{\omega=1}^{|\mathcal{G}_{\mathcal{U}}(M)|}\sum_{\mathcal{A}_{\omega} \cap \mathcal{A}_{\omega'} = \emptyset}\tilde\tau_{\omega} \tilde\tau_{\omega'} - G_N \tau^2 \\
= & \frac{1}{|\mathcal{G}_{\mathcal{U}}(M)|} \sum_{\omega=1}^{|\mathcal{G}_{\mathcal{U}}(M)|}\left(\tilde\tau_{\omega} - \tau\right)^{2} + \left(1 - G_N\right)\left[\tau^2 -  \frac{1}{|\mathcal{G}_{\mathcal{U}}(M)||\mathcal{G}_{\mathcal{U}}(M)\setminus \mathcal{A}_{\omega}|} \sum_{\omega=1}^{|\mathcal{G}_{\mathcal{U}}(M)|}\sum_{\mathcal{A}_{\omega} \cap \mathcal{A}_{\omega'} = \emptyset}\tilde\tau_{\omega} \tilde\tau_{\omega'}\right] \\
= & \frac{1}{|\mathcal{G}_{\mathcal{U}}(M)|} \sum_{\omega=1}^{|\mathcal{G}_{\mathcal{U}}(M)|}\left(\tilde\tau_{\omega} - \tau\right)^{2} + \frac{\left(1 - G_N\right)}{|\mathcal{G}_{\mathcal{U}}(M)|^2} \sum_{\omega=1}^{|\mathcal{G}_{\mathcal{U}}(M)|}\tilde\tau_{\omega}^{2} + \frac{\left(1 - G_N\right)}{|\mathcal{G}_{\mathcal{U}}(M)|^2} \sum_{\omega=1}^{|\mathcal{G}_{\mathcal{U}}(M)|}\sum_{\mathcal{A}_{\omega} \cap \mathcal{A}_{\omega'} \neq \emptyset}\tilde\tau_{\omega} \tilde\tau_{\omega'}\\
& + \frac{\left(1 - G_N\right)}{|\mathcal{G}_{\mathcal{U}}(M)|}\left(\frac{1}{|\mathcal{G}_{\mathcal{U}}(M)|} -  \frac{1}{|\mathcal{G}_{\mathcal{U}}(M)\setminus \mathcal{A}_{\omega}|}\right) \sum_{\omega=1}^{|\mathcal{G}_{\mathcal{U}}(M)|}\sum_{\mathcal{A}_{\omega} \cap \mathcal{A}_{\omega'} = \emptyset}\tilde\tau_{\omega} \tilde\tau_{\omega'}  \\
\simeq & \frac{1}{|\mathcal{G}_{\mathcal{U}}(M)|} \sum_{\omega=1}^{|\mathcal{G}_{\mathcal{U}}(M)|}\left(\tilde\tau_{\omega} - \tau\right)^{2} + \frac{\left(1 - G_N\right)}{|\mathcal{G}_{\mathcal{U}}(M)|^2} \sum_{\omega=1}^{|\mathcal{G}_{\mathcal{U}}(M)|}\tilde\tau_{\omega}^{2}  + \frac{\left(1 - G_N\right)}{|\mathcal{G}_{\mathcal{U}}(M)|^2} \sum_{\omega=1}^{|\mathcal{G}_{\mathcal{U}}(M)|}\sum_{\mathcal{A}_{\omega} \cap \mathcal{A}_{\omega'} \neq \emptyset}\tilde\tau_{\omega}\tilde \tau_{\omega'} \\
& - \frac{\left(1 - G_N\right)}{|\mathcal{G}_{\mathcal{U}}(M)||\mathcal{G}_{\mathcal{U}}(M)\setminus \mathcal{A}_{\omega}|}\frac{M}{G_n} \sum_{\omega=1}^{|\mathcal{G}_{\mathcal{U}}(M)|}\sum_{\mathcal{A}_{\omega} \cap \mathcal{A}_{\omega'} = \emptyset}\tilde\tau_{\omega} \tilde\tau_{\omega'} \simeq \frac{1}{|\mathcal{G}_{\mathcal{U}}(M)|} \sum_{\omega=1}^{|\mathcal{G}_{\mathcal{U}}(M)|}\left(\tilde\tau_{\omega} - \tau\right)^{2}.
\end{aligned}
$$
Since $\frac{1}{|\mathcal{G}_{\mathcal{U}}(M)|} \sum_{\omega=1}^{|\mathcal{G}_{\mathcal{U}}(M)|}\left(\tilde\tau_{\omega} - \tau\right)^{2}$ is uniformly bounded, $\Var\left[\check \tau\right] \to 0$ at the rate of $O\left(G_N^{-1}\right)$. $\check \tau \xrightarrow{P} \tau$ follows from Chebyshev's inequality.
\end{proof}
When $\tilde\tau_{\omega} = \tau_{\omega}$, $\check \tau = \tau_{SAME}$ and $\tau = \tau_{PAME}$. Therefore, Lemma~\ref{lemma:var-g} implies that $\tau_{SAME} \xrightarrow{P} \tau_{PAME}$ under our asymptotic regime.
\begin{lemma}[Asymptotic variance of $\hat \tau_N$]\label{lemma:var-zero}
Define $\bar Y(z) = \frac{1}{|\mathcal{G}_{\mathcal{U}}(M)|} \sum_{\omega = 1}^{|\mathcal{G}_{\mathcal{U}}(M)|}\bar Y_{\omega}(z)$, then
$$
\begin{aligned}
G_N\Var\left[\hat \tau_N\right] \simeq \frac{\sum_{\omega=1}^{|\mathcal{G}_{\mathcal{U}}(M)|}\left(\bar Y_{\omega}(1) - \bar Y(1) \right)^2}{p|\mathcal{G}_{\mathcal{U}}(M)|} + \frac{\sum_{\omega=1}^{|\mathcal{G}_{\mathcal{U}}(M)|}\left(\bar Y_{\omega}(0) - \bar Y(0) \right)^2}{(1-p)|\mathcal{G}_{\mathcal{U}}(M)|}
\end{aligned}
$$
and $\hat \tau_N \xrightarrow{P} \tau_{PAME}$, as $n, N$ grow to infinity under our asymptotic regime.
\end{lemma}
\begin{proof}
Let's first consider $\hat \tau_N$'s variance conditional on group sampling. Due to the randomness of treatment assignment (Assumption 3), this is exactly the classical Neyman variance for the DIM estimator at the group level:
$$
\begin{aligned}
\Var\left[\hat \tau_N \mid \{W_{\omega}\}_{\omega = 1}^{|\mathcal{G}_{\mathcal{U}}(M)|}\right] = & \frac{\sum_{g=1}^{G_N}\left(\bar Y_g(1) - \bar Y_N(1) \right)^2}{G_{N1}(G_N - 1)} + \frac{\sum_{g=1}^{G_N}\left(\bar Y_g(0) - \bar Y_N(0) \right)^2}{G_{N0}(G_N - 1)} - \frac{\sum_{g=1}^{G_N}\left(\tau_g - \tau_{SAME}\right)^2}{G_{N}(G_N - 1)}.
\end{aligned}
$$
From Lemma~\ref{lemma:var-g}, we know that $\bar Y_N(z) \xrightarrow{P} \bar Y(z)$ and $\tau_{SAME} \xrightarrow{P} \tau_{PAME}$. Then, we can show that as $n, N \to \infty$, 
$$
\begin{aligned}
\E\left[G_N\Var\left[\hat \tau_N \mid \{W_{\omega}\}_{\omega = 1}^{|\mathcal{G}_{\mathcal{U}}(M)|}\right]\right] \simeq & \frac{\sum_{\omega = 1}^{\mathcal{G}_{\mathcal{U}}(M)}\left(\bar Y_{\omega}(1) - \bar Y(1) \right)^2}{p|\mathcal{G}_{\mathcal{U}}(M)|} + \frac{\sum_{\omega = 1}^{\mathcal{G}_{\mathcal{U}}(M)}\left(\bar Y_{\omega}(0) - \bar Y(0) \right)^2}{(1-p)|\mathcal{G}_{\mathcal{U}}(M)|} \\
& - \frac{\sum_{\omega = 1}^{\mathcal{G}_{\mathcal{U}}(M)}\left(\tau_{\omega} - \tau_{PAME} \right)^2}{|\mathcal{G}_{\mathcal{U}}(M)|}.
\end{aligned}
$$
From Lemma~\ref{lemma:var-g}, we know that the variance caused by group sampling equals
$$
\begin{aligned}
G_NVar\left[\E\left[\hat \tau_N \mid \{W_{\omega}\}_{\omega \in \mathcal{G}_{\mathcal{U}}(M)} \right]\right] & = G_N*Var\left[\frac{1}{|\mathcal{G}_{\mathcal{U}}(M)|}\sum_{\omega = 1}^{| \mathcal{G}_{\mathcal{U}}(M)|}\frac{W_{\omega}\tau_{\omega}}{\pi_N}\right] \\
& \simeq \frac{1}{|\mathcal{G}_{\mathcal{U}}(M)|} \sum_{\omega=1}^{|\mathcal{G}_{\mathcal{U}}(M)|}\left(\tau_{\omega} - \tau_{PAME}\right)^{2}.
\end{aligned}
$$
Combining the two terms, we have
$$
\begin{aligned}
& G_N\Var\left[\hat \tau_N\right] = G_N*Var\left[\E\left[\hat \tau_N \mid \{W_{\omega}\}_{\omega \in \mathcal{G}_{\mathcal{U}}(M)} \right]\right] + \E\left[G_N*\Var\left[\hat \tau_N \mid \{W_{\omega}\}_{\omega = 1}^{|\mathcal{G}_{\mathcal{U}}(M)|}\right]\right]\\
\simeq & \frac{1}{|\mathcal{G}_{\mathcal{U}}(M)|} \sum_{\omega=1}^{|\mathcal{G}_{\mathcal{U}}(M)|}\left(\tau_{\omega} - \tau_{PAME}\right)^{2} + \frac{\sum_{\omega = 1}^{\mathcal{G}_{\mathcal{U}}(M)}\left(\bar Y_{\omega}(1) - \bar Y(1) \right)^2}{p|\mathcal{G}_{\mathcal{U}}(M)|} + \frac{\sum_{\omega = 1}^{\mathcal{G}_{\mathcal{U}}(M)}\left(\bar Y_{\omega}(0) - \bar Y(0) \right)^2}{(1-p)|\mathcal{G}_{\mathcal{U}}(M)|} \\
& - \frac{\sum_{\omega = 1}^{\mathcal{G}_{\mathcal{U}}(M)}\left(\tau_{\omega} - \tau_{PAME} \right)^2}{|\mathcal{G}_{\mathcal{U}}(M)|}.\\
= & \frac{\sum_{\omega = 1}^{\mathcal{G}_{\mathcal{U}}(M)}\left(\bar Y_{\omega}(1) - \bar Y(1) \right)^2}{p|\mathcal{G}_{\mathcal{U}}(M)|} + \frac{\sum_{\omega = 1}^{\mathcal{G}_{\mathcal{U}}(M)}\left(\bar Y_{\omega}(0) - \bar Y(0) \right)^2}{(1-p)|\mathcal{G}_{\mathcal{U}}(M)|} = O(1).
\end{aligned}
$$
$\hat \tau_N \xrightarrow{P} \tau_{PAME}$ then follows.
\end{proof}

\begin{lemma}[Asymptotic variance of $\hat \tau_N$ with random group formation but without interference]\label{lemma:var-sutva}
In the absence of interference, define $\bar Y_I(z) = \frac{1}{n} \sum_{i = 1}^n Y_i(z)$, then
$$
\begin{aligned}
N\Var\left[\hat \tau_N\right] \simeq  \frac{\sum_{i = 1}^n\left[Y_i(1) - \bar Y_I(1)\right]^2}{np} + \frac{\sum_{i = 1}^n \left[Y_i(0) - \bar Y_I(0)\right]^2}{n(1 - p)}
\end{aligned}
$$
as $n, N$ grow to infinity under our asymptotic regime.
\end{lemma}

\begin{proof}
In this case, $\bar Y_{\omega}(z) = \frac{1}{M} \sum_{i \in \mathcal{A}_{\omega}} Y_i(z)$, and each unit $i$ can belong to ${n-1 \choose M-1}$ different groups, thus
$$
\begin{aligned}
\bar Y(z) = & \frac{1}{|\mathcal{G}_{\mathcal{U}}(M)|} \sum_{\omega=1}^{|\mathcal{G}_{\mathcal{U}}(M)|}\bar Y_{\omega}(z) = \frac{1}{|\mathcal{G}_{\mathcal{U}}(M)|} \sum_{\omega=1}^{|\mathcal{G}_{\mathcal{U}}(M)|}\frac{1}{M} \sum_{i \in \mathcal{A}_{\omega}} Y_i(z) \\
= & \frac{{n-1 \choose M-1}}{|\mathcal{G}_{\mathcal{U}}(M)|M} \sum_{i = 1}^n Y_i(z) = \frac{1}{n}\sum_{i=1}^n Y_i(z) = \bar Y_I(z).
\end{aligned}
$$
We can see that
$$
\begin{aligned}
& \frac{1}{|\mathcal{G}_{\mathcal{U}}(M)|}\sum_{\omega=1}^{|\mathcal{G}_{\mathcal{U}}(M)|}\left(\bar Y_{\omega}(z) - \bar Y(z) \right)^2 = \frac{1}{M^2|\mathcal{G}_{\mathcal{U}}(M)|}\sum_{\omega=1}^{|\mathcal{G}_{\mathcal{U}}(M)|} \left[\sum_{i \in \mathcal{A}_{\omega}} \left(Y_{i}(z) - \bar Y_I(z) \right)\right]^2 \\
= & \frac{1}{M^2|\mathcal{G}_{\mathcal{U}}(M)|}\sum_{\omega=1}^{|\mathcal{G}_{\mathcal{U}}(M)|} \left[\sum_{i \in \mathcal{A}_{\omega}} \left(Y_{i}(z) - \bar Y_I(z) \right)^2 + \sum_{i \in \mathcal{A}_{\omega}}\sum_{j \neq i, j \in \mathcal{A}_{\omega}} \left(Y_{i}(z) - \bar Y_I(z) \right)\left(Y_{j}(z) - \bar Y_I(z) \right)\right] \\
= & \frac{{n-1 \choose M-1}}{M^2|\mathcal{G}_{\mathcal{U}}(M)|}\sum_{i = 1}^n  \left(Y_{i}(z) - \bar Y_I(z) \right)^2 + \frac{{n-2 \choose M-2} }{M^2|\mathcal{G}_{\mathcal{U}}(M)|}\sum_{i = 1}^n\sum_{j \neq i} \left(Y_{i}(z) - \bar Y_I(z) \right)\left(Y_{j}(z) - \bar Y_I(z) \right) \\
= & \frac{1}{Mn}\sum_{i = 1}^n \left(Y_{i}(z) - \bar Y_I(z) \right)^2 + \frac{M-1}{Mn(n-1)}\sum_{i = 1}^n\sum_{j \neq i} \left(Y_{i}(z) - \bar Y_I(z) \right)\left(Y_{j}(z) - \bar Y_I(z) \right)\\ \simeq  & \frac{1}{Mn}\sum_{i = 1}^n \left(Y_{i}(z) - \bar Y_I(z) \right)^2.
\end{aligned}
$$
The last line uses the fact that
$$
\begin{aligned}
   & \frac{1}{n^2}\sum_{i = 1}^n\sum_{j \neq i} \left(Y_{i}(z) - \bar Y_I(z) \right)\left(Y_{j}(z) - \bar Y_I(z) \right) \\
    = & \left(\frac{1}{n}\sum_{i = 1}^n \left(Y_{i}(z) - \bar Y_I(z) \right)\right)^2 - \frac{1}{n^2}\sum_{i = 1}^n \left(Y_{i}(z) - \bar Y_I(z) \right)^2 \\ 
    = & - \frac{1}{n^2}\sum_{i = 1}^n \left(Y_{i}(z) - \bar Y_I(z) \right)^2 \to 0,
\end{aligned}
$$
Consequently,
$$
\begin{aligned}
& N\Var\left[\hat \tau_N\right] = (MG_N)\Var\left[\hat \tau_N\right] \\
\simeq & \frac{M\sum_{\omega=1}^{|\mathcal{G}_{\mathcal{U}}(M)|}\left(\bar Y_{\omega}(1) - \bar Y(1) \right)^2}{p|\mathcal{G}_{\mathcal{U}}(M)|} + \frac{M\sum_{\omega=1}^{|\mathcal{G}_{\mathcal{U}}(M)|}\left(\bar Y_{\omega}(0) - \bar Y(0) \right)^2}{(1-p)|\mathcal{G}_{\mathcal{U}}(M)|} \\
\simeq & \frac{1}{np}\sum_{i = 1}^n \left(Y_{i}(1) - \bar Y_I(1) \right)^2 + \frac{1}{n(1 - p)}\sum_{i = 1}^n \left(Y_{i}(0) - \bar Y_I(0) \right)^2.
\end{aligned}
$$
\end{proof}

To estimate $N\Var\left[\hat \tau_N\right]$, we rely on the group-level HC2 robust variance estimator:
$$
\begin{aligned}
\widehat{\Var} \left[\hat \tau_N\right] = \frac{\frac{\sum_{g=1}^{G_{N}}Z_g\left(\bar Y_{g} - \widehat{\bar Y}(1)\right)^2}{G_{N1} - 1}}{G_{N1}} + \frac{\frac{\sum_{g=1}^{G_{N}}(1 - Z_g)\left(\bar Y_{g} - \widehat{\bar Y}(0)\right)^2}{G_{N0} - 1}}{G_{N0}} = \frac{S_1^2}{G_{N1}} + \frac{S_0^2}{G_{N0}},
\end{aligned}
$$
where $\widehat{\bar Y}(z)= \frac{\sum_{g=1}^{G_N}\mathbf{1}\{Z_g = z\}\bar Y_g}{\sum_{g=1}^{G_N}\mathbf{1}\{Z_g = z\}}$. We establish its consistency in the following lemma. Meanwhile, in the absence of interference, we can use the individual-level HC2 robust variance estimator to estimate $N*\Var\left[\hat \tau_N\right]$. The consistency of this estimator has been established in the literature \citep{su2021model}. We denote these two variance estimators as $\hat{V}^{CR}_N$ and $\hat{V}^{HR}_N$, respectively.

\begin{lemma}[Consistency of the group level HC2 variance estimator]\label{lemma:var-est}
As $n, N$ grow to infinity under our asymptotic regime,
$$
\begin{aligned}
G_N  \widehat{\Var}\left[\hat \tau_N\right] \xrightarrow{P} G_N \Var \left[\hat \tau_N\right].
\end{aligned}
$$
\end{lemma}

\begin{proof}
We first derive the variance estimator's expectation conditional on the sampling indicators. For the first term, we know that 
$$
\begin{aligned}
& \E\left[\frac{\sum_{g=1}^{G_{N}}Z_g\left(\bar Y_{g} - \widehat{\bar Y}(1)\right)^2}{G_{N1} }\mid \{W_{\omega}\}_{\omega = 1}^{|\mathcal{G}_{\mathcal{U}}(M)|}\right] \simeq \frac{\sum_{g=1}^{G_N}\left(\bar Y_{g}(1) - \bar Y_N(1) \right)^2}{G_N}.
\end{aligned}
$$
Next, we take expectation over the sampling indicators:
$$
\begin{aligned}
& \E\left[\frac{1}{G_N} \sum_{g=1}^{G_N}\left(\bar Y_{g}(1) - \bar Y_N(1)\right)^2\right] = \E\left[\frac{1}{G_N} \sum_{g=1}^{G_N}\bar Y_{g}^{2}(1) - \left(\bar Y_N(1)\right)^2\right] \\
= & \E\left[\frac{1}{G_N} \sum_{\omega \in \mathcal{G}_{\mathcal{U}}(M)}w_{\omega}\bar Y_{\omega}^{2}(1) - \left(\bar Y_N(1)\right)^2\right] \\
= & \frac{\pi_N}{G_N} \sum_{\omega \in \mathcal{G}_{\mathcal{U}}(M)}\bar Y_{\omega}^{2}(1) - \frac{1}{G_N^2}E\left[\left(\sum_{\omega \in \mathcal{G}_{\mathcal{U}}(M)}w_{\omega}\bar Y_{\omega}(1)\right)\left(\sum_{\omega \in \mathcal{G}_{\mathcal{U}}(M)}w_{\omega}\bar Y_{\omega}(1)\right)\right] \\
= & \frac{\pi_N}{G_N} \sum_{\omega \in \mathcal{G}_{\mathcal{U}}(M)}\bar Y_{\omega}^{2}(1) - \frac{\pi_N}{G_N^2} \sum_{\omega \in \mathcal{G}_{\mathcal{U}}(M)}\bar Y_{\omega}^{2}(1) - \frac{1}{G_N^2}\frac{\pi_N(G_N-1)}{|\mathcal{G}_{\mathcal{U}}(M) \setminus \mathcal{A}_{\omega}|}\sum_{\omega \in \mathcal{G}_{\mathcal{U}}(M)}\sum_{\omega' \cap \omega = \emptyset}\bar Y_{\omega}(1) \bar Y_{\omega'}(1) \\
= & \frac{1}{|\mathcal{G}_{\mathcal{U}}(M)|} \sum_{\omega \in \mathcal{G}_{\mathcal{U}}(M)}\bar Y_{\omega}^{2}(1) - \frac{1}{G_N|\mathcal{G}_{\mathcal{U}}(M)|} \sum_{\omega \in \mathcal{G}_{\mathcal{U}}(M)}\bar Y_{\omega}^{2}(1) \\
& - \left(1 - \frac{1}{G_N}\right) \frac{1}{|\mathcal{G}_{\mathcal{U}}(M)||\mathcal{G}_{\mathcal{U}}(M) \setminus \mathcal{A}_{\omega}|}\sum_{\omega \in \mathcal{G}_{\mathcal{U}}(M)}\sum_{\omega' \cap \omega = \emptyset}\bar Y_{\omega}(1) \bar Y_{\omega'}(1) \\
\simeq & \frac{1}{|\mathcal{G}_{\mathcal{U}}(M)|} \sum_{\omega \in \mathcal{G}_{\mathcal{U}}(M)}\bar Y_{\omega}^{2}(1) - \frac{1}{|\mathcal{G}_{\mathcal{U}}(M)||\mathcal{G}_{\mathcal{U}}(M) \setminus \mathcal{A}_{\omega}|}\sum_{\omega \in \mathcal{G}_{\mathcal{U}}(M)}\sum_{\omega' \cap \omega = \emptyset}\bar Y_{\omega}(1) \bar Y_{\omega'}(1) \\
\simeq & \frac{1}{|\mathcal{G}_{\mathcal{U}}(M)|} \sum_{\omega \in \mathcal{G}_{\mathcal{U}}(M)}\bar Y_{\omega}^{2}(1) - \bar Y^{2}(1) . 
\end{aligned}
$$
We can similarly show that $\E\left[\frac{\sum_{g=1}^{G_{N}}(1 - Z_g)\left(\bar Y_{g} - \widehat{\bar Y}\right)^2}{G_{N0} }\right] \simeq \frac{1}{|\mathcal{G}_{\mathcal{U}}(M)|} \sum_{\omega \in \mathcal{G}_{\mathcal{U}}(M)}\bar Y_{\omega}^{2}(0) - \bar Y^{2}(0)$. Putting together, we know that
$$
\begin{aligned}
G_N\E\left[\widehat{\Var}\left[\hat \tau_N\right]\right] \simeq & \frac{1}{p} \{ \frac{1}{|\mathcal{G}_{\mathcal{U}}(M)|} \sum_{\omega \in \mathcal{G}_{\mathcal{U}}(M)}\bar Y_{\omega}^{2}(1) - \bar Y^{2}(1)\} +\frac{1}{1-p}  \{\frac{1}{|\mathcal{G}_{\mathcal{U}}(M)|(1 - p)} \sum_{\omega \in \mathcal{G}_{\mathcal{U}}(M)}\bar Y_{\omega}^{2}(0) - \bar Y^{2}(0) \}\\
= & G_N \Var \left[\hat \tau_N\right].
\end{aligned}
$$
As the potential outcomes are uniformly bounded, it is straightforward to show that the estimator's variance converges to zero. The lemma then follows from Chebyshev's inequality.
\end{proof}

Next, we show that group-level sample averages converge to a normal distribution under the restricted group sampling design. Our proof couples the distribution of the group-level sample averages under the restricted design with the distribution under an independent group sampling design, showing that the total variation distance between these two distributions goes to zero asymptotically. We first state the following lemma for the independent group sampling design.

\begin{lemma}[Asymptotic normality of group-level sample averages under independent group sampling]\label{lemma:normality-ind-sampling}
Consider a hypothetical experiment where we independently sample $G_N$ groups one by one with replacement from the population of all groups, $\mathcal{G}_{\mathcal{U}}(M)$. For any $\tilde \tau_{\omega}$ defined for group $\omega$ such that $|\tilde \tau_{\omega}| < C$, $\check \tau = \frac{1}{G_N} \sum_{g=1}^{G_N} \tilde \tau_{g}$, and $\tau = \frac{1}{|\mathcal{G}_{\mathcal{U}}(M)|} \sum_{\omega=1}^{|\mathcal{G}_{\mathcal{U}}(M)|}\tilde \tau_{\omega}$,
$$
\sqrt{G_N}\big(\check \tau - \tau\big) \rightsquigarrow \mathcal{N}\left(0, \sigma^2_{\tau}\right)
$$
as $n, N$ grow to infinity under our asymptotic regime.
\end{lemma}
\begin{proof}
Since $\{\tilde \tau_{g}\}_{g=1}^{G_N}$ are i.i.d. and $\tau_g$ is uniformly bounded, the Lindeberg--Feller central limit theorem applies. 
\end{proof}

\begin{lemma}[Asymptotic normality of group-level sample averages under the restricted group sampling design]\label{lemma:normality-sampling}
For any $\tilde \tau_{\omega}$ defined for group $\omega$ such that $|\tilde \tau_{\omega}| < C$, $\check \tau = \frac{1}{|\mathcal{G}_{\mathcal{U}}(M)|\pi_N} \sum_{\omega=1}^{|\mathcal{G}_{\mathcal{U}}(M)|}W_{\omega} \tilde \tau_{\omega}$, and $\tau = \frac{1}{|\mathcal{G}_{\mathcal{U}}(M)|} \sum_{\omega=1}^{|\mathcal{G}_{\mathcal{U}}(M)|}\tilde \tau_{\omega}$,
$$
\begin{aligned}
\sqrt{G_N}\left(\check \tau - \tau\right) \rightsquigarrow \mathcal{N}\left(0, G_N\Var\left[\check \tau\right]\right)
\end{aligned}
$$
as $n, N$ grow to infinity under our asymptotic regime.
\end{lemma}
\begin{proof}
The outcome of either the original or the hypothetical experiment can be represented as an ordered $G_N$-tuple: $\tilde{\mathcal{A}} = \left(\mathcal{A}_1,\mathcal{A}_2,\dots,\mathcal{A}_{G_N}\right)$. In the hypothetical experiment, the probability of observing any given tuple $\tilde{\mathcal{A}}$ is $q^* = \frac{1}{|\mathcal{G}_{\mathcal{U}}(M)|^{G_N}}$. Under the original sampling algorithm, the probability of observing the same tuple is
$$
q = \frac{1}{\prod_{r=0}^{G_N - 1}{n - rM \choose M}}.
$$
Let $Q^*$ and $Q$ denote the distribution of $\tilde{\mathcal{A}}$ the hypothetical and original experiments, respectively, and let $supp(Q)$ to denote the support of $Q$. Clearly, 
$$
supp(Q) \subset supp(Q^*), |supp(Q)| = \prod_{r=0}^{G_N - 1}{n - rM \choose M} = \frac{1}{q}.
$$

Let $\mathcal{F}^*$ and $\mathcal{F}$ denote the $\sigma$-algebras underlying $Q^*$ and $Q$, respectively. We further use $Q\left(\tilde{\mathcal{A}}\right)$ to represent the distribution of $\tilde{\mathcal{A}}$ under $Q$. The total variation distance between $Q^*$ and $Q$ is
$$
\begin{aligned}
d_{TV}(Q^*, Q) = & \frac{1}{2} \sum_{\tilde{\mathcal{A}}}|Q^*\left(\tilde{\mathcal{A}}\right) - Q\left(\tilde{\mathcal{A}}\right)| = \frac{1}{2}|supp(Q)| \left(q - q^*\right) + \frac{1}{2}|supp(Q^*) \setminus supp(Q)| q^*  \\
= & \frac{1}{2q}\left(q - q^*\right) + \frac{1}{2} - \frac{q^*}{2q} = 1 - \frac{q^*}{q} = 1 - \frac{\prod_{r=0}^{G_N - 1}{n - rM \choose M}}{|\mathcal{G}_{\mathcal{U}}(M)|^{G_N}} \\
= & 1 - \prod_{r=0}^{G_N - 1}\frac{{n - rM \choose M}}{{n \choose M}} \simeq 1 - \prod_{r=0}^{G_N - 1}\left(1 - \frac{rM}{n}\right)^M \\
\simeq & 1 - \prod_{r=0}^{G_N - 1}e^{-\frac{rM^2}{n}} = 1 - e^{-\frac{M^2}{n}\frac{G_N(G_N - 1)}{2}} \simeq \frac{M^2}{n}\frac{G_N(G_N - 1)}{2} \to 0,
\end{aligned}
$$

Let $T$ be any measurable map defined on the space of ordered $G_N$-tuples. Denote by $L_{T, Q}$ and $L_{T, Q^*}$ the induced distributions of $T$ under $Q$ and $Q^*$, respectively. Since total variation distance cannot increase under a measurable map (the data-processing inequality),
$$
d_{TV}(L_{T, Q}, L_{T, Q^*}) \;\le\; d_{TV}(Q, Q^*) \to 0.
$$

Since $T^* = \sqrt{G_N}\Bigg(\frac{1}{G_N}\sum_{g = 1}^{G_N}\tilde\tau_{g} - \tau\Bigg)$ and $T = \sqrt{G_N}\left(\check \tau - \tau\right)$ are the same map applied to $Q^*$ and $Q$, respectively, the total variation distance between them converges to zero. Therefore, $T$ converges in distribution to the same normal limit as $T^*$.
\end{proof}
When $\tilde\tau_{\omega} = \tau_{\omega}$, $\check \tau = \tau_{SAME}$ and $\tau = \tau_{PAME}$. Therefore, Lemma~\ref{lemma:normality-sampling} implies that $\sqrt{G_N}\left(\tau_{SAME} - \tau_{PAME}\right) \rightsquigarrow \mathcal{N}\left(0, \sigma_{\tau}^2\right)$ under our asymptotic regime.

\begin{prop}
\begin{it}
Given Assumption 1 and the Design Cases and Outcome Cases defined in the main text, we have the following as $n$ and $N$ go to infinity in our asymptotic regime:
\begin{enumerate}
    \item With fixed groups and no interference (Design-Outcome Case 0-0), $\hat \tau_N \xrightarrow{P} \tau_{ATE}$.
    \item With random groups and no interference (Design-Outcome Case 1-0), $\hat \tau_N \xrightarrow{P} \tau_{ATE}$.
    \item With fixed groups and interference (Design-Outcome Case 0-1), $\hat \tau_N \xrightarrow{P} TOT$.
    \item With random groups and interference (Design-Outcome Case 1-1), $\hat \tau_N \xrightarrow{P} PAME$.
\end{enumerate}
\end{it}
\end{prop}

\begin{proof}
The first three cases are covered by existing results for cluster and unit randomized designs.
The novel result for case 4 concerns the fourth case with random groups and interference.
The result follows from Chebyshev's inequality given unbiasedness and then convergence of the variance to zero, which is established in Lemmas~\ref{lemma:group-tau} and~\ref{lemma:var-zero}. 
\end{proof}

\begin{prop}
\begin{it}
Given Assumption 1 and the Design Cases and Outcome Cases defined in the main text, we have the following:
\begin{enumerate}
    \item In all four Design-Outcome Case combinations specified above,
    $$
    G_N\left(\hat{V}^{CR}_N - \Var(\hat \tau_N)\right) \xrightarrow{P} 0.
    $$
    \item With random groups and no interference (Design-Outcome Case 1-0),  $$
    N\left(\hat{V}^{HR}_N - \Var(\hat \tau_N)\right) \xrightarrow{P} 0.
    $$
\end{enumerate}
\end{it}
\end{prop}
\begin{proof}
The scaled absolute-consistency statements follow from Lemmas~\ref{lemma:var-zero}-\ref{lemma:var-est}. To obtain the ratio-consistency result in Theorem 2, additionally assume that the corresponding scaled asymptotic variance is nondegenerate: for some $c>0$, $G_N Var(\hat \tau_N) \ge c$ for all sufficiently large N, and, in Case 1-0, $N Var(\hat \tau_N) \ge c$. Then 
$$
\frac{\hat V_n^{CR}}{Var(\hat \tau_N)}-1=\frac{G_N\{\hat V^{CR}_N-Var(\hat \tau_N)\}}{G_N Var(\hat \tau_N)}\xrightarrow{P} 0.
$$
\end{proof}

\begin{prop}
\begin{it}
Given Assumption 1 and the Design Cases and Outcome Cases defined in the main text, indexed respectively by $d\in\{0,1\}$ and $o \in \{0,1\}$, we have the following:
\begin{enumerate}
    \item In all four Design-Outcome Case combinations specified above we have
    $$
\sqrt{G_N}(\hat{\tau}_N - \tau_{do}) \rightsquigarrow \N(0,G_N\Var(\hat \tau_N)).
    $$
        \item With random groups and no interference (Design-Outcome Case 1-0),  $$
\sqrt{N}(\hat{\tau}_N - \tau_{10}) \rightsquigarrow \N(0,N\Var(\hat \tau_N)).
    $$
\end{enumerate}
\end{it}
\end{prop}

\begin{proof}
First, observe that
$$
\sqrt{G_N}(\hat \tau_N - \tau_{PAME}) = \sqrt{G_N}(\hat \tau_N - \tau_{SAME}) + \sqrt{G_N}(\tau_{SAME} - \tau_{PAME}) = S_1 + S_2.
$$
From Lemma~\ref{lemma:normality-sampling}, we know that $S_2$ is asymptotically normal. We prove the asymptotic normality of $S_1 + S_2$ using the two-stage central limit theorem in Ohlsson (1989).

Let's define $\tilde{Y}_g = \frac{\bar Y_{g}(1)}{p} + \frac{\bar Y_{g}(0)}{1 - p}$ and $\tilde{Y}= \frac{1}{G_N}\sum_{g = 1}^{G_{N}}\xi_g = \frac{1}{|\mathcal{G}_{\mathcal{U}}(M)|\pi_N}\sum_{\omega = 1}^{|\mathcal{G}_{\mathcal{U}}(M)|}W_{\omega}\tilde{Y}_{\omega}$, then
$$
\begin{aligned}
    & S_1 = \sqrt{G_N}(\hat \tau_N - \tau_{SAME}) = \sum_{g = 1}^{G_{N}}\sqrt{G_N}\left(\frac{Z_{g}\bar Y_{g}}{G_{N}p} - \frac{(1 - Z_{g})\bar Y_{g}}{G_{N}(1-p)} - \frac{\tau_{g}}{G_N}\right) \\
    = & \sum_{g = 1}^{G_{N}}\frac{\sqrt{G_N}}{G_N}\left(Z_g - p\right) \left(\frac{\bar Y_{g}(1)}{p} + \frac{\bar Y_{g}(0)}{1 - p}\right)  = \sum_{g = 1}^{G_{N}}\frac{\sqrt{G_N}}{G_N}\left(Z_g - p\right) \left(\tilde{Y}_g - \tilde{Y}\right).
\end{aligned}
$$ 
The last equality uses the fact that $\sum_{g = 1}^{G_{N}}\left(Z_g - p\right) = 0$.

Consider an alternative experimental design where each group $\mathcal{A}_{\omega} \in \mathcal{G}_{\mathcal{U}}(M)$ is assigned with the treatment independently with probability $p$ (i.e., a Bernoulli trial). We denote the treatment status for group $\mathcal{A}_{\omega}$ in this design as $Z_{\omega}^*$ and define $S_1^* = \sum_{g = 1}^{G_{N}}\frac{\sqrt{G_N}}{G_N}\left(Z_g^* - p\right) \left(\tilde{Y}_g - \tilde{Y}\right)$.

Let $\mathcal{F}_0$ denote the $\sigma$-algebra generated by $\{W_{\omega}\}_{\omega \in \mathcal{G}_{\mathcal{U}}(M)}$. For $g=1,\ldots,G_N$, define the filtration $\{\mathcal{F}_g\}_{g=0}^{G_N}$ by
\[
\mathcal{F}_g = \mathcal{F}_0 \vee \sigma\!\left(\left\{\frac{\sqrt{G_N}}{G_N}\left(Z_{g'}^*-p\right)\left(A_{g'}-A\right)\right\}_{g'=1}^{g}\right),
\]
where $\sigma(\cdot)$ denotes the $\sigma$-algebra generated by the enclosed random variables, and $\vee$ denotes the smallest $\sigma$-algebra containing both arguments. Then, we can verify that
$$
\E\left[\frac{\sqrt{G_N}}{G_N}\left(Z_g^* - p\right) \left(\tilde{Y}_g - \tilde{Y}\right) \mid \mathcal{F}_{g-1}\right] = 0.
$$
Therefore, $\Big\{\frac{\sqrt{G_N}}{G_N}\left(Z_g^* - p\right) \left(\tilde{Y}_g - \tilde{Y}\right)\Big\}_{g = 1}^{G_N}$ is a martingale difference sequence, and
$$
\begin{aligned}
    & \frac{1}{G_N}\sum_{g=1}^{G_N}\E\left[\left(Z_{g}^* - p\right)^2\left(A_{g} - A\right)^2 \mid \mathcal{F}_{g-1}\right] = \frac{\sum_{g = 1}^{G_N}\left(\bar Y_{g}(1) - \bar Y_N(1) \right)^2}{pG_N} \\
    & + \frac{\sum_{g = 1}^{G_N}\left(\bar Y_{g}(0) - \bar Y_N(0) \right)^2}{(1-p)G_N} - \frac{\sum_{g= 1}^{G_N}\left(\tau_{g} - \tau_{SAME} \right)^2}{G_N} = O(1).
\end{aligned}
$$
From Theorem A.1 in \citet{ohlsson1989asymptotic}, we know that $S_1^* + S_2$ converges to a normal distribution. We next transfer this result to $S_1 + S_2$ using H{\'a}jek's coupling \citep{hajek1960limiting}. Let $G_{N1}^*$ denote the number of treated groups under the Bernoulli assignment. We first draw $G_{N1}^*$ from a binomial distribution $Bin(G_N, p)$. Conditional on $G_{N1}^*$, we construct the Bernoulli and complete-randomization assignments on a common probability space as follows. If $G_{N1}^*>G_{N1}$, we first draw a simple random sample of $G_{N1}^*$ treated groups from the $G_N$ groups and then obtain the complete-randomization assignment by drawing a simple random subsample of size $G_{N1}$ from these treated groups. Conversely, if $G_{N1}^*<G_{N1}$, we first draw a simple random sample of $G_{N1}$ treated groups from the $G_N$ groups and then obtain the Bernoulli assignment by drawing a simple random subsample of size $G_{N1}^*$ from these treated groups. This construction gives the correct marginal distributions for both the complete-randomization assignment and the alternative assignment.

Conditional on $G_{N1}^*$, the randomness in $S_1-S_1^*$ arises solely from this simple random sampling step. Moreover, $\E[S_1-S_1^*\mid G_{N1}^*]=0$. By construction, the set $\{g:Z_g\neq Z_g^*\}$ is a simple random sample of size $|G_{N1}-G_{N1}^*|$ from the $G_N$ groups. Therefore,
$$
\begin{aligned}
    & \E\left[\left(S_1 - S_1^*\right)^2 \mid G_{N1}^*, \{W_{\omega}\}_{\omega \in \mathcal{G}_{\mathcal{U}}(M)}\right] \\
    = &  G_N\E\left[\frac{1}{G_N^2}\left(\sum_{g = 1}^{G_N} \left(\tilde{Y}_g - \tilde{Y}\right) \left(Z_{g} - Z_{g}^*\right)\right)^2 \mid G_{N1}^*, \{W_{\omega}\}_{\omega \in \mathcal{G}_{\mathcal{U}}(M)}\right] \\
    = & G_N\Var\left[\frac{1}{G_N}\left(\sum_{g: Z_{g} \ne Z_{g}^*} \left(\tilde{Y}_g - \tilde{Y}\right) \left(Z_{g} - Z_{g}^*\right)\right) \mid G_{N1}^*, \{W_{\omega}\}_{\omega \in \mathcal{G}_{\mathcal{U}}(M)}\right]   \\
    = & \frac{1}{G_N} \frac{|G_{N1} - G_{N1}^*|}{G_N}\frac{G_N - |G_{N1} - G_{N1}^*|}{G_N - 1}\sum_{g = 1}^{G_N}\left( \tilde{Y}_g - \tilde{Y}\right)^2 \le \frac{|G_{N1} - G_{N1}^*|}{G_N} \frac{\sum_{g = 1}^{G_N}\left( \tilde{Y}_g - \tilde{Y}\right)^2}{G_N} \\
    \asymp & \frac{|G_{N1} - G_{N1}^*|}{G_N},
\end{aligned}
$$
where the third equality follows from the variance formula for the sample mean under simple random sampling without replacement. Since $\E|G_{N1}-G_{N1}^*|=O(\sqrt{G_N})$, the tower property implies that $\E\left[\left(S_1 - S_1^*\right)^2\right] \to 0$. Hence $S_1-S_1^*\to0$ in $L^2$, and therefore also in probability. By Slutsky's theorem, $S_1 - S_1^* \to 0$, $S_1 + S_2 - (S_1^* + S_2) \to 0$, and $S_1 + S_2$ converges to a normal distribution. The remainder of the proof follows from Lemmas \ref{lemma:var-zero}-\ref{lemma:var-est} and Slutsky's theorem.

\end{proof}
The combination of Proposition 1 and Proposition 3 leads to Theorem 1 in the main text for homogeneous group sizes. Proposition 2 implies Theorem 2 in the main text for homogeneous group sizes.


\subsection{Varying group sizes}
We now generalize the analysis to account for varying group sizes.
We assume that there are $K$ admissible group sizes, $\mathcal{M} = \{m_1,...,m_K\}$. A sample of $N$ units is randomly sampled from the population $\mathcal{U}$ and assigned to groups $g=1,...,G_N$ through random partition. The partition is created by a random permutation of an index vector of length $N$ of the form $\underbrace{1,...,1}_{M_1},...,\underbrace{g,...,g}_{M_g},...,\underbrace{G_N,...,G_N}_{M_{G_N}}$, in which each index value $g$ appears $M_g$ times, so that each unit is randomly assigned to a group. We write $G_N(m)$ for the number of groups with $M_g = m \in \mathcal{M}$. Among them, the first $G_{N1}(m)$ groups are assigned to ``treatment'' and the remaining $G_{N0}(m)$ groups are under ``control,'' with group treatment indicator $Z_g \in \{0,1\}$ equal to $1$ for treated groups and $0$ for control groups. This design is more general than Design Case 1 in the main text as it allows the probability of treatment to vary by group sizes, which we denote with $p(m_k)$ for the $k$th admissible group size. $p(m_k)$ can be 0 or 1 for specific $k$, thus admissible group sizes can differ across arms. We let $\mathcal{M}_1$ and $\mathcal{M}_0$ denote the distinct group sizes appearing in the treatment and control arms, respectively. $G_{N1}$ and $G_{N0}$ represent the total number of groups under treatment and control. $N_1$ and $N_0$ denote the numbers of units under treatment and control, respectively, while $N_z(m)$ denotes the number of units in size-m groups assigned to treatment status $z$.

\subsubsection{Generalizing the estimands with varying group sizes}
We now generalize the marginalized estimands to varying group sizes. Let $\mathcal{G}_{\mathcal{U},j}(m_k) = \{\mathcal{S} \subseteq \mathcal{U} : j \in \mathcal{S},\, |\mathcal{S}| = m_k\}$ denote the set of candidate size-$m_k$ groups in the population containing $j$, and define the within-size marginalized potential outcome
$$
\mu_j(z;m_k) = \frac{1}{\binom{n-1}{m_k-1}}\sum_{\mathcal{A} \in \mathcal{G}_{\mathcal{U},j}(m_k)} Y_j(z,\mathcal{A}),
$$
the uniform average of $Y_j(z,\mathcal{A})$ over those candidates. This is uniform marginalization over potential groups within a size class. By the symmetry of simple random sampling composed with random partitioning, $\mu_j(z;m_k)$ is also the design-conditional expectation given size $m_k$, $\mu_j(z;m_k) = \E[Y_j(z,\mathcal{A}_{A_j}) \mid Z_j = z,\, |\mathcal{A}_{A_j}| = m_k]$; this within-size correspondence between design distribution and uniform average carries over from the homogeneous case unchanged. 

We combine the within-size outcomes through a weighted average,
$$
\bar\mu^{(\phi)}_i(z) = \sum_{m_k \in \mathcal{M}_z}\phi(z;m_k)\, \mu_i(z;m_k),
$$
where the weights $\phi(z;m_k) \ge 0$ sum to one over $m_k \in \mathcal{M}_z$ and encode how much each admissible size counts toward the estimand. The unit-level marginalized exposure effect under these weights is $\tau^{(\phi)}_i = \bar\mu^{(\phi)}_i(1) - \bar\mu^{(\phi)}_i(0)$, and the population average marginalized effect under $\phi$ is 
$$
\tau^{(\phi)}_{PAME} = \frac{1}{n}\sum_{i \in \mathcal{U}}\tau^{(\phi)}_i.
$$ 
Different choices of $\phi$ give different forms of this single estimand. We discuss four of them. 

\paragraph{PAME with design-induced weights.}  Let
$$
\phi(z;m_k) = \Pr\!\left(|\mathcal{A}_{A_j}| = m_k \mid Z_j = z\right)
$$
be the probability under the design that unit $j$ is in a group of size $m_k$ when assigned to condition $z$. If the design assigns units to groups using the uniform permutation approach defined in the main text, then the $\phi(z;m_k)$ probabilities may vary by group size. As a toy example, if the design designates that there should be 2 groups of size 2 and 2 groups of size 1 in the treatment arm, then conditional on being in the treatment arm, a unit is twice as likely to be in a group of size 2 than of size 1 (since there are 4 ``slots'' available for groups of size 2 but only 2 slots available for groups of size 1). By the law of iterated expectations, $\bar\mu^{(\phi)}_i(z) = \E[Y_j(z,\mathcal{A}_{A_j}) \mid Z_j = z]$, the design-conditional expectation of unit $j$'s potential outcome. $\tau^{(\phi)}_{PAME}$ in this case describes the realized effect under the design's specific allocation across size classes, weighting size $M_k$ by the share of units the design assigns to that size in arm $z$. 

\paragraph{PAME with even weights across group sizes.} We can weight $\mu_i(z;m_k)$ uniformly over admissible size classes by setting
$$
\phi(z;m_k) = 1/|\mathcal{M}_z|.
$$
The PAME defined this way is a property of the population: changing the design's allocation across size classes $\phi_j(z;m_k)$  does not change $\tau_{PAME}$ as long as the same size classes remain admissible. The PAME defined this way also has stable interpretations over populations that grow to infinity. The first two estimands coincide whenever (i) $K=1$ (homogeneous) or (ii) the design is balanced across size classes so $\Pr\!\left(|\mathcal{A}_{A_j}| = m_k \mid Z_j = z\right) = 1/|\mathcal{M}_z|$. But they differ when more than one size class is admissible in an arm and the design's allocation across classes within that arm is not uniform. In such cases, the first estimand has the property that it would differ depending on the design (e.g., in our toy example above, it would change when going from having 2 groups of 1 unit to 4 groups of 1 unit, thereby evening out the group assignment probabilities to groups of size 2 and size 1). Whether this is desirable or not depends on the substantive question at hand.

\paragraph{PAME with even weights across groups.} Another alternative estimand could be based on uniformly averaging all potential groups, pooling potential groups from all admissible size classes: 
$$
\phi(z;m_k) = \binom{n-1}{m_k-1}/\sum_{m_{k'} \in \mathcal{M}_z}\binom{n-1}{m_{k'}-1}.
$$ 
This object is also identified by the design. But its weights depend on $n$, and as $n \to \infty$ with the set of size classes held fixed, a population average causal effect based on this kind of averaging converges to the within-class effect at the largest admissible size. To see this, consider the case where the admissible groups are of either size 1 or 2.  For each unit $j$ in the population, the number of potential groups of size 2 that it could be part of grows without bound in the population size  $n$, while the number of potential groups of size 1 remains fixed at 1.  For large $n$, the contribution of the potential outcome for a group of size 1 vanishes. This phenomenon generalizes to arbitrary admissible group sizes, in which case the estimand degenerates to the marginal effect over the largest group size. 

\paragraph{Within-size conditional contrasts.} Finally, when group sizes differ across arms (as in the example given above, where size-$m_k$ groups appear only in one arm) the within-size-$m_k$ difference in means is undefined for some $k$, and $\tau_{PAME}$ is instead targeted by cross-condition contrasts, which we can define by setting
$$
\phi(z;m_k) = 1, \phi(z';m_{k'}) = 1, \text{ for given }  z, z', m_k, m_{k'}, \text{ and } \phi(\cdot) = 0 \text{ otherwise}.
$$
Then,
$$
\begin{aligned}
& \tau^{(\phi)}_{PAME} = \frac{1}{n}\sum_{i \in \mathcal{U}}\left(\mu_i(z; m_k) - \mu_i(z'; m_{k'})\right),
\end{aligned}
$$
is well-defined when $\phi(z;m_k) \ne 0$ and $\phi(z';m_{k'}) \ne 0$ for all $i \in \mathcal{U}$.  In the \citet{mendelberg2014does} example discussed in the main text, this estimand would capture pairwise contrasts between outcomes in the various treated and control cells shown in Table 2.  Differences between these estimands capture interaction effects between treatment and group size, as is done in the analysis of the \citet{mendelberg2014does} example. 

\paragraph{Estimation.} To target $\tau^{(\phi)}_{PAME}$ under weights $\phi$, we first estimate, for each arm $z$ and admissible size $m_k \in \mathcal{M}_z$, the within-size mean outcome
$$
\hat\mu(z;m_k) = \frac{\sum_{i \in \mathcal{N}}\mathbf{1}\{|\mathcal{A}_{A_i}| = m_k\}\,\mathbf{1}\{Z_i = z\}\,Y_i}{G_{Nz}(m_k)\,m_k}.
$$
The plug-in estimator for $\tau^{(\phi)}_{PAME}$ is then
$$
\hat\tau_N = \sum_{m_k \in \mathcal{M}_1}\phi(1;m_k)\,\hat\mu(1;m_k) \;-\; \sum_{m_k \in \mathcal{M}_0}\phi(0;m_k)\,\hat\mu(0;m_k).
$$
Setting $\phi$ to the realized design shares $\hat\phi(z;m_k) = N_z(m_k)/N_z$ makes the within-size means aggregate back to the arm means, so $\hat\tau_N = \bar Y_1 - \bar Y_0$ recovers the difference in means and targets the design-weighted form. Setting $\phi(z;m_k) = 1/|\mathcal{M}_z|$ instead targets the PAME with even weights across group sizes. The PAME with even weights across groups and  within-size conditional contrasts can be estimated similarly.

In the asymptotic analysis below, we consider arbitrarily weighted estimands and analogue estiamtors, covering the specific examples shown here as special cases. We develop the analysis for arbitrarily weighted estimands and estimators using the group-level representation.  We first present the translation of the design to the group-level representation, and then introduce the arbitrarily weighted estimand and estimator.

\subsubsection{Translation to the group-level representation}\label{sec:scheme}
As we did above for the homogeneous groups case, we can translate the design with varying group sizes into a group-level sampling design.

\begin{description}
    \item[Design Case 2G:] Sample $\mathcal{A}_1$ uniformly from all groups of size $M_1$ that could be formed by units in $\mathcal{U}$; then sample $\mathcal{A}_2$ from all groups of size $M_2$ that could be formed by units in $\mathcal{U} \setminus \mathcal{A}_1$. Continue sequentially: at step $g$, sample $\mathcal{A}_g$ uniformly from all groups of size $M_g$ in $\mathcal{U} \setminus \bigcup_{h<g} \mathcal{A}_h$, where $M_g \in \mathcal{M} = \{m_1, m_2, \ldots, m_K\}$. Among groups with size $m$, the first $G_{N1}(m)$ are assigned to treatment ($z_g = 1$) and the remaining $G_{N0}(m)$ ones are under control ($z_g = 0$). $G_{N1}(m) / G_{N}(m) = p(m)$.
\end{description}

For any realized partition $a=(\mathcal{A}_1,\mathcal{A}_2,...,\mathcal{A}_{G_N})$, we use $\mathbb{P}_I(a)$ and $\mathbb{P}_{G}(a)$ to denote the probability of the realized partition $a$. The following proposition shows that two sampling schemes are equivalent in the sense that they have the same probability distribution over the realized $(\mathcal{A}_1,\mathcal{A}_2,...,\mathcal{A}_{G_N})$.

\begin{lemma}
	For any realized partition $a=(\mathcal{A}_1,\mathcal{A}_2,...,\mathcal{A}_{G_N})$, both designs have the same probability:
	$$
	\mathbb{P}_I(a)=\mathbb{P}_G(a)=\frac{(n-N)!\prod_{g=1}^{G_N}M_g!}{n!}
	$$
\end{lemma}

\begin{proof}
	Fix $a=(\mathcal{A}_1,\mathcal{A}_2,...,\mathcal{A}_{G_N})$. Consider the Individual Design first. Let $S(a)$ be the sample containing individuals in the realized partition $a$: $\cup_{g=1}^{G_N}\mathcal{A}_g$. The probability of this sample $S(a)$ being drawn is 
	$$
	\mathbb{P}_I[S(a)]=\frac{1}{{n \choose N}}.
	$$
	
	Conditional on $S(a)$, the number of labeled partitions of sizes $(M_1,M_2,...,M_{G_N})$ is
	$$
	{N \choose M_1 \ M_2 \ ...\ M_{G_N}}=\frac{N!}{\prod_{g=1}^{G_N}(M_g!)}
	$$ Therefore,  $\mathbb{P}_I[(\mathcal{A}_1,\mathcal{A}_2,...,\mathcal{A}_{G_N})=a|S(a)]=\frac{\prod_{g=1}^{G_N}(M_g!)}{N!}$.
	
	Combining these together, we conclude
	$$
	\mathbb{P}_I(a)=\frac{N!(n-N)!}{n!}\frac{\prod_{g=1}^{G_N}(M_g!)}{N!}=\frac{(n-N)!\prod_{g=1}^{G_N}(M_g!)}{n!}
	$$
	
	Now, consider the Group Design. Given $a=(\mathcal{A}_1,\mathcal{A}_2,...,\mathcal{A}_{G_N})$, it is straightforward to see that the probability of $\mathcal{A}_1$ being drawn from all $m_1$ size groups that could be formed by the units in $\mathcal{U}$ is ${n \choose M_1}^{-1}$. Generally, at step $g$, the probability of drawing $\mathcal{A}_g$ uniformly from all $m_g$-size groups of $U-\cup_{h<g}\mathcal{A}_h$ is equivalent to drawing $\mathcal{A}_g$ from $U-\bigcup_{h<g}\mathcal{A}_h$.
	
	Therefore, we conclude that
	$$
	\mathbb{P}_G(a)=\prod_{g=1}^{G_N} \frac{1}{{n-\sum_{h<g}M_h \choose M_g}}=\prod_{g=1}^{G_N} \frac{M_g!(n-\sum_{h\leq g}M_h)!}{(n-\sum_{h<g}M_h)!}
	$$
	
	Define $s_g=\sum_{h=1}^g M_h$ and $s_0=0$. Note that $s_{G_N}=N$, and
	$$
	\begin{aligned}
		& \prod_{g=1}^{G_N} \frac{M_g!(n-\sum_{h\leq g}M_h)!}{(n-\sum_{h<g}M_h)!} = \prod_{g=1}^{G_N} M_g!\prod_{g=1}^{G_N} \frac{(n-\sum_{h\leq g}M_h)!}{(n-\sum_{h<g}M_h)!}\\
		= & \prod_{g=1}^{G_N} M_g! \prod_{g=1}^{G_N}\frac{(n-s_g)!}{(n-s_{g-1})!} =\prod_{g=1}^{G_N} M_g! \frac{(n-s_{G_N})!}{(n-s_{0})!}=\prod_{g=1}^{G_N} M_g! \frac{(n-N)!}{n!}
	\end{aligned}
	$$
	
	Therefore, we conclude that
	$$
	\mathbb{P}_G(a)=\mathbb{P}_I(a)
	$$
\end{proof}

We similarly define $\pi(m) = \frac{G_N(m)}{|\mathcal{G}_{\mathcal{U}}(m)|} = \E\left[W_{\omega} \mid |\mathcal{A}_{\omega}| = m\right]$ as in the homogeneous case, where $\omega$ indexes the potential groups in the population, and $W_{\omega}$ indicates whether group $\omega$ is included in the sample. Based on this result, we can generalize Lemmas~\ref{lemma:group-sampling} and~\ref{lemma:group-po} to the case with varying group sizes.

We now show how estimands and estimators can be translated to the group level. For each treatment condition $z \in \{0,1\}$ and group size $m_k$, define the population mean of group-level potential outcomes under condition $z$,
$$
\mu(z;m_k) = \frac{\sum_{\omega: |\mathcal{A}_{\omega}| = m_k} \bar Y_{\omega}(z)}{|\mathcal{G}_{\mathcal{U}}(m_k)|}.
$$
By the same exchanging of summations used in the proof of Lemma~\ref{lemma:group-ame}, $\mu(z;m_k)$ coincides with $\frac{1}{n}\sum_{i \in \mathcal{U}}\mu_i(z;m_k)$, so this group-level definition is equivalent to the individual-level definition from the previous subsection. We can write the PAME under weights $\phi$ as
$$
\begin{aligned}
\tau^{(\phi)}_{PAME} = & \sum_{m_k \in \mathcal{M}_1} \big[\phi(1;m_k)\,\mu(1;m_k)\big] - \sum_{m_k \in \mathcal{M}_0}\big[\phi(0;m_k)\,\mu(0;m_k)\big].
\end{aligned}
$$
which justifies its estimator presented above.

We also define the \emph{sample average marginalized effect} (SAME) for the analysis. First define the per-arm-per-size quantity
$$
\check\mu(z;m_k) = \frac{\sum_{\omega: |\mathcal{A}_{\omega}| = m_k} W_{\omega}\,\bar Y_{\omega}(z)}{G_N(m_k)}
= \frac{\sum_{\omega: |\mathcal{A}_{\omega}| = m_k} W_{\omega}\,\bar Y_{\omega}(z)}{|\mathcal{G}_{\mathcal{U}}(m_k)|\,\pi_N(m_k)},
$$
for $(z,m_k)$ with $\phi(z;m_k) \ne 0$. Then, the SAME is defined as
$$
\tau^{(\phi)}_{SAME} = \sum_{m_k \in \mathcal{M}_1}\phi(1;m_k) \check \mu(1;m_k) - \sum_{m_k \in \mathcal{M}_0}\phi(0;m_k) \check \mu(0;m_k).
$$

\subsubsection{Translating analytical results to the varying group size case}
We now generalize the remaining lemmas developed in the case of homogeneous group sizes to the case of varying group sizes. The generalization of Lemmas~\ref{lemma:group-ame}-\ref{lemma:limits} is straightforward. Following the proof of Lemma~\ref{lemma:limits}, we have:
$$
\begin{aligned}
\frac{|\mathcal{G}_{\mathcal{U}}(m_{k'})\setminus \mathcal{A}_{\omega}|}{|\mathcal{G}_{\mathcal{U}}(m_{k'})|} & = \frac{{n-m_{k}\choose m_{k'}}}{{n\choose m_{k'}}} = \prod_{p=0}^{m_{k'}-1} \left(1 - \frac{m_{k}}{n - p}\right) \\
& \simeq \left(1 - \frac{m_{k}}{n}\right)^{m_{k'}} \simeq 1 - \frac{m_{k'}m_{k}}{n} \to 1.
\end{aligned}
$$


Lemmas~\ref{lemma:var-g-1} is an extension of Lemma~\ref{lemma:var-g}. Again, we derive the result for a general random variable $\tilde \tau_{\omega}(m)$ defined for group $\omega$ with size $m$.
\begin{lemma}[Variance from group sampling]\label{lemma:var-g-1}
For any $\tilde \tau_{\omega}(m)$ with $|\tilde \tau_{\omega}(m)| \le C$ defined for group $\omega$ with size $m$, $\check \tau = \sum_{m \in \mathcal{M}}\phi(m)\frac{\sum_{\omega = 1}^{|\mathcal{G}_{\mathcal{U}}(m)|}W_{\omega}\tilde \tau_{\omega}(m)}{|\mathcal{G}_{\mathcal{U}}(m)|\pi(m)} = \sum_{m \in \mathcal{M}}\phi(m)\check \tau(m)$, and  $\tau = \sum_{m \in \mathcal{M}}\phi(m)\frac{\sum_{\omega = 1}^{|\mathcal{G}_{\mathcal{U}}(m)|}\tilde \tau_{\omega}(m)}{|\mathcal{G}_{\mathcal{U}}(m)|} = \sum_{m \in \mathcal{M}}\phi(m) \tau(m)$, where $\phi(m) \ge 0$ and $\sum_{m \in \mathcal{M}} \phi(m) = 1$,
$$
\begin{aligned}
G_N\,\Var[\check \tau] & \simeq \sum_{m \in \mathcal{M}} \frac{G_N \phi(m)^2}{G_N(m)}\,\frac{\sum_{\omega: |\mathcal{A}_\omega| = m}\big(\tilde\tau_\omega(m) - \tau(m)\big)^2}{|\mathcal{G}_\mathcal{U}(m)|},
\end{aligned}
$$
and $\check\tau \overset{P}{\to} \tau$, as $n, N$ grow to infinity under our asymptotic regime. 
\end{lemma}
\begin{proof}
We know that
$$
\Var[\check\tau] = \sum_{m \in \mathcal{M}} \phi(m)^2 \Var[\check\tau(m)] + \sum_{m,m' \in \mathcal{M}, m \ne m'}\phi(m)\phi(m') \Cov[\check\tau(m), \check\tau(m')].
$$

\emph{Within-$m$ variance.} The sampling probabilities $\Var[W_\omega] = \pi_N(m)(1-\pi_N(m))$ and $\Cov[W_\omega, W_{\omega'}]$ depend only on the size class and on whether $\mathcal{A}_\omega \cap \mathcal{A}_{\omega'} = \emptyset$, not on the values $\tilde\tau_\omega(m)$. So the calculation in the proof of Lemma~\ref{lemma:var-g} carries through verbatim with $\tau_\omega$ replaced by $\tilde\tau_\omega(m)$. Expanding $\Var[\check\tau(m)]$ and collecting terms,
$$
\begin{aligned}
G_N\,\Var[\check\tau(m)] = \;& \frac{G_N\,\pi_N(m)(1-\pi_N(m))}{G_N^2(m)}\sum_{\omega:|\mathcal{A}_\omega|=m}\tilde\tau_\omega(m)^2 \\
& + \frac{G_N\,\pi_N(m)\Big(\tfrac{G_N(m)-1}{|\mathcal{G}_{\mathcal{U}\setminus\mathcal{A}_\omega}(m)|} - \pi_N(m)\Big)}{G_N^2(m)} \sum_{\omega:|\mathcal{A}_\omega|=m}\sum_{\substack{\omega':|\mathcal{A}_{\omega'}|=m\\\mathcal{A}_\omega\cap\mathcal{A}_{\omega'}=\emptyset}}\tilde\tau_\omega(m)\tilde\tau_{\omega'}(m) \\
& - \frac{G_N\,\pi_N^2(m)}{G_N^2(m)}\sum_{\omega:|\mathcal{A}_\omega|=m}\sum_{\substack{\omega':|\mathcal{A}_{\omega'}|=m\\\mathcal{A}_\omega\cap\mathcal{A}_{\omega'}\neq\emptyset}}\tilde\tau_\omega(m)\tilde\tau_{\omega'}(m),
\end{aligned}
$$
which simplifies---by the same algebra as in the proof of Lemma~\ref{lemma:var-g}---to
$$
G_N\,\Var[\check\tau(m)] \simeq \frac{G_N}{G_N(m)}\,\frac{\sum_{\omega:|\mathcal{A}_\omega|=m}\big(\tilde\tau_\omega(m) - \tau(m)\big)^2}{|\mathcal{G}_\mathcal{U}(m)|}.
$$

\emph{Cross-$m$ covariance.} For $m \ne m'$,
$$
\begin{aligned}
\Cov[\check\tau(m), \check\tau(m')] = \frac{1}{G_N(m)G_N(m')} & \sum_{\omega:|\mathcal{A}_\omega|=m}\sum_{\omega':|\mathcal{A}_{\omega'}|=m'} \Cov[W_\omega, W_{\omega'}]\,\tilde\tau_\omega(m)\,\tilde\tau_{\omega'}(m').
\end{aligned}
$$
Splitting into disjoint ($\mathcal{A}_\omega \cap \mathcal{A}_{\omega'} = \emptyset$) and overlapping cases, and using $|\tilde\tau_\omega(m)| \le C$ ,
$$
\begin{aligned}
& G_N\,\big|\Cov[\check\tau(m), \check\tau(m')]\big| \leq G_N C^2[(\frac{1}{|\mathcal{G}_{\mathcal{U}\setminus\mathcal{A}_\omega}(M_{m'})|}-\frac{1}{ |\mathcal{G}_\mathcal{U}(M_{m'})|}|\mathcal{G}_{\mathcal{U}\setminus\mathcal{A}_\omega}(M_{m'})|) \\
&+   \frac{|\mathcal{G}_\mathcal{U}(M_{m'})|-|\mathcal{G}_{\mathcal{U}\setminus\mathcal{A}_\omega}(M_{m'})| }{|\mathcal{G}_\mathcal{U}(M_{m'})|}  ]\\
&=2G_NC^2(1-\frac{|\mathcal{G}_{\mathcal{U}\setminus\mathcal{A}_\omega}(M_{m'})|}{|\mathcal{G}_\mathcal{U}(M_{m'})|})\\
&=O(\frac{G_N mm'}{n})=o(1)
\end{aligned}
$$

Combining within-$m$ and cross-$m$ pieces proves the lemma.
\end{proof}

From Lemma~\ref{lemma:var-g-1}, we can generalize Lemma~\ref{lemma:var-zero}. Consider the plug-in estimator for $\tau^{(\phi)}_{PAME}$:
$$
\hat\tau_N = \sum_{m_k \in \mathcal{M}_1}\phi(1;m_k)\,\hat\mu(1;m_k) \;-\; \sum_{m_k \in \mathcal{M}_0}\phi(0;m_k)\,\hat\mu(0;m_k).
$$
Define the within-size-$m_k$ population variances
$$
S_1^2(m_k) = \frac{\sum_{\omega: |\mathcal{A}_\omega|=m_k}\big(\bar Y_\omega(1) - \mu(1;m_k)\big)^2}{|\mathcal{G}_\mathcal{U}(m_k)|}, \quad S_0^2(m_k) = \frac{\sum_{\omega: |\mathcal{A}_\omega|=m_k}\big(\bar Y_\omega(0) - \mu(0;m_k)\big)^2}{|\mathcal{G}_\mathcal{U}(m_k)|}.
$$
The asymptotic variance of $\hat\tau$ is
$$
G_N\Var\left[\hat \tau\right] \simeq \sum_{m_k \in \mathcal{M}_1} \frac{G_N}{G_N(m_k)} \frac{\phi(1, m_k)^2\,S_1^2(m_k)}{p(m_k)} + \sum_{m_k \in \mathcal{M}_0} \frac{G_N}{G_N(m_k)} \frac{\phi(0, m_k)^2\,S_0^2(m_k)}{1-p(m_k)} =: \sigma_{\tau}^2.
$$
The cross-arm covariance between $\hat \mu(1;m_k)$ and $\hat \mu(0;m_k)$ vanishes asymptotically because each sampled group $\omega$ contributes to only one of the two arms, and the cross-group-size covariance terms ($m \ne m'$) are $o_P(1)$ by the same arguments used in the proof of Lemma~\ref{lemma:var-g-1}. The generalization of Lemmas~\ref{lemma:var-sutva} and~\ref{lemma:var-est} then follows.

Lemma~\ref{lemma:normality-sampling-1} is the generalization of Lemmas~\ref{lemma:normality-ind-sampling} and~\ref{lemma:normality-sampling}.
\begin{lemma}[Asymptotic normality for group-level sample averages]\label{lemma:normality-sampling-1}
For any $\tilde \tau_{\omega}(m)$ with $|\tilde \tau_{\omega}(m)| \le C$ defined for group $\omega$ with size $m$, $\check \tau = \sum_{m \in \mathcal{M}}\phi(m)\frac{\sum_{\omega = 1}^{|\mathcal{G}_{\mathcal{U}}(m)|}W_{\omega}\tilde \tau_{\omega}(m)}{|\mathcal{G}_{\mathcal{U}}(m)|\pi(m)} = \sum_{m \in \mathcal{M}}\phi(m)\check \tau(m)$, and  $\tau = \sum_{m \in \mathcal{M}}\phi(m)\frac{\sum_{\omega = 1}^{|\mathcal{G}_{\mathcal{U}}(m)|}\tilde \tau_{\omega}(m)}{|\mathcal{G}_{\mathcal{U}}(m)|} = \sum_{m \in \mathcal{M}}\phi(m) \tau(m)$, where $|\mathcal{M}| = K$, $\phi(m) \ge 0$, and $\sum_{m \in \mathcal{M}} \phi(m) = 1$,
$$
\begin{aligned}
\sqrt{G_N}\left(\check \tau - \tau\right) \rightsquigarrow \mathcal{N}\left(0, G_N\Var\left[\check \tau\right]\right)
\end{aligned}
$$
as $n, N$ grow to infinity under our asymptotic regime.
\end{lemma}
\begin{proof}
Again, we consider the original and the hypothetical experiment, where the outcome is an ordered $G_N$-tuple: $\tilde{\mathcal{A}} = \left(\mathcal{A}_1,\mathcal{A}_2,\dots,\mathcal{A}_{G_N}\right)$. In the hypothetical experiment, the probability of observing any given tuple $\tilde{\mathcal{A}}$ is $q^* = \frac{1}{\prod_{k=1}^K |\mathcal{G}_{\mathcal{U}}(m_k)|^{G_N(m_k)}}$. Under the original sampling algorithm, the probability of observing the same tuple is
$$
q = \frac{1}{\prod_{k=1}^K \prod_{r=0}^{G_N(m_k) - 1}{n - \sum_{k'=1}^{k-1} G_N(m_{k'})m_{k'} -  rm_k \choose m_k}}.
$$
Remember that $Q^*$ and $Q$ denote the distribution of $\tilde{\mathcal{A}}$ in the hypothetical and original experiments, respectively. The following still holds:
$$
supp(Q) \subset supp(Q^*), |supp(Q)| = \prod_{k=1}^K \prod_{r=0}^{G_N(m_k) - 1}{n - \sum_{k'=1}^{k-1} G_N(m_{k'})m_{k'} -  rm_k \choose m_k} = \frac{1}{q}.
$$
The total variation distance between $Q^*$ and $Q$ is
$$
\begin{aligned}
& d_{TV}(Q^*, Q) =  1 - \frac{q^*}{q} = 1 - \frac{\prod_{k=1}^K \prod_{r=0}^{G_N(m_k) - 1}{n - \sum_{k'=1}^{k-1} G_N(m_{k'})m_{k'} -  rm_k \choose m_k}}{\prod_{k=1}^K |\mathcal{G}_{\mathcal{U}}(m_k)|^{G_N(m_k)}} \\
= & 1 - \prod_{k=1}^K\prod_{r=0}^{G_N(m_k) - 1}\frac{{n - \sum_{k'=1}^{k-1} G_N(m_{k'})m_{k'} -  rm_k \choose m_k}}{{n \choose m_k}} \simeq 1 - \prod_{k=1}^K\prod_{r=0}^{G_N(m_k) - 1}\left(1 - \frac{\sum_{k'=1}^k G_N(m_{k'})m_{k'} +  rm_k}{n}\right)^{m_k} \\
\simeq & 1 - \prod_{k=1}^K\prod_{r=0}^{G_N(m_k) - 1}e^{-\frac{\left(\sum_{k'=1}^{k-1} G_N(m_{k'})m_{k'} +  rm_k\right)m_k}{n}} \simeq \frac{\sum_{k=1}^K\sum_{r=0}^{G_N(m_k) - 1}\left(\sum_{k'=1}^{k-1} G_N(m_{k'})m_{k'} +  rm_k\right)m_k}{2n} \\
= & \frac{\sum_{k=1}^KG_N(m_k) m_k\sum_{k'=1}^{k-1} G_N(m_{k'})m_{k'} + \sum_{k=1}^K m_k^2\sum_{r=0}^{G_N(m_k) - 1} r}{2n} \simeq \frac{\left(\sum_{k=1}^K G_N(m_k) m_k\right)^2}{2n} \to 0,
\end{aligned}
$$
It remains to establish the limiting distribution under $Q^*$ and transfer it to $Q$ via the TV bound.

\emph{Limit under $Q^*$.} Under the hypothetical experiment $Q^*$, groups are sampled with replacement from $\mathcal{G}_\mathcal{U}(m_k)$ for each size class, and the samples in different size classes are independent. Within each $k$, $\check\tau(m_k) = \frac{1}{G_N(m_k)}\sum_{g=1}^{G_N(m_k)}\tilde\tau_{g}(m_k)$ is the average of $G_N(m_k)$ i.i.d.\ draws from $\{\tilde\tau_\omega : |\mathcal{A}_\omega| = m_k\}$. The boundedness $|\tilde\tau_\omega(m_k)| \le C$ ensures the Lindeberg condition is trivially satisfied, so by the classical CLT,
$$
\sqrt{G_N(m_k)}\,\big(\check\tau(m_k) - \tau(m_k)\big) \rightsquigarrow \mathcal{N}(0, G_N(m_k)\Var\left[\check\tau(m_k)\right]) \quad \text{under } Q^*.
$$
Independence across $k$ under $Q^*$, together with $\sqrt{G_N}(\check\tau - \tau) = \sum_{k=1}^K \sqrt{G_N/G_N(m_k)}\phi(m_k)\cdot\sqrt{G_N(m_k)}(\check\tau(m_k) - \tau(m_k))$, gives
$$
\sqrt{G_N}\,(\check\tau - \tau) \rightsquigarrow \mathcal{N}\bigg(0, \sum_{k=1}^K \frac{G_N\phi^2(m_k)}{G_N(m_k)}\,G_N(m_k) \Var\left[\check\tau(m_k)\right]\bigg) = \mathcal{N}\left(0, G_N\Var\left[\check \tau\right]\right)
\quad \text{under } Q^*.
$$

\emph{Transfer to $Q$.} Since $d_{TV}(Q^*, Q) \to 0$, for any bounded continuous $f$, $\big|\E_{Q^*}[f(\sqrt{G_N}(\check\tau - \tau))] - \E_Q[f(\sqrt{G_N}(\check\tau - \tau))]\big| \le 2\|f\|_\infty\,d_{TV}(Q^*, Q) \to 0$, so the same weak limit holds under the original sampling distribution $Q$.
\end{proof}

When $\tilde \tau_{\omega}(m) = \tau_{\omega}(m)$,  Lemma~\ref{lemma:normality-sampling-1} implies that 
$$
\sqrt{G_N}\left(\tau^{(\phi)}_{SAME} - \tau^{(\phi)}_{PAME}\right) \rightsquigarrow \mathcal{N}\left(0, \sigma_{\tau}^2\right).
$$
We can similarly derive the asymptotic distribution of $\hat \tau_N$ using the combination of the CLT in \citet{ohlsson1989asymptotic} and the coupling argument from \citet{hajek1960limiting}. So, Theorem 1 in the main text can be proved in the case of varying group sizes.

\subsection{Monte Carlo simulation design}

The first Monte Carlo study, implemented in \texttt{group interaction simulation-new.R}, examines variance estimation and confidence-interval coverage. Its results are reported in Tables~\ref{tab:sim-outcomes}, \ref{tab:sim-sdv}, and \ref{tab:int_test} and Figure~\ref{fig:sim-vratio} of the main text. The second study, implemented in \texttt{group interaction simulation-s-covariate.R}, examines covariate adjustment and is summarized in Figure~\ref{fig:sim-cov}. The two studies use related data-generating processes, with the covariate-adjustment study adding a second individual covariate and comparing several covariate-adjusted estimators.

\subsubsection{Simulation study 1: variance estimation and coverage}
\label{app:sim1}

\subsubsection*{Population and potential outcomes}

For each sample size
\[
N\in\{80,200,300,400\},
\]
we construct a fixed finite population $\mathcal{U}$ with respective population sizes
\[
n\in\{404{,}772,\ 4{,}000{,}000,\ 11{,}022{,}704,\ 22{,}627{,}420\}.
\]
These values follow $n\propto N^{5/2}$, with each population size rounded up to be divisible by the common group size $M=4$. The resulting values of $N^2/n$ are $0.0158$, $0.0100$, $0.0082$, and $0.0071$, respectively. Before sampling, the population is partitioned into $n/M$ pre-existing groups of size $M$.

In each Monte Carlo replication, the sampling procedure depends on the design case. Under the fixed-group design, we sample $G_N=N/M$ pre-existing population groups without replacement and retain all of their members. Under the random-group design, we instead draw $N$ individual units by simple random sampling without replacement and randomly partition them into $G_N=N/M$ experimental groups of size $M$. Thus, pre-existing population groups determine sampling and experimental group membership only under the fixed-group design.

Each unit $i$ carries a scalar covariate drawn independently from the standard uniform distribution:
\[
X_i \sim \mathrm{Uniform}(0,1).
\]
The covariates and population-group random effects are drawn once for each $(N,n)$ setting and held fixed across the $2{,}000$ Monte Carlo replications. Control potential outcomes combine a group-level random effect, which induces population homophily, with an individual covariate term:
\[
Y_i(0)
=
\tfrac{1}{2}\,\zeta_{g(i)}
+
\tfrac{1}{2}\,X_i,
\qquad
\zeta_g \sim \mathcal{N}(0,1),
\]
where $\zeta_{g(i)}$ is shared by all members of population group $g(i)$. Because the shared component $\tfrac{1}{2}\zeta_g$ dominates the within-group-independent component $\tfrac{1}{2}X_i$, the intraclass correlation of $Y(0)$ within the population groups is approximately $0.9$, consistent with Table~\ref{tab:sim-outcomes}.

The treated potential outcome is $Y_i(1) = Y_i(0) + \tau_i$, where the unit-level treatment effect
$\tau_i$ is the way in which effect heterogeneity and interference enter. Under the no-interference
(SUTVA) specification,
\[
\tau_i \;\sim\; \mathcal{N}\!\left(X_i,\,1\right),
\]
so each unit's effect is centered at its own covariate value with idiosyncratic noise but does not depend
on groupmates. Under the interference specification, $\tau_i$ additionally depends on the covariates of
unit $i$'s $M-1 = 3$ groupmates $\mathcal{P}_i$ through a quadratic peer-exposure function:
\[
\tau_i \;=\; \underbrace{\mathcal{N}\!\left(X_i,\,1\right)}_{\text{own effect}}
\;+\; \sqrt{X_i}\,\Bigg(\sum_{j \in \mathcal{P}_i}\beta_{1j}\,X_j
\;+\; \sum_{j \in \mathcal{P}_i}\sum_{\ell \in \mathcal{P}_i}\beta_{2,j\ell}\,X_j X_\ell\Bigg).
\]
The coefficient vectors are themselves drawn once per simulation call as
half-normals,
\[
\beta_1 \in \mathbb{R}^{M-1},\quad \beta_{1j} = \tfrac{1}{2}\,\lvert\mathcal{N}(0,1)\rvert;
\qquad
\beta_2 \in \mathbb{R}^{(M-1)^2},\quad \beta_{2,j\ell} = \tfrac{1}{2}\,\lvert\mathcal{N}(0,1)\rvert.
\]
The multiplicative $\sqrt{X_i}$ makes the peer-exposure contribution interact with the unit's own
covariate. Observed outcomes are $Y_i = Z_i Y_i(1) + (1-Z_i) Y_i(0)$, where $Z_i$ is the realized
individual treatment indicator induced by the design.

\subsubsection*{Monte Carlo benchmark and estimators}

For Monte Carlo replication $b$, let $\mathcal{N}_b$ denote the sampled
units and let $\mathcal{A}_{ib}$ denote the realized experimental group
containing unit $i$. We compute the realized unit-level treatment effect
\[
\Delta_{ib}
=
Y_i(1,\mathcal{A}_{ib})-Y_i(0,\mathcal{A}_{ib})
\]
and record its sample average,
\[
\widetilde{\tau}_b
=
\frac{1}{N}\sum_{i\in\mathcal{N}_b}\Delta_{ib}.
\]
In the code, $\widetilde{\tau}_b$ is stored as
\texttt{ATEs[b] = mean(tau)}. Under interference, this quantity is
conditional on the group partition realized in replication $b$; it is
therefore not itself the population average marginalized effect.

We define the population-level Monte Carlo benchmark as
\[
\tau_{\mathrm{MC}}
=
\E_{\mathrm{MC}}\!\left[\widetilde{\tau}_b\right],
\]
where the expectation is taken over the sampling and group-formation
processes and, under the implemented data-generating process, over the
idiosyncratic treatment-effect disturbances. We approximate this
expectation using
\[
\widehat{\tau}_{\mathrm{MC}}
=
\frac{1}{B}\sum_{b=1}^{B}\widetilde{\tau}_b,
\]
with $B=2{,}000$. This quantity is stored in the code as
\texttt{popu\_ate = mean(result[["ATEs"]])} and is used as the benchmark
for evaluating bias and confidence-interval coverage.

The interpretation of $\tau_{\mathrm{MC}}$ depends on the design and
outcome specification. It corresponds to the population average
treatment effect under no interference, to the population average total
effect under interference with fixed groups, and to the population
average marginalized effect under interference with randomly formed
groups. Thus, although the same Monte Carlo calculation is used in all
four cases, the resulting benchmark represents a different causal
estimand across cases.

The treatment-effect estimator in each replication is the unweighted
difference in means,
\[
\widehat{\tau}_b
=
\overline{Y}_{1b}-\overline{Y}_{0b},
\]
computed as the coefficient on treatment in the OLS regression of
$Y_i$ on $Z_i$. Because all experimental groups have the common size
$M=4$ and treatment is assigned to equal numbers of groups, this
estimator coincides with the homogeneous-group version of the estimator
defined in the main text. Confidence-interval coverage is evaluated
against $\widehat{\tau}_{\mathrm{MC}}$.

For inference we compute three variance estimators on the same fitted regression:
\begin{itemize}
\item \textbf{HR2 (heteroskedasticity-robust).} The unit-level HC2 estimator, obtained from
\texttt{sandwich::vcovHC(lm\_fit, type = "HC2")}.
\item \textbf{CR2 (cluster-robust).} Group-clustered sandwich estimators obtained from
\texttt{clubSandwich::vcovCR(lm\_fit, cluster = group\_ind, type = "CR2")},
clustering on the realized experimental group. CR2 is the bias-reduced cluster-robust estimator.
\end{itemize}
Ninety-five percent confidence intervals based on HR2 and CR2 use the normal critical value $1.96$. In
addition, we form a small-sample CR2 interval using a Satterthwaite degrees-of-freedom adjustment via
\texttt{clubSandwich::conf\_int(lm\_fit, vcov = vcovCR(\ldots, "CR2"), test = "Satterthwaite")}. Coverage is evaluated against the population estimand
(the across-replication mean of $\tau$). As a robustness check the script also computes a group-aggregated
analysis---regressing group means $\bar Y_g$ on the group treatment indicator and applying HC2 at the
group level---which reproduces the cluster-level variance.

The reported quantities are the Monte Carlo mean of $\widehat{\tau}$, its true across-replication variance, the mean of each variance estimator, the empirical coverage of the corresponding confidence intervals, and the standard deviation of each variance estimator as a measure of stability. The intraclass correlations of $Y(0)$ and $Y(1)$ are computed using \texttt{Hmisc::deff()}, taking the design-effect routine's $\rho$ output for the realized experimental groups.

For the randomly formed-group scenarios, the script also implements the ICC-based interference test described in the main text. Within each treatment-by-group-size cell, it computes the one-way ANOVA statistic
\[
F_{z,m}
=
\frac{\operatorname{MSB}_{z,m}}
{\operatorname{MSW}_{z,m}}
\]
and its one-sided upper-tail $p$-value under the corresponding $F$ distribution. The cell-specific $p$-values are combined using Fisher's statistic,
\[
S_{\mathrm{ICC}}
=
-2\sum_{j=1}^{J}\log(p_j),
\]
which is compared with a $\chi^2_{2J}$ reference distribution. Under the homogeneous group-size design, there are two informative cells, corresponding to the treatment and control arms, so $J=2$. The rejection rate across Monte Carlo replications is reported in Table~\ref{tab:int_test}.

\subsubsection{Simulation study 2: covariate adjustment}
\label{app:sim2}

The covariate-adjustment study considers the random-group design with interference. We construct a finite population of $n=400{,}000$ units, initially partitioned into groups of common size $M=4$. In each of $2{,}000$ Monte Carlo replications, we draw $N=400$ individual units by simple random sampling without replacement and randomly partition them into $G_N=N/M=100$ experimental groups of size four. We then assign exactly half of the experimental groups to treatment. The pre-existing population groups are not used in sampling or experimental group formation, although their shared random effects remain components of the units' control potential outcomes. The population covariates, pre-existing-group random effects, and peer-effect coefficients defined below are generated once at the beginning of the Monte Carlo study and then treated as fixed. Across replications, we redraw only the individual-level sample, the random partition into experimental groups, and the group-level treatment assignment.

\subsubsection*{Population and potential outcomes}

Each population unit has two independently distributed covariates,
\[
X_{1i} \overset{\mathrm{iid}}{\sim} \mathrm{Uniform}(0,1),
\qquad
X_{2i} \overset{\mathrm{iid}}{\sim} \mathrm{Uniform}(0,1).
\]
The control potential outcome is
\[
Y_i(0)
=
\frac{1}{2}\zeta_{g(i)}
+
\frac{1}{2}X_{1i}
+
X_{2i}^{2}
+
X_{1i}X_{2i},
\qquad
\zeta_g \overset{\mathrm{iid}}{\sim}\mathcal{N}(0,1),
\]
where $\zeta_{g(i)}$ is shared by the members of unit $i$'s pre-existing population group.

To define the treatment effect, index unit $i$'s three experimental groupmates by
$j_{i1},j_{i2},j_{i3}$. The code draws
\[
\beta_q=\frac{1}{2}|B_q|,
\qquad
\gamma_{q\ell}=\frac{1}{2}|C_{q\ell}|,
\qquad
B_q,C_{q\ell}\overset{\mathrm{iid}}{\sim}\mathcal{N}(0,1),
\]
for $q,\ell\in\{1,2,3\}$. The same length-three vector
$\boldsymbol{\beta}$ is used for the linear terms in both covariates, and the same
$3\times3$ array $\boldsymbol{\gamma}$ is used for all three quadratic and cross-covariate terms. The unit-level treatment effect is
\[
\begin{aligned}
	\tau_i
	={}&
	X_{1i}+X_{2i}^{2} \\
	&+
	3\sqrt{X_{1i}}
	\Bigg[
	\sum_{q=1}^{3}
	\beta_q
	\left(
	X_{1,j_{iq}}+X_{2,j_{iq}}
	\right) \\
	&\hspace{2.7cm}
	+
	\sum_{q=1}^{3}\sum_{\ell=1}^{3}
	\gamma_{q\ell}
	\Big(
	X_{1,j_{iq}}X_{1,j_{i\ell}}
	+
	X_{2,j_{iq}}X_{2,j_{i\ell}}
	+
	X_{1,j_{iq}}X_{2,j_{i\ell}}
	\Big)
	\Bigg].
\end{aligned}
\]
Thus,
\[
Y_i(1)=Y_i(0)+\tau_i,
\qquad
Y_i=Z_iY_i(1)+(1-Z_i)Y_i(0).
\]

\subsubsection*{Covariate-adjusted estimators}

For each unit, we construct the leave-one-out peer means
\[
\bar X_{1,-i}
=
\frac{1}{3}\sum_{q=1}^{3}X_{1,j_{iq}},
\qquad
\bar X_{2,-i}
=
\frac{1}{3}\sum_{q=1}^{3}X_{2,j_{iq}}.
\]
We compare seven estimators:

\begin{enumerate}
	\item \textbf{No covariates:}
	\[
	Y_i \sim 1+Z_i.
	\]
	
	\item \textbf{Individual covariates:}
	\[
	Y_i \sim 1+Z_i+X_{1i}+X_{2i}.
	\]
	
	\item \textbf{Peer-mean covariates:}
	\[
	Y_i \sim 1+Z_i+\bar X_{1,-i}+\bar X_{2,-i}.
	\]
	
	\item \textbf{Individual and peer-mean covariates:}
	\[
	Y_i
	\sim
	1+Z_i+X_{1i}+X_{2i}
	+\bar X_{1,-i}+\bar X_{2,-i}.
	\]
	
	\item[5--7.] \textbf{Lin adjustment:} For each of the three covariate sets in specifications 2--4, we also apply Lin's centered, treatment-interacted regression adjustment \citep{lin2013agnostic}, implemented using \texttt{estimatr::lm\_lin}.
\end{enumerate}

For specifications 1--4, the point estimate is the coefficient on $Z_i$ from the corresponding OLS regression. For specifications 5--7, it is the treatment coefficient returned by \texttt{lm\_lin}. CR2 variance estimates are computed for all adjusted estimators, clustering on the realized experimental group.

The target displayed in the figure is
\[
\tau_{\mathrm{MC}}
=
\frac{1}{B}
\sum_{b=1}^{B}
\left(
\frac{1}{N}\sum_{i=1}^{N}\tau_{ib}
\right),
\qquad B=2{,}000,
\]
which is the Monte Carlo approximation to the PAME under random group formation. The figure plots the sampling distribution of each estimator and reports its empirical bias and root mean squared error relative to $\tau_{\mathrm{MC}}$. For estimator $k$, these quantities are
\[
\operatorname{Bias}\!\left(\widehat{\tau}^{(k)}\right)
=
\frac{1}{B}\sum_{b=1}^{B}\widehat{\tau}^{(k)}_b
-
\tau_{\mathrm{MC}},
\]
and
\[
\operatorname{RMSE}\!\left(\widehat{\tau}^{(k)}\right)
=
\left[
\frac{1}{B}\sum_{b=1}^{B}
\left(
\widehat{\tau}^{(k)}_b-\tau_{\mathrm{MC}}
\right)^2
\right]^{1/2}.
\]
Thus, the displayed RMSE is calculated from the empirical Monte Carlo distribution for every estimator rather than from its reported CR2 variance estimate.

\end{appendix}

\end{singlespace}

\clearpage
\bibliography{group-interaction}

@article{hajek1960limiting,
  title={Limiting distributions in simple random sampling from a finite population},
  author={H{\'a}jek, Jaroslav},
  journal={A Magyar Tudom{\'a}nyos Akad{\'e}mia Matematikai Kutat{\'o} Int{\'e}zet{\'e}nek k{\"o}zlemenyei},
  volume={5},
  number={3},
  pages={361--374},
  year={1960},
  publisher={Akad{\'e}miai Kiad{\'o}}
}

@article{ohlsson1989asymptotic,
  title={Asymptotic normality for two-stage sampling from a finite population},
  author={Ohlsson, Esbj{\"o}rn},
  journal={Probability theory and related fields},
  volume={81},
  number={3},
  pages={341--352},
  year={1989},
  publisher={Springer}
}

@book{polyanskiy2025information,
  title={Information theory: From coding to learning},
  author={Polyanskiy, Yury and Wu, Yihong},
  year={2025},
  publisher={Cambridge university press}
}

@article{sobel2006randomized,
  title={What do randomized studies of housing mobility demonstrate? Causal inference in the face of interference},
  author={Sobel, Michael E},
  journal={Journal of the American Statistical Association},
  volume={101},
  number={476},
  pages={1398--1407},
  year={2006},
  publisher={Taylor \& Francis}
}

@article{pustejovsky2018small,
  title={Small-sample methods for cluster-robust variance estimation and hypothesis testing in fixed effects models},
  author={Pustejovsky, James E and Tipton, Elizabeth},
  journal={Journal of Business \& Economic Statistics},
  volume={36},
  number={4},
  pages={672--683},
  year={2018},
  publisher={Taylor \& Francis}
}

@article{cox1958planning,
  title={Planning of experiments.},
  author={Cox, David Roxbee},
  year={1958},
  publisher={Wiley}
}

@article{iacovone2022improving,
  title={Improving management with individual and group-based consulting: Results from a randomized experiment in Colombia},
  author={Iacovone, Leonardo and Maloney, William and McKenzie, David},
  journal={The Review of Economic Studies},
  volume={89},
  number={1},
  pages={346--371},
  year={2022},
  publisher={Oxford University Press}
}

@article{imbens2016robust,
  title={Robust standard errors in small samples: Some practical advice},
  author={Imbens, Guido W and Kolesar, Michal},
  journal={Review of Economics and Statistics},
  volume={98},
  number={4},
  pages={701--712},
  year={2016},
  publisher={The MIT Press}
}

@article{pals2008individually,
  title={Individually randomized group treatment trials: A critical appraisal of frequently used design and analytic approaches (American Journal of Public Health (2008) 98 (1418-1424},
  author={Pals, SL and Murray, DM and Alfano, CM and Shadish, WR and Hannan, PJ and Baker, WL},
  journal={American Journal of Public Health},
  volume={98},
  number={12},
  year={2008},
  publisher={American Public Health Association}
}

@article{candel2025efficient,
  title={Efficient design of cluster randomized trials and individually randomized group treatment trials.},
  author={Candel, Math JJM and van Breukelen, Gerard JP},
  journal={Psychological Methods},
  year={2025},
  publisher={American Psychological Association}
}

@article{su2021model,
  title={Model-assisted analyses of cluster-randomized experiments},
  author={Su, Fangzhou and Ding, Peng},
  journal={Journal of the Royal Statistical Society Series B: Statistical Methodology},
  volume={83},
  number={5},
  pages={994--1015},
  year={2021},
  publisher={Oxford University Press}
}

@article{lohr2014partially,
  title={Partially Nested Randomized Controlled Trials in Education Research: A Guide to Design and Analysis. NCER 2014-2000.},
  author={Lohr, Sharon and Schochet, Peter Z and Sanders, Elizabeth},
  journal={National Center for Education Research},
  year={2014},
  publisher={ERIC}
}

@book{imbens2015causal,
  title={Causal inference in statistics, social, and biomedical sciences},
  author={Imbens, Guido W and Rubin, Donald B},
  year={2015},
  publisher={Cambridge university press}
}

@article{lin2013agnostic,
  title={Agnostic notes on regression adjustments to experimental data: Reexamining Freedman’s critique},
  author={Lin, Winston},
  journal={The Annals of Applied Statistics},
  pages={295--318},
  year={2013},
  volume={7},
  number={1}
}

@article{aronow2017estimating,
  title={Estimating Average Causal Effects Under General Interference, With Application To a Social Network Experiment},
  author={Aronow, Peter M and Samii, Cyrus},
  journal={The Annals of Applied Statistics},
  pages={1912--1947},
  year={2017},
  volume={11},
  number={4},
  publisher={JSTOR}
}

@article{li2019randomization,
  title={Randomization Inference for Peer Effects},
  author={Li, Xinran and Ding, Peng and Lin, Qian and Yang, Dawei and Liu, Jun S},
  journal={Journal of the American Statistical Association},
  volume={114},
  number={528},
  pages={1651--1664},
  year={2019},
  publisher={Taylor and Francis Ltd}
}

@article{bai2024primer,
  title={A Primer on the Analysis of Randomized Experiments and a Survey of some Recent Advances},
  author={Bai, Yuehao and Shaikh, Azeem M and Tabord-Meehan, Max},
  journal={arXiv preprint arXiv:2405.03910},
  year={2024}
}

@article{splawa1990application,
  title={On the application of probability theory to agricultural experiments. Essay on principles. Section 9.},
  author={Splawa-Neyman, Jerzy and Dabrowska, Dorota M and Speed, Terrence P},
  journal={Statistical Science},
  pages={465--472},
  year={1990},
  publisher={JSTOR}
}

@article{bai2022optimality,
  title={Optimality of matched-pair designs in randomized controlled trials},
  author={Bai, Yuehao},
  journal={American Economic Review},
  volume={112},
  number={12},
  pages={3911--3940},
  year={2022},
  publisher={American Economic Association 2014 Broadway, Suite 305, Nashville, TN 37203}
}

@article{hu2022average,
  title={Average direct and indirect causal effects under interference},
  author={Hu, Yuchen and Li, Shuangning and Wager, Stefan},
  journal={Biometrika},
  volume={109},
  number={4},
  pages={1165--1172},
  year={2022},
  publisher={Oxford University Press}
}

@article{savje2021average,
  title={Average treatment effects in the presence of unknown interference},
  author={S{\"a}vje, Fredrik and Aronow, Peter and Hudgens, Michael},
  journal={Annals of statistics},
  volume={49},
  number={2},
  pages={673},
  year={2021},
  publisher={NIH Public Access}
}

@article{aronow2021spillover,
  title={Spillover effects in experimental data},
  author={Aronow, Peter M and Eckles, Dean and Samii, Cyrus and Zonszein, Stephanie},
  journal={Advances in experimental political science},
  volume={289},
  pages={319},
  year={2021},
  publisher={Cambridge University Press Cambridge}
}

@article{abadie2023should,
  title={When should you adjust standard errors for clustering?},
  author={Abadie, Alberto and Athey, Susan and Imbens, Guido W and Wooldridge, Jeffrey M},
  journal={The Quarterly Journal of Economics},
  volume={138},
  number={1},
  pages={1--35},
  year={2023},
  publisher={Oxford University Press}
}

@article{hudgens2008toward,
  title={Toward causal inference with interference},
  author={Hudgens, Michael G and Halloran, M Elizabeth},
  journal={Journal of the American Statistical Association},
  volume={103},
  number={482},
  pages={832--842},
  year={2008},
  publisher={Taylor \& Francis}
}

@article{samii2012equivalencies,
  title={On equivalencies between design-based and regression-based variance estimators for randomized experiments},
  author={Samii, Cyrus and Aronow, Peter M},
  journal={Statistics \& Probability Letters},
  volume={82},
  number={2},
  pages={365--370},
  year={2012},
  publisher={Elsevier}
}

@article{mendelberg2014does,
  title={Does descriptive representation facilitate women's distinctive voice? How gender composition and decision rules affect deliberation},
  author={Mendelberg, Tali and Karpowitz, Christopher F and Goedert, Nicholas},
  journal={American journal of political science},
  volume={58},
  number={2},
  pages={291--306},
  year={2014},
  publisher={Wiley Online Library}
}

\end{document}